\documentclass[10pt,journal,twocolumn,letterpaper]{IEEEtran}

\usepackage[T1]{fontenc}
\usepackage[utf8]{inputenc}
\usepackage{amsmath,amssymb,amsfonts,amsthm}
\usepackage{graphicx}
\usepackage{booktabs}
\usepackage{array}
\usepackage{enumitem}
\usepackage{url}
\usepackage{cite}
\usepackage{tikz}
\usetikzlibrary{positioning,arrows.meta,shapes.geometric}
\usepackage[unicode]{hyperref}
\hypersetup{
  colorlinks=true,
  linkcolor=black,
  urlcolor=blue!55!black,
  citecolor=black,
  pdftitle={Algebraic Diversity: Principles of a Group-Theoretic Approach to Signal Processing (v7 draft 2)},
  pdfauthor={M.A. Thornton},
}

\DeclareMathOperator{\Tr}{Tr}
\DeclareMathOperator{\Aut}{Aut}

\DeclareMathOperator{\diag}{diag}
\DeclareMathOperator{\Var}{Var}

\newtheorem{theorem}{Theorem}
\newtheorem{proposition}[theorem]{Proposition}
\newtheorem{lemma}[theorem]{Lemma}
\newtheorem{corollary}[theorem]{Corollary}
\newtheorem{definition}[theorem]{Definition}
\newtheorem{remark}[theorem]{Remark}
\newtheorem{conjecture}[theorem]{Conjecture}

\graphicspath{{figures/}{./}}

\begin{document}

\title{Algebraic Diversity: Principles of a\\Group-Theoretic Approach to Signal Processing}

\author{Mitchell~A.~Thornton
\thanks{E-mail: \texttt{mitch@smu.edu}.}\\
Southern Methodist University\\
May 16, 2026}

\maketitle

\begin{abstract}
We present principles of a group-theoretic approach to signal processing in which the algebraic symmetry of a signal is exploited for variance reduction in second-order estimation. The classes of rearrangements that leave a signal's statistics invariant form a mathematical group; this matched group determines the natural transform through which the signal should be analyzed; and estimators constructed as averages over the group action extract more information from each observation than do conventional methods that implicitly assume only the trivial group. The viewpoint unifies spectral estimation, single-observation measurement, blind equalization, and the dictionary construction step of compressed sensing within a common algebraic frame, and casts the law of large numbers as the trivial-group degenerate case of a $(G, L)$ continuum that combines classical sample-count averaging with algebraic group-orbit averaging. Several techniques are developed in detail: a rank-promotion procedure that admits the group-averaged estimator on scalar data streams; an eigentensor hierarchy for signals with nested symmetry; a blind group-matching methodology that identifies the matched group from data via a polynomial-time eigenvalue problem on the Lie algebra $\mathfrak{u}(M)$; and a continuous-relaxation treatment of transform selection that places the DFT, DCT, and Karhunen--Lo\`{e}ve transforms as distinguished points on a single transform manifold. The Constant Modulus Algorithm's residual phase ambiguity is predicted analytically and matched within $1.6^\circ$ on 3GPP TDL multipath channels; an orbit-dictionary construction is proposed conjecturally for compressed sensing. Four theorems formalize a structural capacity $\kappa$ as the operational measure of single-observation estimation efficiency, complemented by a finite-dimensional conjugate capacity bound for non-commuting Hermitian generators. The structural capacity is the R\'{e}nyi-2 analog of Shannon and von Neumann's R\'{e}nyi-1 entropies and quantifies how a signal's information is organized rather than how much it carries. We hope this work serves as an early step in a broader application of group theory to signal processing, and that the concept of algebraic diversity will be regarded as a contribution alongside existing algebraic approaches to the field such as compressed sensing and algebraic signal processing.
\end{abstract}


\section{Introduction}
\label{sec:intro}

\IEEEPARstart{T}{his} paper is a broader and more general consideration of the concept of algebraic diversity that was introduced in the first arXiv paper on the topic~\cite{thornton2026ad}, where the concept was developed from the narrower viewpoint of reducing the number of samples required to compute a covariance estimate and the interrelationship this gives rise to among the Karhunen--Lo\`{e}ve (KL), discrete Fourier (DFT), and discrete cosine (DCT) transforms. Before the broader development is undertaken, it is useful to review some basic concepts from the mathematical subjects of group theory and representation theory as they apply to discrete signal processing.

\subsection{Background}
\label{sec:background}

\subsubsection{Group Theory}
\label{sec:group_theory}

Because the beneficial effects of AD depend upon the algebraic symmetry of a signal, a brief review of basic concepts is in order. A signal has algebraic symmetry when there is some way to rearrange its components that leaves its statistics unchanged. As an example, a time-varying sinusoid of constant amplitude is invariant under time shifts and can be said to have temporal symmetry since delaying it by one period leaves it indistinguishable from its original form. This same signal, a constant-amplitude sinusoid, has left-right reflective symmetry since it can be time-reversed, or ``flipped'' around the time axis, and still have the same form. Likewise, a multi-dimensional signal that is sensed by $M$ sensors simultaneously can have algebraic symmetry with respect to its spatially-separated sensor array. Since the ordering or indexing of the sensors is arbitrary and their labels do not affect the signal, any cyclic shift of the labels, or even any permutation of the sensor labels, will not affect the signal amplitudes that are sensed; it will only interchange the component values of the $M$-dimensional vector (or tensor of rank one) that represents the measured sensor values at a particular time instant. This indicates a cyclic or permutation symmetry.

Understanding the algebraic symmetry present in a signal is an important basic concept since it describes the classes of signals that benefit the most from algebraic diversity principles. These examples of temporally-varying signals, and in the case of the sensor array, spatially-varying signals, are basic concepts, and many techniques in signal processing can be classified based on temporal interactions or spatial interactions with the signals of interest; that is, they can be viewed as exploiting the temporal diversity or the spatial diversity of the sensor array. Here we consider a third type of diversity that is present due to the algebraic symmetry of a signal, and we accordingly call this class of methods \textit{algebraic diversity} (AD).

The symmetry present in a signal is represented mathematically through the use of group theory. Group theory is a topic in abstract algebra that is concerned with a structure that is formally referred to as a \textit{group}. Because the symmetrical diversity present in a given signal varies according to its structure, different groups are used to represent signals based on their structure. A group that captures exactly the symmetries present in a given signal, neither missing any that are present nor including any that are not, is called a \textit{matched group}, and the formal statement is as follows.

\begin{definition}[Matched group]\label{def:matched}
Let $\mathbf{x}$ be a signal whose statistics are invariant under a collection of component rearrangements that includes the identity rearrangement. The \textit{matched group} of $\mathbf{x}$, denoted $G^*$, is the unique group (containing the identity element $e$) such that (i) the statistics of $\mathbf{x}$ are invariant under the rearrangement induced by every element of $G^*$, and (ii) no rearrangement outside $G^*$ leaves those statistics invariant. The qualifier \textit{matched} reflects that $G^*$ is the most efficient group in terms of representing the symmetry of the signal: smaller groups miss symmetries that are present, and larger groups include symmetries that are not.
\end{definition}

Moreover, we will focus mainly on sampled signals in this work, so we are mainly concerned with discrete groups; however, continuous groups are also available to be used for signals that are modeled as continuous mathematical functions. A discrete algebraic group is a structure comprised of two objects: a set of elements and a binary operator that acts on any pair of elements from the set. The textbook example of a group is the set of all real values and the binary multiplication operator. These two objects form a group because any two values that are multiplied result in another real value or member of the group, yielding the \textit{closure} property. Another required property of a group is the presence of an \textit{identity element} with respect to the operator. Since the set of real values includes the unity value, or ``one,'' this property is satisfied. Yet another requirement is the presence of \textit{inverses} with respect to the operator, meaning that the inverse of any group element is also a member of the group. The multiplicative inverse of any real value $x$ is $x^{-1}$, which happens to be a real value and a member of the set, so this property is also present. The final and fourth property is \textit{associativity}, which means that the order of application of the operator does not matter. In this example consider $x, y, z \in \mathbb{R}$; since $(x \times y) \times z = x \times (y \times z)$, the fourth property is satisfied and the pair of objects $\mathbb{R}$ and $\times$ form a group. The operator is not required to be multiplicative, but since multiplication is the typical example used in defining a group $G$, the group operator is often referred to as the \textit{product operator}, and in fact one may see a group defined simply as $\mathbb{R}$, in which case the multiplicative operator is assumed. However, it should be noted that any binary operation may be used and, in fact, the set of elements need not even be numerical values.

As we are interested in the algebraic symmetry of a sampled signal, the type of group we are interested in here is one whose elements are a finite set of $M$ signal samples and whose operator is not multiplicative but rather one that preserves the symmetric property of interest. These are operators that, in general, interchange the elements of the set. Since a mathematical set has no fixed ordering, a bijection on the set can be viewed as a rearrangement in which each element may be mapped to any other element so long as the map is one-to-one and onto. The collection of all such possible rearrangements are the bijections on the set. The most general group of this form represents all possible symmetries of the set and is called the \textit{symmetric group}, defined as follows.

\begin{definition}[Symmetric group]\label{def:symmetric}
Let $\{1, 2, \ldots, M\}$ be a finite set of cardinality $M$. The \textit{symmetric group} $S_M$ is the group whose elements are all bijections of the set onto itself, with group operator composition of bijections. Its cardinality is $|S_M| = M!$, and it represents every symmetry that the $M$-element set can possess. When $S_M$ acts on a signal $\mathbf{x} \in \mathbb{C}^M$, it does so by permuting the components of $\mathbf{x}$ according to the corresponding bijection of component indices.
\end{definition}

Because a particular signal may exhibit one form of symmetry but not another, $S_M$ can be used to represent the signal's symmetry, but it also contains symmetries that are not present in the signal, and thus it is not the signal's matched group in general. At the opposite extreme of richness, every signal is compatible with the smallest possible group, called the \textit{trivial group}, defined next.

\begin{definition}[Trivial group]\label{def:trivial}
The \textit{trivial group}, denoted $\{e\}$, is the group comprising the single element $e$ whose action leaves every object unchanged. That is, for any object $\mathbf{x}$ on which $\{e\}$ can act, $e \cdot \mathbf{x} = \mathbf{x}$, and $|\{e\}| = 1$. When realized on a single scalar signal sample, the action of $e$ is multiplication by unity; when realized as a subgroup of the symmetric group $S_M$ on an $M$-vector signal, $e$ is the identity permutation leaving every component fixed.
\end{definition}

\begin{remark}\label{rem:trivial_matched}
Every signal is compatible with the trivial group, but $\{e\}$ is the matched group for a given signal if and only if the signal possesses no algebraic symmetry beyond the identity. For signals with additional symmetry, $\{e\}$ is a proper subgroup of the matched group.
\end{remark}

The symmetric group and the trivial group bracket a family of groups that arise in practice. Every permutation group arising in the applications of this paper is the symmetry group of the same $M$-element set endowed with some particular structure, and is therefore a subgroup of $S_M$. Imposing a cyclic (directed) order on the $M$ elements gives the cyclic group $\mathbb{Z}_M$; adding a reflection gives the dihedral group $D_M$; imposing a graph structure on the elements gives the graph automorphism group $\Aut(\mathcal{G})$; and so on. In practice, the matched group for a given signal is almost always a proper subgroup of $S_M$, because measurement systems impose specific structural constraints (cyclic lattices, symmetric boundaries, graph topologies) under which only a subset of all possible bijections leaves the signal statistics invariant. A useful distinction among these groups is whether their operator is commutative: a group is \textit{Abelian} if $g \cdot h = h \cdot g$ for every pair of elements and \textit{non-Abelian} otherwise. The cyclic group $\mathbb{Z}_M$ and direct products of cyclic groups are Abelian; the dihedral group $D_M$ (for $M \geq 3$), the symmetric group $S_M$ (for $M \geq 3$), and most nontrivial graph automorphism groups are non-Abelian. This distinction will have computational consequences in the developments that follow. The map from abstract group to signal-processing transform is supplied by representation theory, which we take up next: each group acts on $\mathbb{C}^M$ by a unitary representation, and the characters of that representation's irreducible components form the columns of the matched transform.

\subsubsection{Representation Theory}
\label{sec:rep_theory}

For a discrete group $G$ to act on a vector of $M$ signal samples $\mathbf{x} \in \mathbb{C}^M$, a concrete mechanism must be specified by which each abstract group element $g \in G$ is associated with a definite operation on the vector. This association is called a \textit{representation} of the group, and is written $\pi: G \to \mathrm{GL}(M, \mathbb{C})$, where $\mathrm{GL}(M, \mathbb{C})$ is the set of invertible $M \times M$ complex matrices. In this notation, $\pi$ is the representation itself, understood as the rule that assigns a matrix to each group element, and $\pi_g$ or $\pi(g)$ denotes the particular matrix that the representation assigns to the element $g$. The key property that makes $\pi$ a representation, rather than merely an assignment of matrices to group elements, is that it preserves the group's structure: for any two elements $g, h \in G$, the matrix associated with their product satisfies $\pi(gh) = \pi(g)\pi(h)$, so that the group's binary operator is realized as ordinary matrix multiplication. In particular, the identity element $e$ of the group is always sent to the $M \times M$ identity matrix, $\pi(e) = \mathbf{I}_M$, and the inverse of a group element is always sent to the inverse of the associated matrix, $\pi(g^{-1}) = \pi(g)^{-1}$. In this way, the abstract group is realized concretely as a collection of matrices, and the action of a group element on a signal vector is simply the matrix-vector product $\pi_g \mathbf{x}$.

The representations of interest for signal processing applications are those that preserve energy, meaning that the norm of the signal is unchanged by the group action: $\|\pi_g \mathbf{x}\| = \|\mathbf{x}\|$ for every $g \in G$ and every $\mathbf{x} \in \mathbb{C}^M$. Matrices with this property are called \textit{unitary} and satisfy $\pi_g^H \pi_g = \mathbf{I}_M$, where $\pi_g^H$ denotes the conjugate transpose. A representation whose matrices are all unitary is called a \textit{unitary representation}, and it is a foundational result that every representation of a finite group is equivalent to a unitary representation. This is fortunate, because the energy-preserving property is precisely what allows AD to combine multiple group-transformed views of a signal without any artificial rescaling: each view $\pi_g \mathbf{x}$ is statistically equivalent in energy to the original signal $\mathbf{x}$, differing only in how that energy is distributed across the $M$ components of the vector. The group-averaged estimator introduced below, in equation~\eqref{eq:gae}, takes advantage of precisely this property.

A unitary representation of a group on $\mathbb{C}^M$ may or may not be the simplest possible representation of that group; it is often the case that the $M$-dimensional space can be decomposed into a direct sum of smaller subspaces, each of which is invariant under the group action, and each of which cannot be decomposed further. When no such nontrivial decomposition is possible, the representation is called \textit{irreducible}, and the subspace on which it acts is the smallest one on which the group action ``fits'' without leaving any invariant sub-subspace untouched. An irreducible representation is the basic building block of representation theory in the same way that a prime number is the basic building block of integer arithmetic: every unitary representation of a finite group can be decomposed uniquely (up to isomorphism) into a direct sum of irreducible representations, a result known as the Peter--Weyl decomposition.

For an Abelian group such as $\mathbb{Z}_M$, every irreducible representation is one-dimensional, and the irreducibles correspond exactly to the columns of the discrete Fourier transform matrix; this is why the DFT diagonalizes any covariance matrix that commutes with $\mathbb{Z}_M$, and why the matched group for cyclically-sampled signals produces the DFT as its natural transform. For a non-Abelian group, some of the irreducible representations have dimension $d_i > 1$, and the associated matched transform is correspondingly block-diagonal rather than fully diagonal, with each block of size $d_i \times d_i$ corresponding to one irreducible representation. The signal processing benefit of identifying and using a non-Abelian matched group comes directly from this block structure, as will become evident in the development that follows.

With the notion of a unitary representation in hand, the informal ``invariance of statistics'' discussed in the definition of the matched group takes a precise form. For a signal $\mathbf{x}$ with covariance $\mathbf{R} = \mathbb{E}[\mathbf{x} \mathbf{x}^H]$, the statistics of $\mathbf{x}$ are invariant under the action of $G$ through the representation $\pi$ if and only if the covariance matrix commutes with every $\pi_g$, written $[\pi_g, \mathbf{R}] \equiv \pi_g \mathbf{R} - \mathbf{R} \pi_g = 0$ for every $g \in G$, or equivalently $\pi_g \mathbf{R} \pi_g^H = \mathbf{R}$. This commutation relation is the formal expression of the symmetry property and will be used throughout the paper.

\subsubsection{Subgroups and the Containment Lattice}
\label{sec:subgroups}

Throughout this paper, groups are frequently compared by containment: one group is said to be \textit{inside} another when every element of the smaller group is also an element of the larger, and the group operator of the smaller group is inherited unchanged from the larger. Formally, for two groups $H$ and $G$, we say that $H$ is a \textit{subgroup} of $G$, written $H \subset G$, if $H$ is a subset of $G$ that is itself closed under the group operator of $G$ and contains the identity element of $G$. When $H \subset G$ and $H \neq G$, we say $H$ is a \textit{proper subgroup} of $G$.

Subgroup containment organizes the universe of matched-group candidates into a lattice with the trivial group $\{e\}$ at the bottom and the symmetric group $S_M$ at the top. Between these extremes lie the cyclic group $\mathbb{Z}_M$, the dihedral group $D_M$, and the graph automorphism group $\Aut(\mathcal{G})$ for various graph structures, each related to the others by inclusion: $\{e\} \subset \mathbb{Z}_M \subset D_M \subset S_M$, for example. The matched group of a given signal occupies one specific position on this lattice, determined by the structure of the measurement system, as discussed in Section~\ref{sec:apparent_rarity}. Many results in the developments below take the form of comparisons between a candidate group and one of its subgroups, and the relative efficiency of a larger matched group over a smaller one scales with the ratio of their cardinalities. Figure~\ref{fig:hasse} displays this lattice organization as a Hasse diagram.

\begin{figure}[t]
\centering
\resizebox{\columnwidth}{!}{%
\begin{tikzpicture}[
  node distance=10mm and 6mm,
  every node/.style={font=\scriptsize},
  group/.style={draw, rounded corners=2pt, inner sep=3pt, minimum width=10mm, align=center},
  card/.style={font=\scriptsize, gray},
  edge/.style={-, thick}
]
\node[group] (SM)  at (0, 4.0) {$S_M$};
\node[card]  at (1.5, 4.0) {$|S_M|=M!$};

\node[group] (AM)  at (-1.0, 2.8) {$A_M$};
\node[card]  at (-2.5, 2.8) {$|A_M|=M!/2$};

\node[group] (DM)  at (1.0, 2.0) {$D_M$};
\node[card]  at (2.4, 2.0) {$|D_M|=2M$};

\node[group] (ZM)  at (1.0, 1.0) {$\mathbb{Z}_M$};
\node[card]  at (1.95, 0.6) {$|\mathbb{Z}_M|=M$};

\node[group] (Z2k) at (-1.0, 1.0) {$\mathbb{Z}_2^k$};
\node[card]  at (-2.5, 1.0) {$|\mathbb{Z}_2^k|=2^k$};

\node[group] (e)   at (0, -0.2) {$\{e\}$};
\node[card]  at (-1.4, -0.2) {$|\{e\}|=1$};

\draw[edge] (SM) -- (AM);
\draw[edge] (SM) -- (DM);
\draw[edge] (DM) -- (ZM);
\draw[edge] (ZM) -- (e);
\draw[edge] (Z2k) -- (e);
\draw[edge] (AM) -- (e);

\node[group, dashed] (autG) at (4.4, 2.0) {$\Aut(\mathcal{G})$};
\node[card, align=left] at (4.4, 3.0)
  {position depends\\on graph $\mathcal{G}$;\\$|\Aut(\mathcal{G})|$ from\\$1$ to $M!$};

\draw[dashed, gray, ->] (autG) to[bend left=15] (SM);
\draw[dashed, gray, ->] (autG) to[bend left=30] (e.east);

\end{tikzpicture}%
}
\caption{Hasse diagram of the matched-group lattice for algebraic diversity. Each node is annotated with its cardinality as a function of $M$; the combinatorial growth of $|S_M| = M!$ relative to its subgroups indicates the rapidly expanding space of matched-group choices as $M$ increases. Lines between nodes denote subgroup containment; the absence of a line means the two groups are either unrelated by containment or occupy parallel branches of the lattice (for example, $D_M \not\subset A_M$ for $M \geq 3$, and $\mathbb{Z}_M$ and $\mathbb{Z}_2^k$ are non-isomorphic for $k \geq 2$ even when $M = 2^k$). The graph automorphism group $\Aut(\mathcal{G})$ occupies a level determined by the graph $\mathcal{G}$, ranging from the trivial group (rigid graph) to $S_M$ (complete graph); its variable position is indicated separately. The diagram is schematic: the full subgroup lattice of $S_M$ contains many more intermediate subgroups than are shown, including additional cyclic $\mathbb{Z}_d$ for $d \mid M$, additional dihedral subgroups, and sporadic subgroups arising from specific graph geometries.}
\label{fig:hasse}
\end{figure}

\subsubsection{Group Action and Orbits}
\label{sec:action_orbits}

The representation $\pi$ defined above specifies how an abstract group element $g \in G$ operates on the signal vector $\mathbf{x} \in \mathbb{C}^M$: the rule is the matrix-vector product $\pi_g \mathbf{x}$, which produces another vector in $\mathbb{C}^M$. The collection of vectors that arise from applying every group element to a fixed $\mathbf{x}$ is the central object on which algebraic diversity operates, and is defined as follows.

\begin{definition}[Group orbit]\label{def:orbit}
Let $G$ be a group with unitary representation $\pi: G \to U(M)$ acting on $\mathbb{C}^M$, and let $\mathbf{x} \in \mathbb{C}^M$ be a signal vector. The \textit{group orbit} of $\mathbf{x}$ under the action of $G$, denoted $\mathcal{O}_G(\mathbf{x})$, is the set of vectors
\begin{equation}
\mathcal{O}_G(\mathbf{x}) \;=\; \{\pi_g \mathbf{x} : g \in G\}.
\end{equation}
The orbit has at most $|G|$ distinct elements, with equality when no nontrivial group element fixes $\mathbf{x}$ (a \textit{free action}).
\end{definition}

The orbit is the concrete object on which the group-averaged estimator of algebraic diversity operates: each orbit element $\pi_g \mathbf{x}$ is an equivalent view of the original signal (equivalent in the sense that its statistics match those of $\mathbf{x}$ whenever $G$ is the matched group), and averaging some statistic over the orbit produces an estimator with reduced variance compared to the same statistic evaluated on $\mathbf{x}$ alone.

The size of the orbit depends on both the group $G$ and the signal $\mathbf{x}$. For a generic signal and a finite group $G$, the orbit has exactly $|G|$ distinct elements, so the orbit size equals the group cardinality. For a signal with extra structure (for example, an all-ones vector under $S_M$), many group elements can map $\mathbf{x}$ to itself, and the orbit may be much smaller than $|G|$. The number of algebraically distinct orbit elements, for a given statistic $f$, is the \textit{effective group order} $d_{\mathrm{eff}}(G, f)$, which will appear explicitly in the variance expression for the group-averaged estimator in Section~\ref{sec:gl_manifold}.

\subsubsection{Characters and the Matched Transform}
\label{sec:characters}

The connection from an abstract group to a concrete signal-processing transform passes through a distinguished scalar quantity attached to each representation. For a representation $\pi: G \to \mathrm{GL}(M, \mathbb{C})$, the \textit{character} of $\pi$ is the function $\chi_\pi: G \to \mathbb{C}$ defined by
\begin{equation}
\chi_\pi(g) \;=\; \mathrm{tr}(\pi_g),
\end{equation}
the trace of the matrix assigned to $g$. The character assigns a single complex number to each group element, and it is invariant under conjugation in the group because trace is cyclic, so $\chi_\pi$ depends only on the conjugacy class of its argument. In this sense the character is a ``scalar fingerprint'' of the representation, and for finite groups different inequivalent irreducible representations have different characters.

The characters of the irreducible representations of a finite group satisfy an orthogonality relation that is central to the connection with matched transforms. If $\chi_i$ and $\chi_j$ are the characters of two inequivalent irreducible representations of $G$, then
\begin{equation}
\label{eq:char_orth}
\frac{1}{|G|} \sum_{g \in G} \overline{\chi_i(g)} \,\chi_j(g) \;=\; \delta_{ij},
\end{equation}
where $\delta_{ij}$ is the Kronecker delta. The characters, viewed as vectors indexed by group elements, therefore form an orthonormal system under a group-averaged inner product. This orthogonality is the algebraic statement that underwrites the diagonalization of the covariance by the matched transform.

For an Abelian group every irreducible representation is one-dimensional, so $\pi_g$ is a single complex number and the character coincides with $\pi_g$ itself. Arranging the characters of the distinct irreducibles as columns of an $|G| \times |G|$ matrix, indexed by irreducible representation and by group element, produces exactly the matched transform matrix. For the cyclic group $\mathbb{Z}_M$ on $\mathbb{C}^M$, the characters are $\chi_k(n) = e^{2\pi i k n / M}$, which are exactly the columns of the discrete Fourier transform, and equation~\eqref{eq:char_orth} becomes the familiar orthonormality of the DFT basis. The dihedral group is non-Abelian, and two cases must be kept apart~\cite{thornton2026unification}. The group $D_M$ acting on the $M$ samples as the rotations and reflections of the $M$-cycle has the symmetric circulant matrices as its commutant, and its matched transform is the real Fourier (Hartley-type) basis, in which the cosine and the sine of each frequency span one two-dimensional irreducible. The DCT-II is instead the matched transform of $D_{2M} = \mathbb{Z}_{2M} \rtimes \mathbb{Z}_2$ acting on the even-reflected extension of length $2M$, restricted to the reflection-symmetric subspace that the $M$ samples occupy. For a non-Abelian group, the irreducibles of dimension $d_i > 1$ produce characters that label blocks rather than individual columns, and the resulting matched transform is block-diagonal with blocks of sizes $d_i \times d_i$, as described in the preceding subsubsection.

The character paragraph closes the chain of concepts that this Background builds up: \textit{matched group} $\to$ \textit{unitary representation} $\to$ \textit{irreducible components} $\to$ \textit{characters} $\to$ \textit{matched transform columns}. With this chain in hand, the remainder of the paper can proceed to use each of its links without further pause for definitions.

\begin{remark}[Commutation criterion for matchedness]\label{rem:commutation}
The commutation relation $[\pi_g, \mathbf{R}] = \mathbf{0}$ established at the end of Section~\ref{sec:rep_theory} admits two consequences that are worth stating explicitly. First, each $\pi_g$ is unitary and $\mathbf{R}$ is Hermitian, so all of these matrices are normal. For an Abelian group the $\pi_g$ commute among themselves as well as with $\mathbf{R}$, and a commuting family of normal matrices shares a common eigenbasis, so for a multiplicity-free representation the matched transform of $G$ simultaneously diagonalizes $\mathbf{R}$. For a non-Abelian group, or when an irreducible is repeated, Schur's lemma gives the corresponding block-diagonal statement. In this sense the Karhunen-Lo\`{e}ve transform of a structured signal coincides with the transform matched to its symmetry group. Equivalently, the Cayley-graph spectrum of $G$ and the spectrum of $\mathbf{R}$ share a common eigenbasis~\cite{thornton2026ad}. Second, because commutation need only be checked on a generating set, a candidate group is tested for matchedness by the simple algebraic check of whether $[\pi_g, \mathbf{R}] = \mathbf{0}$ for each $g$ in a generating set of $G$ (a single generator when $G$ is cyclic); this commutator criterion, relaxed to a soft residual $\delta(G, \mathbf{R})$, underlies the blind group-matching methodology of Section~\ref{sec:blindmatching}.
\end{remark}

\begin{remark}[Statistical symmetry]\label{rem:statistical_symmetry}
The commutation condition $[\pi_g, \mathbf{R}] = \mathbf{0}$ of Remark~\ref{rem:commutation} admits a statistical rereading that connects the algebraic and statistical pictures of AD. Because $\pi_g$ is unitary, commutation of $\pi_g$ with $\mathbf{R}$ is equivalent to $\pi_g \mathbf{R} \pi_g^H = \mathbf{R}$: the second-order statistics of $\mathbf{x}$ are invariant under the action of $G$. We refer to this invariance as the \textit{statistical symmetry} of the signal under the group action. Statistical symmetry unifies several classical signal-processing concepts under a single umbrella: exchangeability is the statistical symmetry of a signal under the full symmetric group $S_M$; wide-sense stationarity is the statistical symmetry of a length-$M$ signal under the cyclic group $\mathbb{Z}_M$ at the second-order level; and rotational invariance of a random field is statistical symmetry under the relevant rotation group. The AD framework may be read as the exploitation of statistical symmetry for variance reduction in second-order estimation, and the matched group of Definition~\ref{def:matched} is equivalently characterized as the largest group under which the signal possesses statistical symmetry.
\end{remark}

\subsection{The $(G, L)$ Continuum}
\label{sec:gl_manifold}

The group-averaged estimator of algebraic diversity~\cite{thornton2026ad} takes the form
\begin{equation}\label{eq:gae}
\hat{\mathbf{R}}_G(\mathbf{x}) \;=\; \frac{1}{|G|}\sum_{g \in G} \bigl(\pi_g \mathbf{x}\bigr)\bigl(\pi_g \mathbf{x}\bigr)^H,
\end{equation}
where $\mathbf{x} \in \mathbb{C}^M$ is a single observation, $G$ is a finite group, and $\pi: G \to U(M)$ is a unitary representation. When the signal covariance $\mathbf{R} = \mathbb{E}[\mathbf{x}\mathbf{x}^H]$ commutes with every $\pi_g$, the estimator is unbiased, and two distinct quantities describe what the group buys. The first is the \emph{effective group order} $d_{\mathrm{eff}}(G, f)$, the number of algebraically distinct values the statistic $f$ takes on the group orbit for generic $\mathbf{x}$: an algebraic degrees-of-freedom count that fixes the number of free coordinates the orbit average populates and, with it, the effective sample size available to any functional of the estimate. For the outer-product statistic $f(\mathbf{x}) = \mathbf{x}\mathbf{x}^H$ with a matched group, $d_{\mathrm{eff}}$ coincides with $|G|$ for the regular representation and with $\sum_i d_i^2$ across irreducible representations of dimensions $\{d_i\}$ under the Peter-Weyl decomposition~\cite{thornton2026ad}. For other statistics, $d_{\mathrm{eff}}$ may be smaller than $|G|$ when the statistic's symmetry collapses the orbit (for example, any $S_M$-invariant scalar has $d_{\mathrm{eff}} = 1$ regardless of $|G|$). The second quantity is the \emph{realized} per-entry variance-reduction factor, which depends on how the energy of the averaged data distributes over those coordinates: it equals $d_{\mathrm{eff}}$ exactly when that distribution is flat (a constant-modulus deterministic signal in white noise, a Gaussian signal whose matched-basis spectrum is flat, or noise-only covariance estimation), and for concentrated profiles it equals a participation-ratio functional bounded above by $d_{\mathrm{eff}}$; for a Gaussian signal with matched-basis spectrum $\{\lambda_m\}$ it is exactly $(\sum_m \lambda_m)^2/\sum_m \lambda_m^2$, the spectral form of the structural capacity of Section~\ref{sec:kappa_two}. Theorem~\ref{thm:outer} records both roles. In the $L$-observation extension,
\begin{equation}\label{eq:gae_L}
\hat{\mathbf{R}}_{G,L} \;=\; \frac{1}{L|G|}\sum_{\ell=1}^L \sum_{g \in G} \bigl(\pi_g \mathbf{x}_\ell\bigr)\bigl(\pi_g \mathbf{x}_\ell\bigr)^H,
\end{equation}
each entry's variance scales as $1/(L \cdot g)$, where $g$ is the realized factor above, reaching its algebraic ceiling $1/(L \cdot d_{\mathrm{eff}})$ in the flat case. The product $L \cdot d_{\mathrm{eff}}$ is the \emph{effective sample size}, and the classical law of large numbers corresponds to the degenerate choice $G = \{e\}$, $d_{\mathrm{eff}} = 1$.

This structure admits a natural reading as a \emph{two-axis tradeoff}. A system designer with a budget of $N$ effective samples may spend it at any $(G, L)$ point on the hyperbola $L \cdot d_{\mathrm{eff}} = N$. At $(G = \{e\}, L = N)$ one performs conventional multi-snapshot averaging. At $(L = 1, |G| = N)$ one performs pure single-snapshot algebraic diversity. Intermediate points trade physical repetition for algebraic structure. The choice among points on the hyperbola is driven by practical constraints: availability of independent snapshots, knowledge of the signal's algebraic symmetry, computational cost of the group average, and whether the signal structure is stationary across the required observation window. Figure~\ref{fig:GL} illustrates the tradeoff for $M = 4$.

The variance reduction afforded by the group-averaged estimator can be combined with classical covariance regularization. A companion development~\cite{thornton2026shrinkage} casts the group average as the structured target of a convex shrinkage estimator in the style of Ledoit and Wolf~\cite{ledoit2004}, in which the sample covariance is shrunk toward the Reynolds-projected estimator $\hat{\mathbf{R}}_G$ rather than toward the scaled identity; when the signal supports the group, the symmetry-aware shrinkage target improves on the isotropic Ledoit-Wolf target in high-dimensional regimes. The present paper develops the group-averaged estimator and its blind identification; the shrinkage extension is treated separately.

\begin{figure}[t]
\centering
\includegraphics[width=\columnwidth]{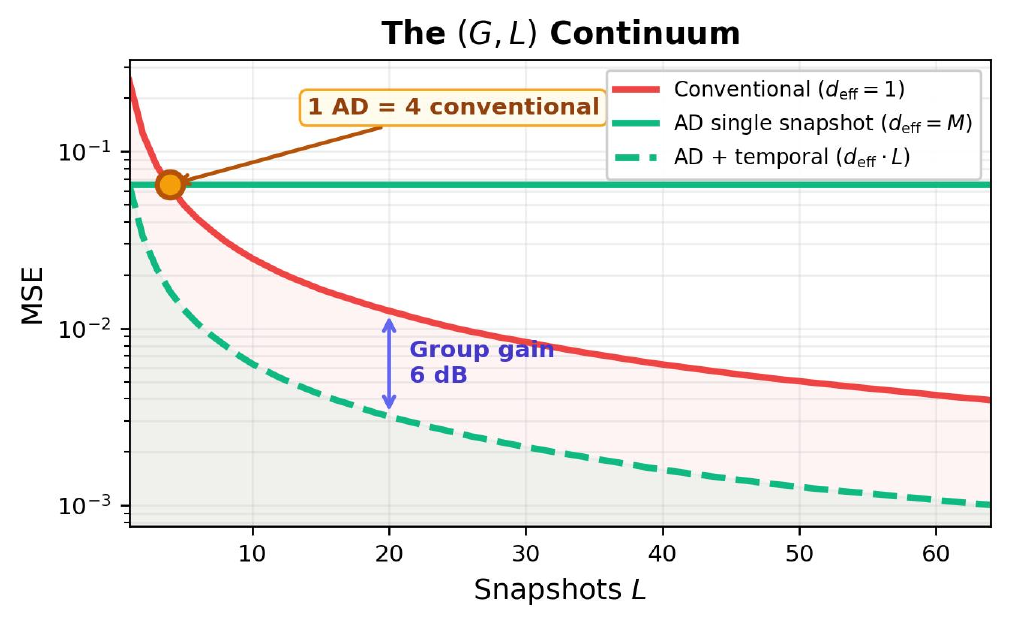}
\caption{The $(G, L)$ continuum, illustrated for $M = 4$ in the flat-profile case where the realized gain attains its ceiling $d_{\mathrm{eff}}$. The conventional sample-covariance estimator (red, $d_{\mathrm{eff}} = 1$) requires $L$ independent snapshots and achieves MSE proportional to $1/L$. The single-snapshot AD estimator (green solid, $d_{\mathrm{eff}} = M$) achieves equivalent MSE from one observation when the matched group acts. Combining AD with temporal averaging (green dashed, $d_{\mathrm{eff}} \cdot L$) shifts the curve down by the group gain, $10\log_{10}(M)$~dB at all snapshot counts in this case; for concentrated energy profiles the per-entry gain is the participation-ratio factor of Theorem~\ref{thm:outer}. The orange marker indicates that one AD measurement reaches the same MSE as four conventional snapshots; the blue arrow indicates the persistent group gain.}
\label{fig:GL}
\end{figure}

\subsection{The transform manifold}
\label{sec:manifold}

The previously published AD paper~\cite{thornton2026ad} established the correspondence between groups and transforms, and showed that the Karhunen--Lo\`{e}ve (KL) transform is optimal in a universal sense because it is associated, through a Cayley graph construction, with the symmetric group $S_M$, which contains all the other permutation groups as subgroups. When the signals of interest are matched to one of these subgroups, the transform associated with that subgroup is the KL-optimal transform for those signals. As an example, signals measured with sensor arrays that are cyclically symmetric, such as uniform linear or circular arrays with uniform sampling period, naturally map to the cyclic group $\mathbb{Z}_M$, whose irreducible representations are the Fourier basis functions; the group-averaged estimator under $\mathbb{Z}_M$ is the periodogram.

The fact that single-snapshot direction-of-arrival estimation is possible for these geometries is already known in the classical array-processing literature. Shan, Wax, and Kailath~\cite{shan1985} established spatial smoothing on uniform linear arrays by exploiting the Toeplitz structure of the array manifold, and Pillai and Kwon~\cite{pillai1989} subsequently obtained the forward/backward spatial smoothing variant via an associated Hankel construction; Tewfik and Hong~\cite{tewfik1992} extended the shift-invariance argument to uniform circular arrays. Although the original justifications for these results appeal to different matrix structures, from the AD viewpoint both the Toeplitz and Hankel constructions are side effects of the same underlying fact: the array admits the cyclic group $\mathbb{Z}_M$ as a matched group, so a single observation vector already carries $M$ algebraically distinct views of the signal through the cyclic shift action, and any second-order estimator that exploits those views is an instance of the group-averaged estimator under $\mathbb{Z}_M$.

A further extension of the spatial-smoothing viewpoint to arbitrary array geometries is due to Friedlander and Weiss~\cite{friedlander1992}, who proposed an interpolation construction that maps the response of a non-uniform array onto that of a virtual uniform linear array within which the cyclic-group machinery of the ULA case applies. The interpolation is performed sector-by-sector through a least-squares fit $\mathbf{T}\,\mathbf{a}(\theta) \approx \mathbf{a}_{\mathrm{ULA}}(\theta)$ on a discrete grid of look angles within each sector. From the AD viewpoint the interpolation transformation $\mathbf{T}$ is an approximate equivariant embedding: a sector-by-sector intertwiner between the non-uniform array's natural representation and the regular representation of $\mathbb{Z}_M$ on the virtual ULA. The construction is not unitary in general; the interpolation transformation changes the noise covariance and the energy normalization, so downstream estimation must account for this nonlinear change to the noise spectrum, and the clean variance-reduction guarantees of the matched-group theorems apply only on the virtual ULA up to the interpolation error.

Likewise, signals with symmetric boundary conditions, such as images or sequences of equally sampled image frames, are matched, through their even-reflected extension of length $2M$, to the dihedral group $D_{2M}$, whose matched transform on that extension is the DCT~\cite{thornton2026unification}.

This family of correspondences suggests viewing all transforms as points on a manifold on which matching a signal's symmetry group maps to a KL-optimal transform. We refer to this as the \emph{transform manifold}. A unified development of classical transforms from this group-theoretic standpoint is given in a companion arXiv paper~\cite{thornton2026unification}. A convenient way to classify points on this manifold is to consider the tensorial form of the signal data representation. An equi-spaced array of $M$ sensors sampled with constant period for $L$ snapshots maps to a tensor of rank two with dimensions $M$ and $L$, admitting the cyclic group and the corresponding DFT. When boundary conditions are dihedral or symmetric (for example, even-symmetric images), the dihedral group $D_{2M}$ acting on the even-reflected extension appears as a different point on the manifold, with the DCT as the associated transform.

\subsection{The tensor-group correspondence and the apparent rarity of non-cyclic matched groups}
\label{sec:apparent_rarity}

A sharper statement of the above is the \emph{tensor-group correspondence}: the matched group is determined by the lattice geometry of the data tensor, which is determined in turn by the physical construction of the measurement system. Concretely:
\begin{itemize}[nosep]
\item A single periodic index (uniform sampling, ULA, circular array) has symmetry group $\mathbb{Z}_M$ (cyclic).
\item A symmetric index (palindromic boundary, even-symmetric image) has symmetry group $D_{2M}$ (dihedral), acting on the even-reflected extension of the index~\cite{thornton2026unification}.
\item Two independent periodic indices (rectangular sensor array) have symmetry group $\mathbb{Z}_{M_1} \times \mathbb{Z}_{M_2}$ (product).
\item An unstructured index set has symmetry group $S_M$ (symmetric), which is computationally intractable for moderate $M$ except via the Cayley-graph construction of~\cite{thornton2026ad}.
\end{itemize}

Extensive experimentation across diverse signal classes (tonal, chirp, AR, MIMO, attention heads in transformer models) has produced a consistent empirical observation: the cyclic group $\mathbb{Z}_M$ is the matched group in the overwhelming majority of cases. Signals for which a non-cyclic matched group produces strictly better estimation than $\mathbb{Z}_M$ are sufficiently rare in practice that their apparent rarity demands explanation.

The explanation is that cyclic groups dominate because \emph{humans build measurement systems on regular lattices with uniform sampling}. A ULA samples space uniformly; an ADC samples time uniformly; a pixel grid samples an image uniformly. Uniform sampling on a one-dimensional lattice has cyclic symmetry, and the DFT (the spectral decomposition of $\mathbb{Z}_M$) diagonalizes the covariance of any signal whose statistics are translation-invariant on that lattice. The DFT is not a mathematical accident of signal theory; it is the inevitable consequence of engineering design choices, and its ubiquity in classical signal processing is a reflection of the engineering culture that built the sensors on which those signals are recorded. Signals measured on irregular lattices, graph-structured networks, or nonuniform sampling grids have different matched groups, and the tensor-group correspondence predicts which groups those are.

\emph{Non-cyclic matched groups are therefore not rare because non-cyclic groups are mathematically unusual}. They are rare because non-cyclic measurement geometries are engineered to be rare. Graph-structured data (social networks, molecular structures, neural connectivity, sensor networks with nontrivial topology) produces non-cyclic matched groups naturally, and this is the setting in which the AD framework's broader range of group structures becomes operationally important.

Many practical signal acquisition scenarios nevertheless fall within the cyclic or dihedral classes, which accounts for the historical predominance (and correctness) of the DFT and DCT in classical signal processing. The AD viewpoint says: the DFT and DCT are correct choices not because they are canonical transforms but because they are the matched transforms for the sampling geometries that engineers have historically preferred.

\subsection{Two foundational questions}
\label{sec:twoQ}

These observations point to two foundational questions for the AD viewpoint:
\begin{enumerate}[label=(\roman*)]
\item \emph{How do we choose the matched group?} When prior knowledge of the data acquisition geometry is available, the matched group is known (as in the array examples above). When no such knowledge is available, the group must be selected from the data; we call this the \emph{blind group matching problem}. Section~\ref{sec:blindmatching} presents the current methodology.
\item \emph{How can data that is scalar (a single time series) benefit from AD?} A scalar stream $\{x[n]\}$ does not immediately admit a nontrivial group action on its samples, because the samples lack the ambient vector space structure that AD requires. The answer is \emph{rank promotion}, the reorganization of a scalar stream into a higher-rank tensor to which a nontrivial group can act. This question is addressed first, in Section~\ref{sec:rankpromote}.
\end{enumerate}

A standing observation, repeated where helpful, is that AD has always been present: every estimator that averages over $L$ samples implicitly uses the trivial group $G = \{e\}$ acting on a single snapshot, $L$ times. The law of large numbers is the trivial-group case of a broader algebraic averaging theorem, and the AD benefit is realized whenever the trivial group is not the matched group.

\subsection{Organization}
\label{sec:organization}

The principle of algebraic diversity (AD) as applied to second-order estimators and single-snapshot measurement was developed in a companion work~\cite{thornton2026ad}. This paper broadens the AD viewpoint beyond measurement to address rank promotion of scalar data streams, blind group matching, blind equalization, and a family of other blind signal processing problems. The framework is further extended in a set of companion papers that fall outside the present scope: a continuous formulation on Lie groups acting on $L^2(\mathbb{R})$~\cite{thornton2026continuous}, a polynomial-time treatment of the blind group-selection problem via the double-commutator eigenvalue problem~\cite{thornton2026gevp}, a unified development of classical signal transforms from the group-theoretic standpoint~\cite{thornton2026unification}, a symmetry-aware convex shrinkage estimator for high-dimensional covariance~\cite{thornton2026shrinkage}, and a quantum extension in which a group-structured measurement set recovers a density matrix from a single copy~\cite{thornton2026quantum}.

We develop the $(G, L)$ continuum that relates the classical sample-count axis $L$ of the law of large numbers to the group-order axis $|G|$ of algebraic averaging, and we articulate the \emph{tensor-group correspondence}: the matched group is determined by the lattice geometry of the measurement system, which explains the empirical dominance of the cyclic group across classical signal processing applications. We summarize and extend the rank-promotion principle that stratifies scalar data into tensorial form admitting nontrivial group action, state the \emph{structural coding rate conjecture} $n^* \approx \lceil 2^{H_{\mathrm{struct}}} \rceil$ as the source-coding analog for structured estimation, and extend the eigentensor hierarchy and algebraic waveform diversity.

We formulate the blind group matching problem and present a current methodology that moves from an early spectral-concentration criterion (which we show carries an orbit-size bias) to a cross-validation criterion $D_{CV}$, thence to a polynomial-time continuous relaxation via a double-commutator generalized eigenvalue problem on the finite-dimensional Lie algebra $\mathfrak{u}(M)$. Supporting closed-form results include an eigenvalue-difference expression for the permutation commutativity residual and a sequential GEVP with group-theoretic deflation whose accepted permutations are guaranteed to lie in $\Aut(\mathbf{R})$; full recovery of $\Aut(\mathbf{R})$ holds trivially when $\Aut(\mathbf{R}) = S_M$, and for a proper subgroup it is recovered by a span-search form of the procedure whenever the candidate basis spans a generating set, leaving the a~priori basis-richness question as the residual open part, as we describe in Section~\ref{sec:seqgevp}. The methodology is organized in two stages, library discovery followed by a two-tier selection, and we report partial resolutions of three cases previously left open: the proper-subgroup recovery just noted, signals carrying several symmetries at once (addressed by an iterative-stripping analog of the CLEAN algorithm), and the degenerate double-commutator spectrum (addressed by a null-space tiebreaker). Empirical support for the methodology is supplied by a variance-scaling dichotomy distinguishing matched groups (convergent power-law variance $\sigma \propto \mathrm{SNR}^{-\beta}$ with $\beta$ strictly positive) from mismatched groups (constant variance, $\beta \approx 0$). The continuous relaxation operates on what we characterize as the \emph{transform manifold} on which the DFT, DCT, and KL transforms are distinguished points.

Using blind group matching as a bridge, we extend AD from measurement to blind signal processing generally; blind equalization provides the first detailed example, with the Constant Modulus Algorithm's residual phase standard deviation $45^\circ/\sqrt{3} \approx 26^\circ$ predicted analytically and matched within $1.6^\circ$ on 3GPP TR~38.901 TDL multipath channels. We state four theorems of an \emph{information structure theory} that parallel Shannon's foundational results and that together make the structural capacity $\kappa$ the operational measure of single-observation estimation efficiency, complemented by a finite-dimensional conjugate capacity bound $\kappa_A \cdot \kappa_B \leq 4/c^2$ for non-commuting Hermitian generators on $\mathbb{C}^M$. An expanded Erlangen-program reading situates AD among Shannon's classical information content, the AD structural capacity, and von Neumann's entropy of quantum measurement. We close with the relationship of AD to four streams of prior work, namely classical invariant estimation, minimax and convex-optimization approaches to robust estimation, algebraic signal processing, and compressed sensing, noting where AD can contribute to or compose with those frameworks without diminishing them.

Sections~\ref{sec:rankpromote}--\ref{sec:eigentensor} take up the rank-promotion question; Section~\ref{sec:blindmatching} develops the blind group matching methodology; Section~\ref{sec:blindproblems} uses that machinery to extend AD to blind signal processing generally, with blind equalization as the lead detailed example and a briefly indicated set of further blind problems to which the framework applies, with detailed treatments left to forthcoming companion papers; Section~\ref{sec:erlangen} presents the expanded Erlangen-program interpretation; Section~\ref{sec:related} discusses the relationship of AD to four streams of prior work (classical invariant estimation, minimax and convex-optimization approaches to robust estimation, algebraic signal processing, and compressed sensing); and Section~\ref{sec:conclusion} closes.

\section{Rank Promotion: Admitting AD on Scalar Data}
\label{sec:rankpromote}

The Background and the $(G, L)$ Continuum subsection both proceed on the assumption that the measurement is a vector observation $\mathbf{x} \in \mathbb{C}^M$, which provides the $M$-dimensional state space on which a matched group can act. A significant fraction of signal processing practice, however, deals with signals that are scalar at the point of measurement: a temperature sensor reading, a blood pressure measurement, a single-microphone recording, or any stream of samples $\{x[n]\}$ taken one at a time. The present section takes up the question of how such scalar data can be brought into the AD framework, introduces the resolution we call \emph{rank promotion}, and establishes the formal connection between the classical law of large numbers and the general algebraic averaging principle that AD extends.

\subsection{The scalar-data obstruction and its resolution}
\label{sec:scalar}

The obstruction that motivates this section is a direct consequence of the formulation given in the previous section. The group-averaged estimator of equation~\eqref{eq:gae} was defined for a vector-valued observation $\mathbf{x} \in \mathbb{C}^M$ and a unitary representation $\pi: G \to U(M)$ that acts on that vector by matrix multiplication. A purely scalar observation $x[n] \in \mathbb{C}$ is, in this sense, a vector of dimension one, and the only unitary representation available on a one-dimensional space is the trivial representation $\pi_g = 1$ of the trivial group $\{e\}$. The group-averaged estimator collapses to $\hat{\theta} = f(x[n])$, which is no estimator at all: a single scalar evaluation of the statistic. No algebraic benefit is available in this state of affairs, and the only way to reduce variance is to collect additional independent samples and average them in the classical fashion. The classical treatment of such scalar data is the sample mean, $\hat{\mu}_N = N^{-1}\sum_{n=1}^N x[n]$, whose variance scales as $\sigma^2/N$ by the law of large numbers. This is the situation in which every conventional scalar signal processing estimator operates, and it is the target of the rank-promotion construction.

The resolution of the obstruction is the \emph{stratification principle}, which we now describe. Suppose the scalar stream has length $N = M \cdot L$ for some factorization $N = M \cdot L$ with $M \geq 2$. The stream can then be reorganized into an $M$-by-$L$ array by grouping the samples into $L$ blocks of $M$ samples each, and each block can be treated as a single vector observation $\mathbf{x}_\ell \in \mathbb{C}^M$. The key question is whether the block structure admits a meaningful group action. If the $M$ samples within a block possess an algebraic symmetry (for example the cyclic group $\mathbb{Z}_M$ acting by circular shifts of the within-block sample indices), then the group-averaged estimator of equation~\eqref{eq:gae} becomes applicable to the block, and the $(G, L)$ continuum takes over: $L$ independent snapshots of a $\mathbb{Z}_M$-structured vector observation, with group gain $|G| = M$ on each snapshot. The scalar observation has been \emph{rank-promoted} from rank zero (a single number) to rank one (a vector with a nontrivial group acting on its components). The sample mean becomes an instance of the group-averaged estimator, the law of large numbers becomes an instance of the general algebraic averaging theorem (stated as Theorem~\ref{thm:gaat} below), and the framework of the previous section applies unchanged.

It is important to distinguish rank promotion from Permutation-Adapted Signal Estimation (PASE). PASE~\cite{thornton2026ad} is a \emph{sampling} scheme within a given group: given a signal whose matched group is already known to be $G$, PASE prescribes which group elements to evaluate in the sum of equation~\eqref{eq:gae}, ordered according to one of three structural coding levels (Level~1 draws antithetic pairs, Level~2 draws conjugacy-diverse elements together with their antithetic complements, Level~3 draws coset representatives). Rank promotion, in contrast, is a \emph{reorganization} of the data itself, in which a scalar stream that admits no nontrivial group action is restructured into a tensorial form on which a group acts. The two constructions are therefore complementary: rank promotion is the prior step that creates the state space on which a matched group can act, and PASE is a refinement that chooses which elements of that group to evaluate once the state space has been set up. Throughout what follows we take rank promotion as the foundational construction and use the term PASE only when discussing the sampling refinement within a given rank-promoted setting.

\subsection{The Trivial Group Embedding Theorem and GAAT}
\label{sec:gaat}

The law of large numbers is the universal statement that the sample mean of $N$ independent observations converges to the population mean at rate $1/\sqrt{N}$. Its status in the present framework is not that of a foundational result independent of AD, but rather that of the degenerate case of a broader algebraic-averaging principle in which the group acting on each observation is trivial. We now make this correspondence precise. The first result below is a reframing: it asserts that conventional temporal averaging is not an alternative to AD but a special case of it, namely the $(G = \{e\}, L)$ corner of the $(G, L)$ continuum introduced in the previous section.

\begin{theorem}[Trivial Group Embedding, TET]\label{thm:tet}
Conventional temporal averaging of $L$ independent observations is algebraic diversity with the trivial group $G = \{e\}$ applied to each observation independently. The sample covariance $\hat{\mathbf{R}}_L = L^{-1} \sum_{\ell=1}^L \mathbf{x}_\ell \mathbf{x}_\ell^H$ is the group-averaged estimator with $|G| = 1$ and $L$ snapshots, occupying the $(G = \{e\}, L)$ point on the continuum of Section~\ref{sec:gl_manifold}.
\end{theorem}

The Trivial Group Embedding Theorem, which we will refer to as TET throughout the paper, has the following interpretation. Every estimator that is built on the sample covariance has, implicitly, been operating at the trivial-group end of the $(G, L)$ continuum, and has been relying on physical repetition (collecting $L$ independent snapshots) as its sole source of variance reduction. The AD machinery extends this by replacing the trivial group with a matched group of cardinality $|G| > 1$, producing a multi-fold within-observation variance reduction that is algebraic in origin. The two reductions compose multiplicatively to yield the $(G, L)$ continuum scaling $1/(|G| \cdot L)$. TET is therefore a consistency statement between AD and classical averaging: AD does not replace the law of large numbers but rather subsumes it as the specific case in which no algebraic structure is available or being exploited.

Before stating the extension of TET beyond the outer product, we fix a quantity that governs the variance-reduction factor precisely, and that has appeared informally in the Background and the $(G, L)$ continuum subsection under the name \emph{effective group order}. The full definition is slightly subtle because the quantity depends on both the group $G$ and the statistic $f$, and the subtlety is worth spelling out before it is invoked in the theorems that follow.

\begin{definition}[Effective group order]\label{def:deff}
Let $G$ be a finite group with unitary representation $\pi: G \to U(M)$ acting on $\mathbb{C}^M$, let $f: \mathbb{C}^M \to \mathcal{V}$ be a statistic taking values in a complex vector space $\mathcal{V}$, and let $\mathbf{x} \in \mathbb{C}^M$ be a signal admitting $G$ as a matched group. The \textit{effective group order} of the pair $(G, f)$, denoted $d_{\mathrm{eff}}(G, f)$, is the dimension of the complex vector space spanned by the statistic evaluated on every element of the group orbit, that is
\begin{equation}\label{eq:deff}
d_{\mathrm{eff}}(G, f) \;=\; \dim_{\mathbb{C}} \mathrm{span}\bigl\{\, f(\pi_g \mathbf{x}) \,:\, g \in G \,\bigr\},
\end{equation}
where the dimension is evaluated for a \emph{generic} signal $\mathbf{x}$, meaning for almost every $\mathbf{x}$ with respect to the natural Lebesgue measure on the signal space.
\end{definition}

The effective group order is the number of algebraically distinct evaluations of the statistic that the group orbit produces, counted in the linear-algebra sense: two evaluations are treated as distinct if they are not related by a linear combination of each other, and counted as one if they are. The quantity $d_{\mathrm{eff}}(G, f)$ is therefore the arithmetic resource made available by averaging the statistic $f$ over the orbit of $\mathbf{x}$ under $G$, and it is this quantity (rather than the raw group order $|G|$) that bounds the variance reduction in the theorems of this section. Its role is that of a degrees-of-freedom count: it fixes the number of free coordinates the orbit average populates, and hence the effective sample size available to any functional of the resulting estimate. The per-entry variance reduction itself is a companion quantity: it equals $d_{\mathrm{eff}}$ exactly when the energy of the averaged data is distributed flatly over those coordinates, and in general it is a participation-ratio functional of the orbit energy distribution bounded above by $d_{\mathrm{eff}}$; Theorem~\ref{thm:outer} states both roles precisely for the outer product. Three consequences of the definition merit attention, because each is needed at a different point in the development that follows.

First, whenever the group acts on the observation with nontrivial intersections among its orbit points (for instance when $\mathbf{x}$ has extra symmetry that allows some $\pi_g \mathbf{x}$ to coincide), $d_{\mathrm{eff}}$ may be strictly smaller than $|G|$. For a generic signal under a free action, however, the orbit has $|G|$ distinct points and the statistic evaluated on them is generically linearly independent, so $d_{\mathrm{eff}}(G, f) = |G|$ is the default case.

Second, for the outer-product statistic $f(\mathbf{x}) = \mathbf{x} \mathbf{x}^H$ with a matched group acting on $\mathbf{x}$, the Peter-Weyl decomposition gives a sharp accounting: the orbit of the outer product spans the regular representation of $G$, and its dimension equals the sum of squared irreducible representation dimensions, $d_{\mathrm{eff}}(G, f_{\mathrm{outer}}) = \sum_i d_i^2$. For an Abelian group every $d_i = 1$, so $d_{\mathrm{eff}} = |G|$ and the group order is fully available as a variance-reduction factor; for a non-Abelian group with some $d_i > 1$, the same relation holds and the non-Abelian irreducibles contribute $d_i^2$ each to the effective order.

Third, for a statistic that is invariant under the full group (for example, any $S_M$-invariant scalar statistic such as a function of the elementary symmetric polynomials of the components of $\mathbf{x}$), every orbit element produces the same value, the span collapses to a one-dimensional subspace, and $d_{\mathrm{eff}}(G, f) = 1$ regardless of $|G|$. This is the structural reason that the law of large numbers scaling $\sigma^2 / N$ governs the $S_M$-symmetric scalar case: the group-averaged estimator under $S_M$ is a single number, not an $|G|$-fold average, and the only available source of variance reduction is classical repetition across independent observations.

With the effective group order pinned down, the extension of the trivial-group embedding to a broader class of statistics becomes clean to state. We begin with the outer-product case, for which a sharp variance expression is available.

\begin{theorem}[Outer-product algebraic averaging]\label{thm:outer}
Let $\mathbf{x} = \mathbf{s} + \mathbf{n}$ with $\mathbf{s}$ deterministic and $\mathbf{n}$ having zero mean. Let $G$ be a finite group with unitary representation $\pi$. Under signal equivariance and noise ergodicity, the outer-product group-averaged estimator of equation~\eqref{eq:gae} is unbiased for the $G$-invariant projection of $\mathbf{R} = \mathbb{E}[\mathbf{x}\mathbf{x}^H]$ and has variance on each entry reduced by a factor $g(G, \mathbf{x})$ satisfying $1 \leq g(G, \mathbf{x}) \leq d_{\mathrm{eff}}(G, f)$, a participation-ratio functional of the distribution of the averaged energy over the orbit coordinates. The ceiling $g = d_{\mathrm{eff}}$ is attained exactly when that distribution is flat: for constant-modulus $\mathbf{s}$ in white noise, for noise-only estimation ($\mathbf{s} = \mathbf{0}$, white $\mathbf{n}$), and, in the Gaussian stochastic extension $\mathbf{x} \sim \mathcal{CN}(\mathbf{0}, \mathbf{R})$ with $[\pi_g, \mathbf{R}] = 0$, when the matched-basis spectrum of $\mathbf{R}$ is flat. In the Gaussian Abelian case the factor is exact and computable from the spectrum $\{\lambda_m\}$ of $\mathbf{R}$:
\begin{equation}\label{eq:g_realized}
g(G, \mathbf{x}) \;=\; \frac{\bigl(\sum_m \lambda_m\bigr)^2}{\sum_m \lambda_m^2},
\end{equation}
a consequence of the per-coordinate law $\Var(\hat{\lambda}_m) = \lambda_m^2$ of Theorem~\ref{thm:converse} together with Parseval's identity on the character basis. The $(G, L)$ continuum gives $\Var(\hat{R}_{G, L, ij}) = \Var(\hat{R}_{\{e\}, ij}) / (L \cdot g(G, \mathbf{x}))$, with the flat-profile ceiling $\Var(\hat{R}_{\{e\}, ij}) / (L \cdot d_{\mathrm{eff}})$.
\end{theorem}

A sharper statement is available in the Gaussian case through a separation of the Abelian and non-Abelian regimes. The following result makes this precise and identifies the group-averaged estimator with the maximum-likelihood estimator, thereby establishing that the variance reduction it delivers attains the information-theoretic lower bound and cannot be improved by any other estimator. The statement is drawn from the Converse Theorem of the companion paper~\cite{thornton2026ad}, to which we refer for the proof; the quantities it involves are precisely those defined in the preceding Definition~\ref{def:deff} and will be used throughout the remainder of the paper.

\begin{theorem}[Converse Theorem, Abelian and non-Abelian cases]\label{thm:converse}
Let $\mathbf{x} \sim \mathcal{CN}(\mathbf{0}, \mathbf{R})$ with $[\pi_g, \mathbf{R}] = 0$ for all $g \in G$. Let $\mathbf{R}$ have eigenvalues $\{\lambda_i\}$ with multiplicities tied to the irreducible representations of $G$ of dimensions $\{d_i\}$. Then the group-averaged estimator is the maximum-likelihood estimator for the $G$-invariant parameters and achieves the Cram\'{e}r-Rao bound with equality:
\begin{equation}\label{eq:converse}
\Var(\hat{\lambda}_i) \;=\; \frac{\lambda_i^2}{d_i} \;=\; \mathrm{CRB}(\lambda_i).
\end{equation}
For Abelian groups all $d_i = 1$ and $\Var(\hat{\lambda}_k) = \lambda_k^2$. For non-Abelian groups with $d_i > 1$, irreducible representations of dimension $d_i$ provide a $1/d_i$ within-block variance reduction that Abelian groups structurally cannot access.
\end{theorem}

The non-Abelian part of the Converse Theorem deserves a brief comment because it is the point at which the distinction between Abelian and non-Abelian matched groups has concrete estimation-theoretic consequences. The proof proceeds via Schur's lemma, applied within each isotypic block of the representation: averaging the covariance over an irreducible representation of dimension $d_i$ forces the resulting estimator to be proportional to the $d_i \times d_i$ identity matrix on that block. The effect is a $d_i$-fold within-block variance reduction that Abelian groups, whose irreducible representations are all one-dimensional, structurally cannot deliver. The consequence at the estimator level is that a non-Abelian matched group estimates $r < M$ distinct parameters at variance $\lambda_i^2 / d_i$ per parameter, whereas an Abelian matched group of the same cardinality estimates $M$ parameters at variance $\lambda_k^2$ per parameter. When a signal's algebraic structure genuinely supports a non-Abelian matched group, the non-Abelian choice is therefore strictly superior in estimation; when the algebraic structure is only Abelian, no such additional within-block reduction is available. The following result, which we call Supergroup Dominance, uses this observation to establish a partial order among matched-group candidates related by subgroup containment.

\begin{theorem}[Supergroup Dominance]\label{thm:supergroup}
Let $G$ act on $\{1, \ldots, M\}$ by a unitary representation with $\delta(G, \mathbf{R}) = 0$, and let $H \subset G$ be any proper subgroup. Then
\begin{enumerate}[label=(\roman*), nosep]
\item $\mathrm{MSE}(\hat{\mathbf{R}}_G) \leq \mathrm{MSE}(\hat{\mathbf{R}}_H)$,
\item $\bar{D}_{CV}(G) \leq \bar{D}_{CV}(H)$,
\end{enumerate}
with strict inequality whenever the commutant containment $\mathrm{Comm}(G) \subset \mathrm{Comm}(H)$ is proper and the observation has an absolutely continuous distribution; when $\mathrm{Comm}(G) = \mathrm{Comm}(H)$ the two estimators coincide identically and equality holds.
\end{theorem}

\begin{proof}
Since $H \subset G$, every element of $H$ lies in $G$, and the commutant algebras satisfy $\mathrm{Comm}(G) \subseteq \mathrm{Comm}(H)$, so the Reynolds projectors nest: $\mathcal{P}_G = \mathcal{P}_G \circ \mathcal{P}_H$, and hence $\hat{\mathbf{R}}_G = \mathcal{P}_G(\hat{\mathbf{R}}_H)$. From $\delta(G, \mathbf{R}) = 0$ we have $\mathbf{R} \in \mathrm{Comm}(G) \subseteq \mathrm{Comm}(H)$, so $\delta(H, \mathbf{R}) = 0$ as well, both estimators are unbiased, and $\mathcal{P}_G(\mathbf{R}) = \mathbf{R}$. For every realization, decompose $\hat{\mathbf{R}}_H - \mathbf{R}$ orthogonally under the Hilbert--Schmidt inner product into its components inside and outside $\mathrm{Comm}(G)$:
\[
\bigl\|\hat{\mathbf{R}}_H - \mathbf{R}\bigr\|_F^2
\;=\; \bigl\|\hat{\mathbf{R}}_G - \mathbf{R}\bigr\|_F^2
\;+\; \bigl\|(\mathcal{I} - \mathcal{P}_G)\,\hat{\mathbf{R}}_H\bigr\|_F^2 .
\]
Taking expectations gives $\mathrm{MSE}(\hat{\mathbf{R}}_H) = \mathrm{MSE}(\hat{\mathbf{R}}_G) + \mathbb{E}\|(\mathcal{I} - \mathcal{P}_G)\hat{\mathbf{R}}_H\|_F^2$, which proves~(i) with no distributional assumption beyond finite second moments of $\hat{\mathbf{R}}_H$; the inequality is strict exactly when $\hat{\mathbf{R}}_H$ has a nonvanishing component outside $\mathrm{Comm}(G)$ with positive probability, which holds for absolutely continuous observations whenever $\mathrm{Comm}(G) \subset \mathrm{Comm}(H)$ is proper, and fails, with the two estimators coinciding, when the commutants are equal. For~(ii), the cross-validation criterion $D_{CV}$ defined in equation~\eqref{eq:dcv} measures the expected Frobenius distance between independent copies of the normalized estimator, and for unbiased estimators $\mathbb{E}\|\hat{\mathbf{R}} - \hat{\mathbf{R}}'\|_F^2 = 2\,\mathrm{MSE}(\hat{\mathbf{R}})$ by independence, so~(ii) follows from~(i).
\end{proof}

\begin{corollary}[Graph Automorphism Optimality]\label{cor:aut}
For a graph signal on $M$ vertices whose covariance commutes with the full automorphism group $\Aut(\mathcal{G})$, the full $\Aut(\mathcal{G})$ dominates every subgroup $H \subset \Aut(\mathcal{G})$ in both MSE and $D_{CV}$. The estimation gain over the largest Abelian subgroup $H_{\max}$ scales as $|\Aut(\mathcal{G})|/|H_{\max}|$.
\end{corollary}

\begin{remark}\label{rem:supergroup_scope}
Theorem~\ref{thm:supergroup} applies whenever $H \subset G$ (subgroup containment). It does \emph{not} apply to groups of the same order that are not subgroup-related, such as $D_4$ and $\mathbb{Z}_8$ at $M = 8$. In that case the Abelian group often dominates empirically at $|G| = M$ because its one-dimensional irreps give decoupled estimation, while the non-Abelian irreps of dimension $d_i > 1$ introduce coupled estimation within blocks. The practical consequence is that non-Abelian groups are preferred when their order substantially exceeds $M$ (as for $\Aut(\mathcal{G})$ on a highly symmetric graph), while Abelian groups are preferred in the common $|G| = M$ regime of cyclic measurement lattices.
\end{remark}

The results above treat the outer-product statistic, which is the case of greatest practical relevance for spectral estimation and second-order signal processing. The corresponding result for a general statistic is the \emph{General Algebraic Averaging Theorem}, which asserts that the variance-reduction phenomenon carries over beyond the outer product under appropriate structural assumptions on the statistic and the noise. The current proof status of the general statement is partial, with several important cases established rigorously and the remaining cases conditional on a Clebsch-Gordan decomposition of the cross-moment structure of the orbit. We present the full statement, then summarize the established and conditional cases, and then display the connection between the classical sample-moment formulas and their GAAT counterparts in tabular form.

\begin{theorem}[General Algebraic Averaging, GAAT]\label{thm:gaat}
Let $f: \mathbb{C}^M \to \mathcal{V}$ be a $G$-compatible statistic (one whose value at $\pi_g \mathbf{x}$ is determined by $g$ and $f(\mathbf{x})$ through a representation of $G$ on $\mathcal{V}$). Under signal equivariance and noise ergodicity, the group-averaged estimator
\begin{equation}\label{eq:gaat}
\hat{\theta}_G \;=\; \frac{1}{|G|}\sum_{g \in G} f(\pi_g \mathbf{x})
\end{equation}
is unbiased for the $G$-invariant projection of $\theta_f = \mathbb{E}[f(\mathbf{x})]$, and $\Var(\hat{\theta}_G) \leq C(f, \mathbf{s}) / d_{\mathrm{eff}}(G, f)$. The law of large numbers, $\Var \propto 1/N$ for the trivial group with $N$ observations, is the degenerate case.
\end{theorem}

The unbiasedness assertion follows from linearity of expectation and is unconditional. The variance assertion is established for several important cases and conditional for the general case. Specifically:
\begin{itemize}[nosep]
\item \textbf{Outer product ($k = 2$):} proved (Theorem~\ref{thm:outer}). The argument uses Schur orthogonality at order two together with the Isserlis identity for fourth moments.
\item \textbf{Scalar $S_M$-symmetric statistics ($d_{\mathrm{eff}} = 1$):} proved. The orbit of $f$ under $S_M$ collapses to a single value, recovering the law-of-large-numbers scaling.
\item \textbf{Sample mean and second central moment under rank promotion:} proved as instances of the outer-product case.
\item \textbf{Rao--Blackwell optimality:} proved. The group-averaged estimator is the conditional expectation of $f$ given the $G$-invariant sigma-field, and Rao--Blackwell guarantees that this conditioning never increases variance.
\item \textbf{General degree-$k$ statistics ($k \geq 3$):} \emph{conditional}. The proof reduces to a Clebsch--Gordan decomposition of the cross-moment structure of the orbit. For Gaussian noise the Isserlis identity reduces $2k$-th moments to second moments, and Schur orthogonality handles the group sum; for non-Gaussian noise additional moment conditions are required.
\end{itemize}

The four-moment case has been verified by Monte Carlo experiment to four-digit precision across five representative group types and eight $(G, L)$ configurations, giving numerical support to the GAAT statement in the regime of practical interest. Table~\ref{tab:moments} displays the parallel between the classical sample-moment formulas and their GAAT counterparts for the first four moments. The classical column averages over $N$ independent scalar observations in the familiar way; the GAAT column averages a single zero-indexed component of the signal vector over the $|G|$ images of one vector observation under the group action. The two columns produce estimators with matched first-order behaviour and variance scalings proportional to $1/N$ in the classical case versus $1/d_{\mathrm{eff}}$ in the algebraic case. When both axes are exploited simultaneously on $L$ independent rank-promoted observations, the two factors combine multiplicatively and the variance scales as $\sigma^2/(d_{\mathrm{eff}} \cdot L)$, which is the $(G, L)$ continuum at work.

\begin{table*}[t]
\centering
\caption{The first four sample moments in classical and GAAT formulations. The classical column averages over $N$ independent scalar observations; the GAAT column averages the zeroth component of $\pi_g \mathbf{x}$ over the $|G|$ group elements of one vector observation. The two columns produce estimators with matched first-order behaviour and variance scalings $\propto 1/N$ versus $\propto 1/|G|$ respectively. The bottom rows give the variance scaling and the $(G, L)$ continuum scaling that combines both axes.}
\label{tab:moments}
\renewcommand{\arraystretch}{1.5}
\small
\begin{tabular}{@{}lll@{}}
\toprule
\textbf{Statistic} & \textbf{Classical (discrete, $N$ obs.)} & \textbf{GAAT (discrete, single obs., group $G$)} \\
\midrule
Sample mean (1st)
  & $\bar{x} = \frac{1}{N} \sum_{i=1}^N x_i$
  & $\hat{\mu}_G = \frac{1}{|G|} \sum_{g \in G} [\pi_g \mathbf{x}]_0$ \\
Sample variance (2nd central)
  & $s^2 = \frac{1}{N-1} \sum_{i=1}^N (x_i - \bar{x})^2$
  & $\hat{\sigma}^2_G = \frac{1}{|G|} \sum_{g \in G} \bigl([\pi_g \mathbf{x}]_0 - \hat{\mu}_G\bigr)^2$ \\
Sample skewness (3rd standardized)
  & $g_1 = \frac{1}{N} \sum_{i=1}^N \bigl(\frac{x_i - \bar{x}}{s}\bigr)^3$
  & $\hat{\gamma}_G = \frac{1}{|G|} \sum_{g \in G} \bigl(\frac{[\pi_g \mathbf{x}]_0 - \hat{\mu}_G}{\hat{\sigma}_G}\bigr)^3$ \\
Sample kurtosis (4th standardized)
  & $g_2 = \frac{1}{N} \sum_{i=1}^N \bigl(\frac{x_i - \bar{x}}{s}\bigr)^4$
  & $\hat{k}_G = \frac{1}{|G|} \sum_{g \in G} \bigl(\frac{[\pi_g \mathbf{x}]_0 - \hat{\mu}_G}{\hat{\sigma}_G}\bigr)^4$ \\
\midrule
Variance of sample mean
  & $\sigma^2 / N$
  & $\sigma^2 / d_{\mathrm{eff}}$ \\
$(G, L)$ continuum (mean)
  & \multicolumn{2}{c}{$\sigma^2 / (d_{\mathrm{eff}} \cdot L)$} \\
\bottomrule
\end{tabular}
\end{table*}

A brief interpretive remark is in order before leaving the GAAT material. For a transitive group such as $\mathbb{Z}_M$ acting on $\mathbb{C}^M$ by cyclic shifts, the orbit of the zeroth component visits every component of the signal vector exactly once. Consequently the value $[\pi_g \mathbf{x}]_0$ ranges over all of $x_0, x_1, \ldots, x_{M-1}$ as $g$ ranges over $G$, and the GAAT mean given in the right-hand column of Table~\ref{tab:moments} is numerically identical to the classical mean of those $M$ values. The two estimators return the same number for the same input. What distinguishes them is not their point value but their variance scaling, and it is at the variance level that the substantive content of GAAT lies: under the signal equivariance and noise ergodicity assumptions, the group-averaged moment from a single rank-promoted observation has the same scaling exponent in $|G|$ that the classical moment has in $L$, and the two scalings combine on the $(G, L)$ continuum. The rank-promotion viewpoint therefore does not alter the point estimate; it alters the accounting of how that estimate was produced, and the accounting is what determines how the statistic's variance responds to additional data.

\subsection{Temporal Algebraic Diversity as the entry case of rank promotion}
\label{sec:tad}

Among the rank-promotion constructions, the one that arises most naturally from a stream of scalar samples is the circular partition: a scalar time series of length $N = M \cdot L$ is segmented into $L$ successive blocks of length $M$, and each block is regarded as a vector observation with the cyclic group $\mathbb{Z}_M$ acting on its components by circular shift. We refer to this as \emph{Temporal Algebraic Diversity} (TAD), and it occupies a privileged position among the rank-promotion forms available in the framework. The reason for the privileged status is representation-theoretic: the circular stratification endows the within-block index with both an Abelian symmetry (which admits PASE Level-1 antithetic pairing and Level-3 coset representatives) and a full conjugacy structure (which admits PASE Level-2 conjugacy-diverse sampling). TAD is the unique rank-promotion case that admits all three PASE structural coding levels simultaneously, and this fact explains its frequent appearance in the applications of later sections.

A concrete quantitative illustration of the rank-promotion speedup is furnished by a classical problem: the Monte Carlo estimation of the transcendental constant $\pi$ via the quarter-circle integral $\pi = \int_0^1 4\sqrt{1-u^2}\,du$. The standard Monte Carlo procedure draws $N$ independent uniform samples $U_1, \ldots, U_N$ on the interval $[0, 1]$ and averages the integrand values $4\sqrt{1 - U_i^2}$ to form an estimate. This is a rank-zero (scalar) estimator whose matched group is trivial, and its convergence rate is the familiar $O(1/\sqrt{N})$ law-of-large-numbers rate. Under rank promotion, the unit interval is partitioned into $M$ equal strata $S_k = [k/M, (k+1)/M)$ for $k = 0, 1, \ldots, M-1$, and a single sample is drawn from each stratum, producing an $M$-dimensional structured observation whose strata are in one-to-one correspondence with the cosets of $\mathbb{Z}_M$ acting on $[0, 1]$ by cyclic rotation. One sample per coset is the PASE Level-3 construction (coset representatives), and the resulting integral estimator is the group-averaged estimator under $\mathbb{Z}_M$ operating on the stratified observation. In an illustrative experiment comparing the two procedures on the same integrand, rank promotion with PASE Level-3 structural coding reached six correct digits of $\pi$ in 140 iterations, while standard Monte Carlo on the same target required approximately 29{,}000 iterations to reach the same criterion; this corresponds to a speedup of approximately $207\times$ on the six-digit convergence criterion.

A methodological clarification is in order because the specific problem is rather special: the purpose of the experiment is to illustrate the rank-promotion mechanism rather than to compete with specialized $\pi$-computation methods. Dedicated algorithms such as the Chudnovsky series, the BBP formula, and the arithmetic-geometric-mean iteration compute far more digits of $\pi$ at far higher speed and are outside the scope of this comparison. What the experiment demonstrates is that the rank-promotion operation, which converts a scalar Monte Carlo target into a structured array admitting a nontrivial matched group, produces a tangible and measurable speedup whose origin is algebraic rather than arithmetic, and that the speedup follows from the variance-reduction accounting of the $(G, L)$ continuum exactly as Theorem~\ref{thm:gaat} predicts.

A second concrete example, from the clinical domain, brings the rank-promotion idea closer to routine practice and clarifies its practical value. A patient's systolic blood pressure is classically measured by taking $N = 10$ independent readings across a period of two weeks and averaging them; the variance of the resulting estimate scales as $\sigma^2/10$ by the law of large numbers. Under rank promotion with the cyclic group $\mathbb{Z}_{10}$ applied to a single cycle of 10 temporally consecutive measurements taken on a single clinic visit, the group-averaged mean achieves the same variance reduction from what is, in the rank-promoted representation, effectively a single structured observation of a 10-vector. The $(G, L)$ continuum of the previous section tells us that all four budget allocations $(d_{\mathrm{eff}}, L) \in \{(1, 10), (2, 5), (5, 2), (10, 1)\}$ produce estimators of the same variance, so the clinician has complete freedom to trade physical repetition for algebraic structure according to operational convenience. Fewer clinic visits at the same diagnostic quality is the practical consequence, once the rank-promoted view of the scalar time series is adopted.

\subsection{Ramifications for the law of large numbers}

The results of this section admit a unified interpretation that is worth stating explicitly, because it reorients the reader's view of one of the most frequently invoked results in classical statistics. The law of large numbers is the universal statement that the sample mean of $N$ independent observations of a random variable converges to the population mean at rate $1/\sqrt{N}$, and generations of signal processing and estimation theory have treated it as a first principle from which other variance-reduction results derive. The rank-promotion viewpoint developed here invites a slightly different reading. Every estimator of the averaging form $\hat{\theta}_N = N^{-1} \sum_{i=1}^N f(x_i)$ has an algebraic counterpart of the form $\hat{\theta}_G = |G|^{-1} \sum_{g \in G} f(\pi_g \mathbf{x})$, which, under rank promotion of the scalar sequence into vector form together with the signal-equivariance and noise-ergodicity conditions, achieves variance $\sigma^2 / |G|$ from a single structured observation. The scaling exponents of the classical $N$-observation average and the algebraic $|G|$-element orbit average coincide, and the two axes combine multiplicatively into the $(G, L)$ continuum whose scaling has appeared throughout this section. The net implication is not that the law of large numbers is in any way incorrect, but that it is the degenerate case (corresponding to the trivial group) of a broader algebraic averaging principle. Where a signal admits nontrivial algebraic structure, the group-order axis $|G|$ can substitute for, or can supplement, the sample-count axis $L$, and the practitioner recovers the design freedom illustrated by the blood-pressure example of the preceding subsection.

\subsection{The PASE hierarchy and the structural coding rate conjecture}
\label{sec:pase_rate}

Once a rank-promoted representation has been established and a matched group $G$ acting on it has been identified, a secondary question arises that is distinct from the matched-group identification itself: within the $|G|$ elements of the group, is every choice of $|G|$-fold averaging equally informative, or are some choices algebraically more efficient than others? The answer, developed in the companion paper and summarized here for completeness, is that the choice of which group elements to include in the average is itself a meaningful design decision. The Permutation-Adapted Signal Estimation (PASE) hierarchy~\cite{thornton2026ad} organizes three structural coding strategies for sampling the group orbit, ordered by increasing sophistication in the sense that each strategy strictly refines the preceding one by using representation-theoretic structure that the preceding one ignored. We list the three levels below, after which we state the structural coding rate conjecture that governs how many orbit elements a given signal's structural content demands.
\begin{itemize}[nosep]
\item \emph{Level~1 (antithetic pairs)}: for each selected permutation $\sigma$, the inverse permutation $\sigma^{-1}$ is also included. Averaging the pair projects out the antisymmetric noise component of the statistic, cutting the number of required group elements to $n^* \leq \lceil M/2 \rceil$ for structurally simple signals.
\item \emph{Level~2 (conjugacy-diverse plus antithetic)}: permutations are selected from distinct cycle types of $S_M$ (the conjugacy classes of $S_M$), and each such permutation is paired with its inverse. Different cycle types produce algebraically distinct views of the signal in the sense that they activate different irreducible representations of the acting group, so the resulting orbit sample is representation-theoretically rich.
\item \emph{Level~3 (coset representatives)}: one permutation is selected from each left coset of the matched group $G$ in $S_M$. The coset representatives are algebraically independent with respect to $G$, and the resulting orbit sample gives the smallest $n^*$ of the three levels.
\end{itemize}
A cautionary remark is in order, because a natural diversity metric fails in a way that is informative about what PASE actually measures. If one were to rank permutations by Hamming distance on the symmetric group (the number of index positions at which two permutations differ), one would conclude that cyclic shifts are maximally diverse because their pairwise Hamming distance is maximal. In fact, cyclic shifts under the cyclic group produce the \emph{worst} spectral estimation among tested strategies, because Hamming distance measures positional displacement rather than representation-theoretic independence. The PASE hierarchy ranks strategies by the algebraic quantity that actually governs variance, which is the span of the statistic under the group orbit, and that quantity can be largest or smallest along directions that are unrelated to positional separation of the permutations.

The number of orbit elements required to saturate the variance-reduction bound scales with the structural content of the signal itself, measured through a quantity that plays the role of a Shannon-style entropy of the signal's spectral distribution. We call this quantity the \emph{structural entropy}:
\begin{equation}\label{eq:hstruct}
H_{\mathrm{struct}}(\mathbf{R}) \;=\; -\sum_k p_k \log_2 p_k, \qquad p_k = \lambda_k(\mathbf{R}) / \Tr(\mathbf{R}).
\end{equation}
$H_{\mathrm{struct}}$ quantifies how broadly the signal's variance is spread across its spectral modes: concentrated on a few modes gives low $H_{\mathrm{struct}}$, spread uniformly across all modes gives the maximum $H_{\mathrm{struct}} = \log_2 M$. Experimental validation across eight covariance models, covering AR and MA processes, sparse spikes, sinusoids in noise, and graph-Laplacian signals, produces a regularity in how many orbit elements the rank-promoted estimator requires to reach its variance floor. We express the regularity as a conjecture, because a formal proof is not yet available although the empirical support is substantial.

\begin{conjecture}[Structural Coding Rate]\label{conj:rate}
In the structured regime $H_{\mathrm{struct}} < \frac{1}{2}\log_2 M$, the optimal number of permutations satisfies
\begin{equation}
n^*(\mathbf{R}) \;\approx\; \lceil 2^{H_{\mathrm{struct}}(\mathbf{R})} \rceil,
\end{equation}
with $n^*/2^{H_{\mathrm{struct}}} \in [0.74, 1.20]$ confirmed empirically. In the diffuse regime $H_{\mathrm{struct}} \to \log_2 M$, $n^* = \Theta(M)$.
\end{conjecture}

Conjecture~\ref{conj:rate} has the interpretation of a structural analog of Shannon's source coding theorem, and the analogy is close enough to state explicitly. In Shannon's setting, a source with entropy $H$ bits per symbol can be compressed to $H$ bits per symbol and not fewer, and compression is possible when there is structure in the source distribution (non-uniformity of the symbol probabilities) and not when the source is uniform. In the structural setting, a signal with structural entropy $H_{\mathrm{struct}}$ bits of internal organization can be estimated from $2^{H_{\mathrm{struct}}}$ group actions and not fewer, and orbit-sampling compression is possible when there is structure in the spectrum (non-uniformity of the spectral distribution) and not when the spectrum is flat. The empirical support for the conjecture covers a broad range of covariance models with the match ratio $n^*/2^{H_{\mathrm{struct}}}$ staying within $[0.74, 1.20]$, but a formal proof is not yet available as of this writing. A proof is expected to follow from a Fisher-information argument for the statistic's variance as a function of orbit-sampling density, together with the Rényi-2 entropy interpretation of the structural capacity $\kappa$ that appears in Section~\ref{sec:fourthms} below, but the details have not yet been worked out.

\subsection{Algebraic waveform diversity}
\label{sec:waveform}

All of the discussion so far has treated rank promotion as an operation performed on received data: a scalar stream is reorganized into a structured array, a matched group is identified, and the group-averaged estimator extracts information from the structured observation. A dual viewpoint exists in which the rank-promotion operation is moved from the receiver to the transmitter: rather than arranging the received data into a group-structured form at analysis time, one arranges the transmitted waveform itself into a group-structured form at design time. Given a target or channel that is hypothesized to possess an algebraic symmetry group $G$, a transmit waveform sequence chosen as an orbit of a seed waveform under $G$ produces a response whose single-snapshot group-averaged estimator is automatically full-rank, because the structure that the receiver would ordinarily have to discover in the rank-promoted received data has been imposed at transmission. We refer to this active counterpart of rank promotion as \emph{algebraic waveform diversity}. Several practical examples illustrate the principle and motivate the development we leave to forthcoming companion papers.

In single-pulse ultrasonic non-destructive testing, one chooses a pulse-shape set whose elements are matched to the expected material dispersion classes of the target, so that a single pulse-echo return from a layered structure (for example, a pipe wall with potential internal corrosion) admits separation of the overlapping echoes by cumulative dispersion coefficient rather than requiring a temporally separated sequence of pulses. In single-pulse biomedical ultrasonography, the pulse-shape set is matched to expected tissue attenuation and dispersion classes, so that a single A-line return supports tissue-class discrimination from a single insonification rather than from many. In multi-frequency imaging, a frequency set matched to expected material-class structure allows classification from a simultaneous multi-frequency insonification. In single-shot quantum state tomography, a measurement-basis set matched to the unitary group action expected to diagonalize the state recovers spectral information from a single quantum measurement, a topic we have developed in detail in a separate companion paper on quantum algebraic diversity~\cite{thornton2026quantum}. In each case the common structure is the same: the group-orbit organization is imposed at the waveform or measurement-design stage rather than discovered after the fact in the received signal.

\section{The Eigentensor Hierarchy}
\label{sec:eigentensor}

The rank-promotion construction of the previous section treated a single rank-one observation vector with a single matched group acting on its components. Many problems of practical interest exhibit symmetry at multiple levels simultaneously: a sequence of vector measurements may share a within-observation symmetry that governs each element individually, and at the same time an across-observation symmetry that governs the relationship among the elements of the sequence. The natural extension of rank promotion to this multi-level setting is the \emph{eigentensor hierarchy}, in which a sequence of group actions is applied in succession and the resulting estimator carries information at each level of the resulting tensor decomposition. We develop the construction here, describe the compound group structure that arises when two actions are composed, and present a worked example drawn from multi-frequency ultrasonic non-destructive testing that makes the estimator's operational content concrete.

Let the observation be a sequence $\mathbf{x}_1, \mathbf{x}_2, \ldots, \mathbf{x}_K$ of vector measurements that share a common within-observation symmetry group $G_1$ (acting on the components of each $\mathbf{x}_k$) and an across-observation symmetry group $G_2$ (acting on the indexing of the sequence of measurements). The combined tensor observation admits a compound group action constructed from $G_1$ and $G_2$, and the form of the compound action depends on how the two groups interact. When the two group actions commute, in the sense that the within-observation action at each index $k$ does not depend on the across-observation permutation of the indices, the compound action is the direct product $G_1 \times G_2$. When the two actions do not commute, a more careful construction is required: the within-observation action at each index $k$ depends on how the index was reached by the across-observation permutation, and the appropriate compound is the \emph{wreath product} $G_1 \wr G_2$, an extension of the direct product that allows the $G_1$-action on each index to depend on the $G_2$-action on the indexing. The wreath product is the standard algebraic tool for hierarchical symmetry and is the natural object that governs multi-level rank promotion in non-commuting cases.

The group-averaged estimator, evaluated at the compound tensor level, produces an \emph{eigentensor} whose decomposition captures both the within-observation spectral structure (indexed by irreducible representations of $G_1$) and the across-observation structure (indexed by irreducible representations of $G_2$). Eigentensors form a hierarchy indexed by the number of nested group actions that have been composed. Level-1 eigentensors are the basic AD eigenstructure of a single vector observation under a single matched group, which is the object of Section~\ref{sec:rankpromote}. Level-2 eigentensors capture compound spatial-by-frequency structure in multi-frequency sensor data of the kind we treat as the worked example below. Higher levels in principle capture nested symmetry structure in hierarchical data acquisition, although practical applications to date have not gone beyond Level 2. We proceed with the Level-2 ultrasonic example.

Level-2 eigentensor analysis has been demonstrated on multi-frequency ultrasonic non-destructive testing (NDT) data for pipe wall inspection, and the construction there illustrates both the spatial-frequency compound symmetry and the practical payoff of extracting information at multiple symmetry levels simultaneously. The scenario is as follows. A single test-point measurement of a steel pipe wall is performed using four ultrasonic probe frequencies (1, 2.5, 5, and 10~MHz), the four being obtained simultaneously or in rapid sequence through a single multi-frequency transducer. At each frequency, a Level-1 AD estimator computes the spectral concentration $\psi$ of the pulse-echo return under the cyclic group acting on the transducer array; this produces a four-element Level-1 characterization vector $\boldsymbol{\psi} = (\psi_1, \psi_{2.5}, \psi_5, \psi_{10})$ per test point. A Level-2 AD estimator then applies the symmetric group $S_4$ to this characterization vector as its across-observation action and extracts the Level-2 spectral concentration $\psi^{(2)}$. In the terminology introduced above, the within-observation group is $G_1 = \mathbb{Z}_M$ (the cyclic array group), the across-observation group is $G_2 = S_4$ (the symmetric group on the frequency set), and the combined action on the sequence of pulse-echo returns is the wreath product $\mathbb{Z}_M \wr S_4$, in the $G_1 \wr G_2$ ordering introduced above.

The interpretation of the Level-2 output is direct and physically meaningful. Healthy steel exhibits approximately frequency-independent scattering across the 1--10~MHz band, so the Level-1 $\boldsymbol{\psi}$ vector is nearly uniform in its four components (which is to say, rank-1 in the characterization space) and the Level-2 concentration is correspondingly high ($\psi^{(2)} \approx 0.878$ in Monte Carlo averages). Corroded steel, in contrast, scatters ultrasound much more strongly at high frequencies, because the oxide layer introduces grain-scale roughness and porosity on the scale of the short wavelengths. Consequently the Level-1 $\boldsymbol{\psi}$ vector becomes strongly frequency-dependent (higher rank in characterization space), and the Level-2 concentration is different in value ($\psi^{(2)} \approx 0.929$ under the sign convention adopted here, reflecting a different dominant eigenvector orientation). The scalar $\psi^{(2)}$ therefore encodes the full multi-frequency scattering signature of the test point in a single number, and the sign of its difference from a baseline determines the defect class from a single observation.

The quantitative performance of the Level-2 classifier is as follows. At 15~dB SNR, which is representative of in-field contact ultrasonic measurements, a Monte Carlo study of 50 healthy and 50 corroded samples produces a between-class separation of $3.3\sigma$ on $\psi^{(2)}$ and a simple-threshold classification accuracy of 92\% on the held-out test set, with zero spatial scanning and no baseline subtraction. At higher SNR (30~dB, representative of well-coupled contact transducers), the separation rises to $6.6\sigma$ with perfect classification on the same sample size. Figures~\ref{fig:ndt_l1}, \ref{fig:ndt_l2}, and \ref{fig:ndt_cm} show the Level-1 frequency-dependent $\psi$ profiles, the Level-2 separation histogram, and the resulting confusion matrix at 15~dB SNR.

\begin{figure}[t]
\centering
\includegraphics[width=\columnwidth]{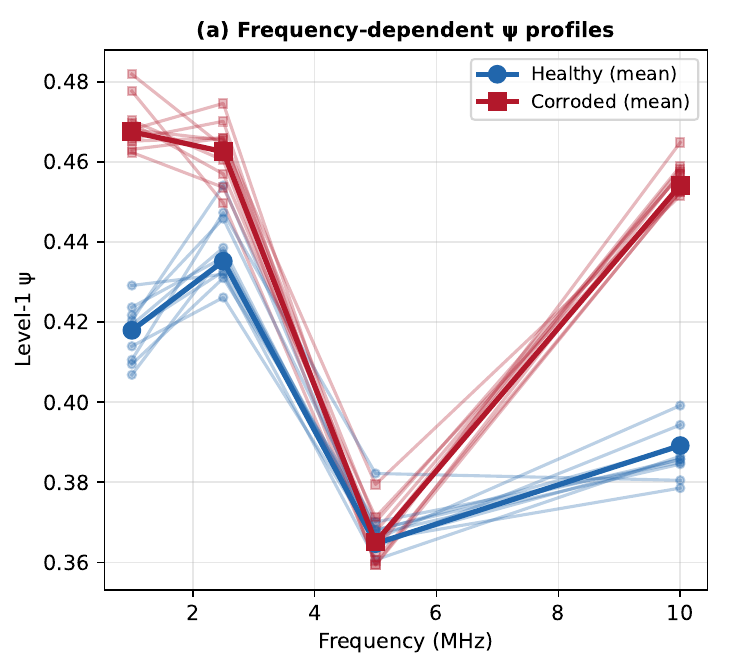}
\caption{NDT pipe corrosion detection, Level-1: frequency-dependent $\psi$ profiles for healthy (blue) and corroded (red) samples across four ultrasonic probe frequencies (1, 2.5, 5, 10~MHz). Healthy steel produces nearly uniform Level-1 spectral concentration; corroded steel produces a strongly frequency-dependent Level-1 concentration due to oxide-layer scattering at high frequencies. Light traces are individual samples (50 per class); heavy traces are class means.}
\label{fig:ndt_l1}
\end{figure}

\begin{figure}[t]
\centering
\includegraphics[width=\columnwidth]{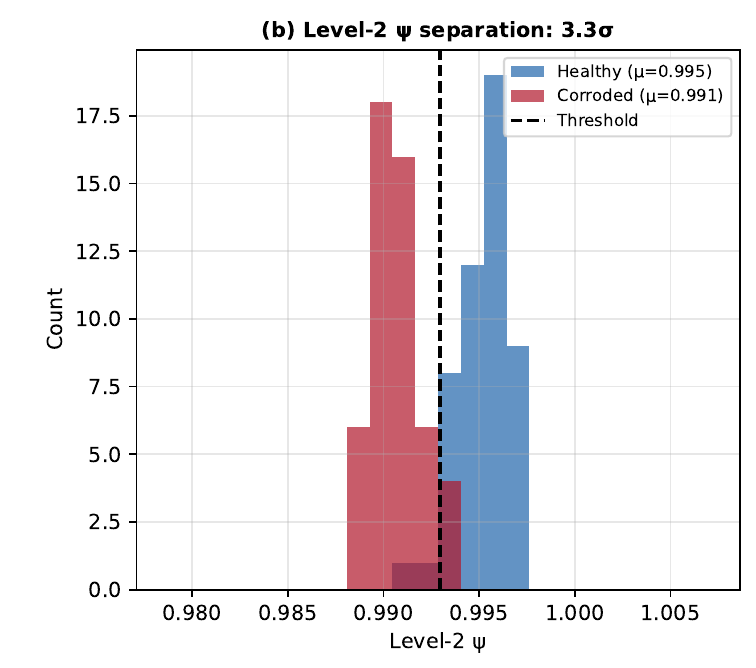}
\caption{NDT pipe corrosion detection, Level-2: histogram of the Level-2 spectral concentration $\psi^{(2)}$ across 50 healthy and 50 corroded samples. The two classes separate by $3.3\sigma$ at 15~dB SNR; the dashed line indicates the simple decision threshold used in Figure~\ref{fig:ndt_cm}.}
\label{fig:ndt_l2}
\end{figure}

\begin{figure}[t]
\centering
\includegraphics[width=\columnwidth]{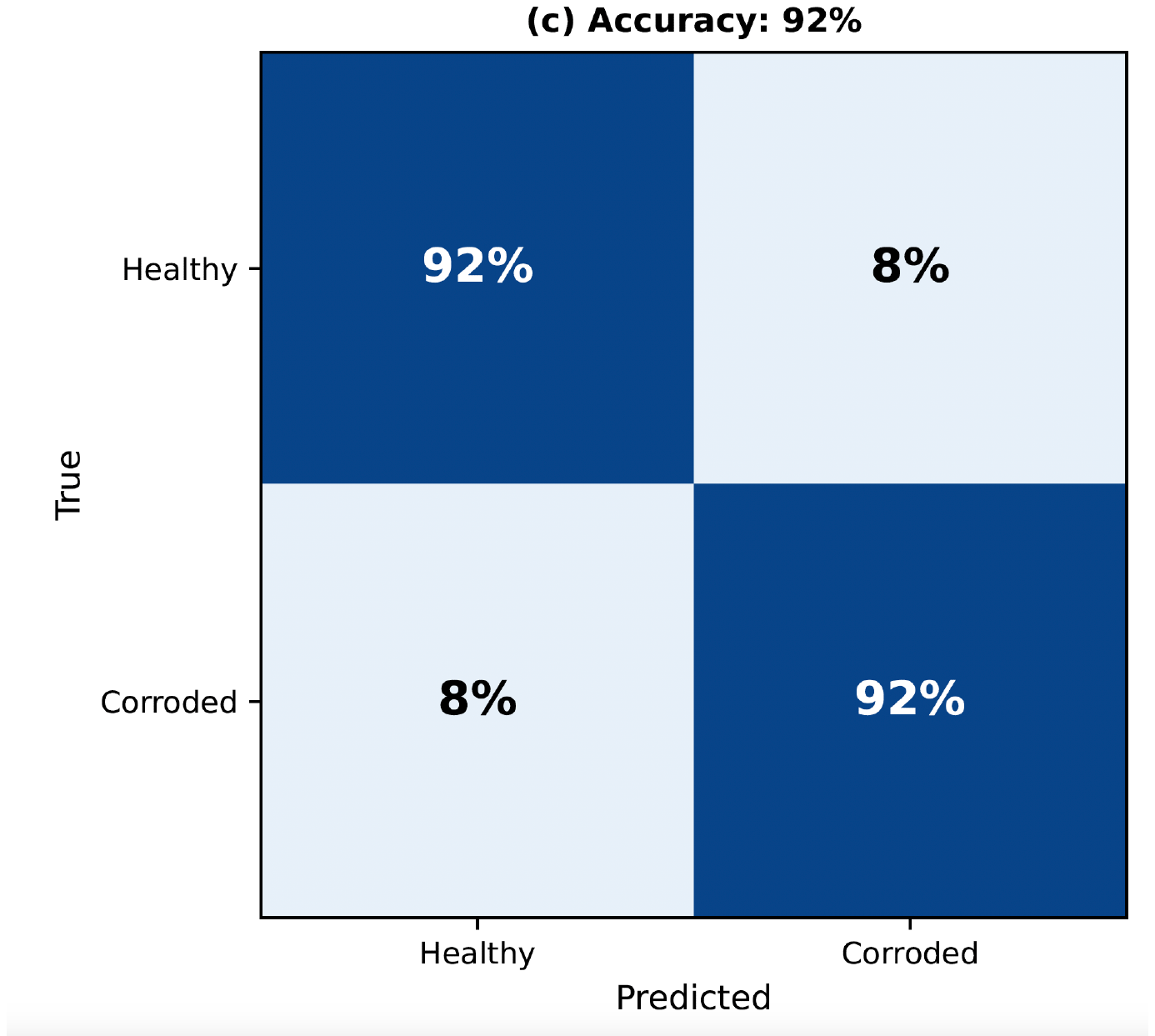}
\caption{NDT pipe corrosion detection: confusion matrix from a single-threshold classifier on the Level-2 statistic $\psi^{(2)}$ shown in Figure~\ref{fig:ndt_l2}. Single-point classification accuracy is $92\%$ at 15~dB SNR with no spatial scanning and no baseline subtraction.}
\label{fig:ndt_cm}
\end{figure}

The comparison to current industrial practice in non-destructive pipe wall inspection is substantive enough to warrant a brief digression, because it illustrates what algebraic diversity offers beyond a variance reduction in the abstract. Single-frequency thickness gauging, the workhorse of in-service pipe inspection, does not reliably detect early-stage pitting because the defect footprint is small relative to the beam footprint; industry references explicitly note that the detection of pitting by conventional ultrasonic methods is unreliable when the defect size is small compared to the inspected area~\cite{qualitymag2016}. Phased-array ultrasonic testing improves the probability of detection through beam steering and focusing, but at the cost of requiring specialized trained operators and substantially more expensive instrumentation~\cite{teaminc2025}. A comparative study of magnetic flux leakage and ultrasonic methods on ferrous pipe found that magnetic flux leakage overestimates pit depth by approximately 10\% while ultrasonic testing underestimates by approximately 10\%, with one ultrasonic inspection team missing two pits entirely despite knowing their approximate locations~\cite{ndtnet2000}. Crucially, none of these conventional methods performs defect \emph{classification}; they measure wall thickness at discrete points and leave the interpretation of what the thickness profile means to the technician.

The eigentensor approach provides three distinct advantages over this conventional practice, each of which is a direct consequence of the algebraic-diversity view at the Level-2 compound symmetry. First, the eigentensor estimator classifies defect type (healthy, corroded, cracked, porous) rather than merely flagging defect presence, by exploiting the frequency-dependent scattering signature that conventional single-frequency methods collapse into a single thickness reading. Second, the eigentensor estimator operates from a single multi-frequency observation at each test point, requiring no spatial scanning, no baseline subtraction, and no return-visit protocol for comparison against a reference scan. Third, the scalar $\psi^{(2)}$ provides a quantitative, operator-independent confidence metric: a fielded system can report a numeric defect-likelihood value in place of technician-dependent A-scan interpretation. The accuracy figures cited above are achieved without any of the scanning or baseline subtraction that would further improve performance in a production deployment, so they should be read as a lower bound on achievable operational performance rather than as a best-case figure.

\subsection{Tensor-FFT computation and the Abelianization trick}
\label{sec:tensor_fft}

A practical concern associated with any proposed estimator is its computational cost, and the group-averaged estimator is no exception. The nominal cost of evaluating equation~\eqref{eq:gae} is $O(|G| \cdot M^2)$, which for large matched groups would be prohibitive. In practice, the cost reduces dramatically for the two most common classes of matched groups, and we sketch the reductions here together with a useful approximation for the more difficult non-Abelian case.

For Abelian groups, the fundamental theorem of finite Abelian groups provides the reduction. Every finite Abelian group is isomorphic to a direct product of cyclic groups, $G \cong \mathbb{Z}_{n_1} \times \mathbb{Z}_{n_2} \times \cdots \times \mathbb{Z}_{n_k}$, and the corresponding group-Fourier transform factors as a $k$-dimensional tensor product of one-dimensional DFTs of lengths $(n_1, n_2, \ldots, n_k)$. The group-averaged estimator is therefore exactly computable via a tensor reshape followed by per-axis fast Fourier transforms, at total cost $O(M \log M)$ with $M = \prod_i n_i$. The dihedral case, although non-Abelian, has comparable cost: the matched transform of $D_{2M}$ on the even-reflected extension is the discrete cosine transform, for which $O(M \log M)$ algorithms are standard~\cite{thornton2026unification}. In the Abelian case, therefore, algebraic diversity enjoys the full computational advantage of the FFT and carries essentially no computational overhead relative to the classical estimator it replaces.

For non-Abelian groups, the reduction is less favorable. Exact computation via the Peter-Weyl decomposition has cost proportional to the sum of cubed irreducible representation dimensions, $\sum_i d_i^3$, which is in general not $O(M \log M)$ and may be as costly as $O(M^3)$ in the worst case. A tractable approximation is available for applications that can tolerate a controlled information loss, and we refer to it as the \emph{Abelianization trick}. The idea is to replace the full non-Abelian group $G$ by its largest Abelian quotient $G^{\mathrm{ab}} = G / [G, G]$, where $[G, G]$ denotes the commutator subgroup (the smallest normal subgroup whose quotient is Abelian), and to compute the group-averaged estimator on $G^{\mathrm{ab}}$ at the FFT cost $O(M \log M)$. The information discarded by this substitution is exactly the content of the irreducible representations of $G$ whose dimension $d_i$ is greater than one, because these are the representations that do not descend to the Abelian quotient. The loss is quantified by the commutativity residual $\delta([G, G], \mathbf{R})$ evaluated on the commutator subgroup (see the companion paper~\cite{thornton2026ad} for the definition and properties of $\delta$).

For graph-structured data whose automorphism groups are close to Abelian (as is common on regular lattices, rings, and trees), the loss is small in practice, and the abelianization gives a practically useful approximation at FFT cost. For highly non-Abelian structures such as the Petersen graph with $\Aut(\mathcal{G}) \cong S_5$, the abelianization discards most of the rich permutation structure and a full Peter-Weyl computation is warranted. In the intermediate regime, $\delta([G, G], \mathbf{R})$ provides a quantitative measure of how much information the abelianization sacrifices, and the practitioner can use this quantity to decide whether the approximation is acceptable for a given application.

\section{Blind Group Matching}
\label{sec:blindmatching}

With the framework and its main applications to rank-promoted data and compound-symmetry observations laid out, we now turn to the algorithmic question that all of the preceding development depends upon: how is the matched group identified when it is not known a priori? This section and those that follow treat the methodology, its convergence properties, and its extension to blind signal processing generally. The register of the remaining sections shifts to a more compressed technical style appropriate for algorithmic and theorem-proof material; the longer expository register of Sections~\ref{sec:rankpromote} and~\ref{sec:eigentensor} was adopted there to establish the conceptual apparatus, whereas the sections that follow develop algorithmic and theoretical machinery for which a denser presentation is better matched to the content.

When prior knowledge of the signal's algebraic symmetry is absent, the group must be selected from the data. We describe the evolution of the blind group-matching methodology, moving from an earlier spectral-concentration approach whose limitations we identify to an improved cross-validation criterion, and thence to a continuous relaxation that returns the matched group in polynomial time via an eigenvalue computation on the Lie algebra $\mathfrak{u}(M)$.

\subsection{Earlier methods: spectral concentration}
\label{sec:psi}

Earlier methods used the \emph{spectral concentration} metric $\hat{\psi}(G) = \lambda_{\max}(\hat{\mathbf{R}}_G) / \Tr(\hat{\mathbf{R}}_G)$ for the blind group matching problem, scoring each candidate group $G$ by how concentrated its group-averaged estimator's spectrum was on the dominant eigenvalue; this was found to be a poor choice due to its bias (defined and analyzed in~\cite{thornton2026ad}). Selection was by maximization of $\hat{\psi}$ across a library of candidate groups.

The spectral concentration criterion suffers from an \emph{orbit-size bias}. Groups with smaller orbits (smaller $|G|$, or equivalently lower-dimensional irreducible representations) produce more concentrated spectra almost by construction, independent of whether they are matched to the signal. For structured signals, $\hat{\psi}(\mathbb{Z}_2^k)$ may exceed $\hat{\psi}(\mathbb{Z}_{2^k})$ even when $\mathbb{Z}_{2^k}$ is the matched group. The bias becomes more severe as $L$ grows, because larger $L$ drives $\hat{\psi}$ toward its population value, which preserves the bias. An earlier universal form of the $\hat{\psi}$-blind conjecture must therefore be qualified: $\hat{\psi}$ is a valid selector \emph{within} a family of groups with the same orbit structure (for example, among conjugated copies of the same group, or within a chirp-rate sweep), but it is not a valid selector across families with different orbit structures. The mechanism is geometric: $\hat{\psi}$ is the $\ell_\infty$ peak of the matched power spectrum (Section~\ref{sec:diagoffdiag}), and a peak statistic is sensitive to how that spectrum is partitioned into orbit blocks, so a finer orbit structure raises the peak independently of whether the group is matched.

\subsection{The improved estimator: cross-validation criterion $D_{CV}$}
\label{sec:dcv}

The bias identified in the spectral-concentration criterion motivates a criterion based on a different principle: rather than measuring how concentrated the group-averaged estimator's spectrum is for a given group, measure how consistent the estimator is across independent realizations of the same signal. A matched group produces a group-averaged estimator whose spectral structure is reproducible from one realization to the next (up to noise), whereas a mismatched group imposes structure that does not reproduce. We refer to the resulting quantity as the \emph{cross-validation residual}, defined by
\begin{equation}\label{eq:dcv}
D_{CV}(G) \;=\; \mathbb{E}_{\ell, \ell'}\!\left[\, \bigl\| \hat{\mathbf{R}}_G^{(\ell)} - \hat{\mathbf{R}}_G^{(\ell')} \bigr\|_F^2 \,\right],
\end{equation}
measured on pairs of independent snapshots $\mathbf{x}_\ell, \mathbf{x}_{\ell'}$, each estimated under the candidate group $G$. Small $D_{CV}(G)$ indicates that repeated estimates under $G$ are internally consistent, which is a direct test of whether $G$ matches the signal's structure. The criterion is independent of orbit size, and in simulation it selects the correct group from libraries of sizes ranging from a few to all Abelian groups of the relevant order at 99\% or better accuracy at 20~dB SNR with $L = 3$ snapshots.

$L = 2$ is usually sufficient for $D_{CV}$; if lower estimator error is desired, additional observations may be collected and the error is further reduced in accordance with the law of large numbers. This modest snapshot requirement is the price paid for blind identification; once the matched group has been identified, subsequent processing can proceed with $L = 1$ and the framework's usual single-snapshot algebraic diversity apparatus.

\subsection{The matched covariance: diagonal and off-diagonal diagnostics}
\label{sec:diagoffdiag}

The diagnostics used for blind matching, and the structural capacity of Section~\ref{sec:fourthms}, organize naturally once the covariance is written in the group-adapted basis. Let $\mathbf{U}_G$ be the matched transform of $G$, the unitary that diagonalizes the commutant of $G$ (the DFT for the cyclic group, and the Karhunen--Lo\`{e}ve basis whenever $G$ is matched to the signal), and write
\begin{equation}\label{eq:matched_basis_cov}
\mathbf{M} \;=\; \mathbf{U}_G^H\, \mathbf{R}\, \mathbf{U}_G .
\end{equation}
The group-averaged estimator $\hat{\mathbf{R}}_G$ is the Reynolds projection of $\mathbf{R}$ onto the commutant of $G$, which in this basis keeps the diagonal of $\mathbf{M}$ and discards everything off it. The diagonal $\diag(\mathbf{M})$ is the matched power spectrum; the off-diagonal is the commutant residual, the part of $\mathbf{R}$ that does not commute with $G$. The diagnostics partition along exactly this split.

On the diagonal sit the structural summaries. The spectral concentration $\hat{\psi} = \lambda_{\max}(\hat{\mathbf{R}}_G)/\Tr(\hat{\mathbf{R}}_G)$ is the dominant-bin fraction of the matched spectrum, an $\ell_\infty$ (peak) functional of $\diag(\mathbf{M})$, whereas the structural capacity evaluated on the same estimator, $\kappa(\hat{\mathbf{R}}_G) = (\Tr \hat{\mathbf{R}}_G)^2 / \|\hat{\mathbf{R}}_G\|_F^2$, is the $\ell_2$ participation of that diagonal; when $G$ is matched the eigenvalues of $\mathbf{R}$ coincide with $\diag(\mathbf{M})$, so the structural capacity of Section~\ref{sec:fourthms} is the $\ell_2$ companion to the $\ell_\infty$ peak $\hat{\psi}$. The two report one object through different norms, which is why $\hat{\psi}$ inherits the block-size sensitivity of a peak statistic, the orbit-size bias of Section~\ref{sec:psi}, while $\kappa$, using the whole diagonal, does not. The coloring index splits along the same axis: with $\bar{q} = \Tr(\mathbf{R})/M$, unitary invariance of the Frobenius norm gives
\begin{equation}\label{eq:alpha_decomp}
\alpha(\mathbf{R})^2 \;=\; \frac{\|\diag(\mathbf{M}) - \bar{q}\,\mathbf{I}\|_F^2}{\|\mathbf{R}\|_F^2} \;+\; \frac{\|\mathrm{offdiag}(\mathbf{M})\|_F^2}{\|\mathbf{R}\|_F^2},
\end{equation}
a diagonal coloring leg plus an off-diagonal leg, the latter being the squared distance of $\mathbf{R}$ from the commutant that the residual $\delta$ also tracks. Thus $\alpha$ combines how structured the matched spectrum is with how far the signal is from commuting with $G$.

On the off-diagonal, the commutativity residual $\delta(G, \mathbf{R})$ measures the \emph{size} of the residual but not its content. Its content is read by the \emph{spectral coherence index}
\begin{equation}\label{eq:chi}
\chi(G, \mathbf{R}) \;=\; \max_{i \neq j}\ \frac{|\mathbf{M}_{ij}|^2}{\mathbf{M}_{ii}\,\mathbf{M}_{jj}} \;\in\; [0, 1],
\end{equation}
taken over the active spectral bins, the off-diagonal analog of $\hat{\psi}$: where $\hat{\psi}$ reads the dominant diagonal mass, $\chi$ reads the dominant off-diagonal coupling. It vanishes when the active components are mutually incoherent and approaches one when a pair is phase-locked across the ensemble, so it distinguishes a commutant residual that signals group mismatch ($\chi \approx 0$) from one produced by genuinely coherent components ($\chi$ near unity). That distinction the size metrics cannot make, and it bears directly on the mixed-structure case of Section~\ref{sec:mixed}, where sequential stripping presumes the extracted component is independent of the remainder. For the cyclic group $\mathbf{M}$ is the bifrequency spectrum and $\chi$ coincides with the spectral coherence of cyclostationary analysis~\cite{gardner1991}; the framework recovers that quantity as the matched-group special case rather than introducing a new one.

\subsection{Signal classes where blind group matching succeeds}
\label{sec:scope}

The $D_{CV}$ criterion, together with the complementary $\hat{\kappa}$-trajectory and variance-scaling diagnostics developed in the following subsections, covers a broad range of signal classes. Extensive Monte Carlo across representative signal models demonstrates success in the following cases, each verified at 20~dB SNR with $L = 3$ snapshots unless otherwise noted:

\begin{itemize}[nosep]
\item \textbf{Periodic tones (on-grid).} Single or multiple sinusoids at integer-bin frequencies match $\mathbb{Z}_M$; $D_{CV}$ selects $\mathbb{Z}_M$ at $99\%$ or better accuracy against the full library of Abelian groups of order $M$.
\item \textbf{Periodic tones (off-grid).} Sinusoids at non-integer bins (for example, $k = 5.3$, $k = 11.7$) are handled by a Kaiser-windowed extension in which $D_{CV}$ selects the correct cyclic group with $100\%$ accuracy at 20~dB and the Kaiser parameter $\beta \approx 8$--$9$ is itself selected by $D_{CV}$.
\item \textbf{Multipath signals.} A direct path plus two reflections with independent phases is selected to $\mathbb{Z}_{32}$ with $100\%$ accuracy; the resulting estimation error at $L = 2$ is $30\times$ smaller than the sample covariance at $L = 32$.
\item \textbf{Chirps.} Linear frequency-modulated signals match a conjugated cyclic group $D_\mu \mathbb{Z}_M D_\mu^{-1}$, where $D_\mu$ is a dechirp operator parametrized by the chirp rate $\mu$. Blind identification of $\mu$ is achieved by sweeping $\hat{\psi}$ along the conjugation parameter within the single-orbit class where $\hat{\psi}$ is a valid selector.
\item \textbf{Autoregressive signals.} AR(1) through AR($n$) covariances, which are banded Toeplitz, are selected to the cyclic group $\mathbb{Z}_M$ (under a windowed DCT-like extension with $\beta \approx 12$) with $100\%$ accuracy against the Abelian library; $\hat{\kappa}$-trajectory confirms the match through a flat trajectory across $L \in \{1, \ldots, 10\}$.
\item \textbf{Colored noise.} AR(1) with $\rho = 0.8$ is handled identically to the AR($n$) case, with consistent $\mathbb{Z}_{32}$ selection and $8\times$ estimation gain.
\item \textbf{DFT-matched and DCT-matched signals.} The Kaiser-windowed cyclic family interpolates between DFT ($\beta \approx 0$) and DCT ($\beta \gg 1$) boundary conditions continuously; the optimal $\beta^*$ is selected from data by $D_{CV}$, generalizing empirical DCT-selection practice into a principled data-driven procedure.
\item \textbf{Graph-structured signals.} Signals whose covariance has the form $\mathbf{R} = f(\mathbf{L})$ for a graph Laplacian $\mathbf{L}$ with injective $f$ on the spectrum are handled by the exact commutator residual of Section~\ref{sec:dadcad}, with the Sequential GEVP recovering a subgroup $G_K \subseteq \Aut(\mathcal{G})$. For graphs whose automorphism group is the full symmetric group, including $K_4$ ($\Aut = S_4$) and $K_5$ ($\Aut = S_5$), the algorithm achieves $G_K = \Aut(\mathcal{G})$ at $100\%$ accuracy on Monte Carlo trials, since every Hungarian-rounded permutation lies in $\Aut$ by construction and the rejection step is never invoked. For graphs with a proper-subgroup automorphism group, including $C_n$ ($\Aut = D_n$) for $n \geq 4$, full recovery $G_K = \Aut(\mathcal{G})$ is obtained by the span-search form of the Sequential GEVP discussed in Section~\ref{sec:seqgevp} when the candidate basis spans a generating set; the plain single-direction procedure typically recovers only a non-trivial subgroup.
\item \textbf{Non-Abelian supergroup cases.} When the signal's natural matched group is Abelian but a non-Abelian supergroup provides additional dimensional structure, the Supergroup Dominance result of Theorem~\ref{thm:supergroup} selects the non-Abelian supergroup where it is strictly superior; in the remaining cases covered by Remark~\ref{rem:supergroup_scope} the Abelian submatch is preserved. The methodology covers both regimes.
\end{itemize}

Together these cases cover a broad range of the signal classes encountered in routine spectral-estimation practice. The success is underwritten by two distributional regularities whose structure we develop next: the variance-scaling dichotomy of Section~\ref{sec:scaling} and the $\hat{\kappa}$-trajectory test that follows it.

\subsection{The variance-scaling dichotomy as distributional fingerprint}
\label{sec:scaling}

The empirical reliability of $D_{CV}$ rests on a deeper distributional regularity. The five AD sample diagnostics $\hat{\delta}$ (commutativity residual), $\hat{\psi}$ (spectral concentration), $\hat{\kappa}$ (structural capacity), $\hat{\alpha}$ (coloring index), and $\hat{r}_{\mathrm{eff}}$ (effective rank), defined in the companion paper~\cite{thornton2026ad} and organized by the diagonal/off-diagonal geometry of Section~\ref{sec:diagoffdiag}, exhibit a sharp dichotomy in their variance scaling with SNR, depending on whether the candidate group is matched to the signal:

\begin{proposition}[Variance scaling dichotomy]\label{prop:scaling}
The standard deviation of each AD sample estimator follows power-law scaling $\sigma(\hat{m}) \sim c / \mathrm{SNR}^{\beta}$. For the matched group, $\beta$ is strictly positive and the variance vanishes as SNR grows (convergent variance). For mismatched groups, $\beta \approx 0$ and the variance remains constant, independent of SNR.
\end{proposition}

\begin{table}[t]
\centering
\caption{Variance scaling exponents $\beta$ and prefactors $c$ for the four AD sample diagnostics, matched and mismatched, measured on a representative two-tone signal as a matched instance (the tone class of Section~\ref{sec:scope}) with $\mathbb{Z}_M$ as the matched candidate group and $\mathbb{Z}_2^{\log_2 M}$ as a mismatched candidate group of the same order. Matched-group variances vanish as $\mathrm{SNR} \to \infty$; mismatched-group variances saturate at a finite floor set by the structural mismatch. The dichotomy is the distributional fingerprint that underwrites $D_{CV}$.}
\label{tab:scaling}
\renewcommand{\arraystretch}{1.2}
\begin{tabular}{lcccc}
\toprule
\textbf{Estimator} & \multicolumn{2}{c}{\textbf{Matched}} & \multicolumn{2}{c}{\textbf{Mismatched}} \\
 & $c$ & $\beta$ & $c$ & $\beta$ \\
\midrule
$\hat{\kappa}$           & 6.08  & 1.08 & 2.35  & $+0.07$ \\
$\hat{\psi}$             & 0.066 & 0.48 & 0.084 & $-0.02$ \\
$\hat{r}_{\mathrm{eff}}$ & 7.42  & 0.90 & 2.51  & $+0.06$ \\
$\hat{\alpha}$           & 0.100 & 0.84 & 0.045 & $-0.06$ \\
\bottomrule
\end{tabular}
\end{table}

Table~\ref{tab:scaling} reports the measured exponents on the two-tone instance, representative of the periodic-tone class where the $D_{CV}$ blind matching procedure succeeds with the cyclic group $\mathbb{Z}_M$ as the matched candidate. The matched-group entries all exhibit clean power-law decay ($\beta$ between roughly $1/2$ and $1$); the mismatched-group entries are statistically indistinguishable from $\beta = 0$. The mechanism is direct: the matched group's averaging integrates out the random signal phases, leaving residual variance set by additive noise (which decays with SNR); the mismatched group's averaging cannot integrate the phases out, leaving phase-dependent artifacts that persist regardless of SNR. The dichotomy is the finite-sample analog of the consistency-versus-inconsistency distinction in classical estimation: matched groups produce consistent estimators, mismatched ones produce inconsistent estimators with irreducible variance. Analogous power-law behaviour has been observed on each of the matched classes listed in Section~\ref{sec:scope}; the two-tone numbers are illustrative of the class as a whole.

A complementary diagnostic is the \emph{$\hat{\kappa}$ trajectory} as $L$ varies from 1 to 10: matched groups produce a flat trajectory, while mismatched groups produce a monotonically rising trajectory. The directional shift $\mathrm{sign}(\hat{\kappa}(2) - \hat{\kappa}(1))$ separates matched from mismatched groups cleanly in the signal classes tested, and in the experiments of Section~\ref{sec:scope} provided the largest matched-to-mismatched separation among the five diagnostics. This trajectory test runs alongside $D_{CV}$ in the methodology summary of Section~\ref{sec:methodology}.

\subsection{Mixed-structure signals: partially solved by iterative stripping}
\label{sec:mixed}

The most demanding case is the composite signal. Signals composed of multiple components with \emph{distinct} matched groups, for example a sum of two tones with a chirp and an AR(1) process, present a problem that is not a failure of the methodology but a fundamental question about what ``the matched group'' means for such a composite. No single finite group simultaneously matches a tone component ($\mathbb{Z}_M$), a chirp component ($D_\mu \mathbb{Z}_M D_\mu^{-1}$ with different conjugation), and an AR($n$) component (Kaiser-windowed cyclic with its own $\beta^*$). Running $D_{CV}$ on the composite signal returns the product group $\mathbb{Z}_{M/2} \times \mathbb{Z}_2$ as a compromise candidate, and the $\hat{\psi}$ sweep sees only the dominant tone component ($\hat{\psi} \approx 0.70$ on the tone axis versus $\hat{\psi} \approx 0.14$ on the chirp axis); the composite structure does not project cleanly onto any single-group axis.

The difficulty just described is specific to the one- or two-snapshot $D_{CV}$ pipeline: $D_{CV}$ is economical precisely because it uses all $M^2$ entries of a covariance estimate formed from as few as two snapshots, but a single group-averaged covariance of a composite signal is a poor model of the whole, so the criterion settles on a product-group compromise rather than resolving the components. Of the three partial strategies originally proposed, namely subspace-then-AD (estimate one component, remove it, blind-match the residual), iterative stripping (sequential group matching with group-theoretic deflation applied across component classes rather than across generators of a fixed group), and product-group matching (search over $G_1 \times G_2$), the second now resolves the case within stated constraints. Iterative stripping, realized as an algebraic-diversity analog of the CLEAN algorithm, selects a group, projects out its coherent component, and recurses on the residue until no group concentrates the remainder. The component that survives is the algebraic residue of Section~\ref{sec:residue}. Concretely, given an observation $\mathbf{x}$ and a candidate library $\mathcal{L}$, the procedure is the following.
\begin{enumerate}[label=(\arabic*),leftmargin=2.2em,itemsep=1pt,topsep=2pt]
\item Initialize the working residue $\mathbf{r} \leftarrow \mathbf{x}$.
\item Form the covariance of $\mathbf{r}$ and select the group $G \in \mathcal{L}$ of smallest symmetry defect $\delta(G, \mathbf{r})$ that retains a coherent eigen-subspace; if no nontrivial group concentrates $\mathbf{r}$, stop.
\item Project $\mathbf{r}$ onto the coherent subspace of $\mathbf{R}_G$ (hard or soft Wiener gain) to obtain the structured component $\mathbf{y} = \mathbf{P}_G \mathbf{r}$.
\item If the energy of $\mathbf{y}$ falls below a minimum-peel threshold, stop; otherwise record $\mathbf{y}$, update $\mathbf{r} \leftarrow \mathbf{r} - \mathbf{y}$, and return to step (2).
\end{enumerate}
The recorded components are the structured parts, each in its own group-adapted basis, and the final $\mathbf{r}$ is the algebraic residue. The first-peel selection in step (2) uses the symmetry defect rather than the held-out likelihood, for the reason given below. With this procedure a composite of a chirp and an AR(1) process in noise, the case on which the $D_{CV}$ pipeline returns only the dominant component, separates into its two parts, each recovered in its own group-adapted basis.

The first-peel selection criterion is the point at which the precise tradeoff appears. The held-out Gaussian likelihood that selects the matched group for a single-structure signal is the wrong criterion for the first peel of a composite, and increasingly so at high signal-to-noise ratio: a group-averaged covariance that captures only one component is a poor full-covariance model of the whole, so the likelihood prefers the trivial group and the procedure terminates prematurely. The correct first-peel criterion is the symmetry defect $\delta(G, \mathbf{R})$ itself, which asks which group's structure is \emph{present} rather than which group best models the entire signal; the held-out likelihood is retained only to confirm a single-structure residue and to set the stopping point. The cost of the resolution is data. The economical $D_{CV}$ criterion operates from one or two snapshots, whereas iterative stripping forms and validates a group-averaged covariance at each peel, which consumes more snapshots than the two-shot budget: the defect-driven peel requires a reliable per-peel covariance estimate, and the likelihood-based confirmation requires a held-out split on top of it. The mixed-structure case is therefore no longer wholly open but partially solved: components with \emph{distinct} matched groups, each carrying sufficient energy, above a soft signal-to-noise-ratio floor, are separated by iterative stripping at the expense of the two-snapshot economy of $D_{CV}$; components sharing a matched group are recovered only in aggregate, and the uniqueness and order-independence of the decomposition when component subspaces overlap remains open. The $\hat{\kappa}$ trajectory and variance-scaling diagnostics of Section~\ref{sec:scaling} remain reliable indicators of \emph{whether} a candidate group is a complete match and therefore continue to detect the mixed-structure condition.

\subsection{Peel ordering: extracted energy versus residue whitening}
\label{sec:peel-order}

The stripping procedure of Section~\ref{sec:mixed} removes one structured component per iteration, and the order-independence of the result when component subspaces overlap was noted there as open. The symmetry defect decides which structure is present and when to stop; a separate question is the objective by which the present components are \emph{ordered} for removal. The default objective is classical CLEAN ordering, which peels at each iteration the matched group whose coherent eigen-subspace removes the most power from the residue. On overlapping components a greedy maximum-energy peel can sweep up power that belongs to a neighboring component, a cross-talk leakage that distorts both the peeled component and what remains. This motivates an alternative objective: peel the group whose removal leaves the whitest residue, measured by the structural capacity $\kappa(\mathbf{R}) = (\Tr \mathbf{R})^2/\|\mathbf{R}\|_F^2$ of Section~\ref{sec:diagoffdiag}, which is maximal for an isotropic spectrum. Equivalently this peels the most spectrally concentrated structure first. The hypothesis is that this residue-whitening ordering reduces cross-talk on overlapping components, and that the better ordering is governed by the noise floor and the separability of the structure rather than being fixed.

We compose multi-component scenes and peel each with both orderings, sweeping the signal-to-noise ratio from 0 to 24 decibels. The candidate library is the full finite-group library used throughout (cyclic, dihedral, Cartesian product, quotient, Weyl-Heisenberg, and wreath groups), augmented, in scenes that contain a chirp, by the matched metaplectic chirp group; decoy groups are excluded so both orderings choose among legitimate matched groups only. The reconstruction error is the relative error $\|\hat{\mathbf{x}} - \mathbf{x}_{\mathrm{clean}}\| / \|\mathbf{x}_{\mathrm{clean}}\|$ between the structured estimate and the noiseless mixture, with both orderings run under identical settings. The scenes are an AR(1) process at two-thirds the chirp power; the same with an added tone; a binary phase-shift-keyed signal with an amplitude-modulated band and a tone (no metaplectic structure); a comparable tone and chirp; a strong AR(1) with a weak tone; and a quadrature signal with a binary signal and a tone. Table~\ref{tab:peelorder12} reports the operating point at 12 decibels, and Figure~\ref{fig:peelsweep} reports the full sweep.

\begin{table*}[t]
\centering
\caption{Reconstruction error at 12 decibels, ten trials per scene, mean $\pm$ standard deviation. ``BMG gain'' is the first-iteration best-matched-group log-likelihood improvement over the trivial group, a proxy for how easily the structure is discriminated. At this moderate noise floor the orderings are close for the overlapping continuous scenes, and residue whitening is best on the AR(1) plus chirp scene.}
\label{tab:peelorder12}
\small
\renewcommand{\arraystretch}{1.15}
\begin{tabular}{lccccl}
\toprule
scene & metaplectic & energy-first & whitening-first & BMG gain & best \\
\midrule
AR(1) $+$ chirp            & yes & $0.265 \pm 0.007$ & $0.255 \pm 0.006$ & 2495 & whitening \\
AR(1) $+$ chirp $+$ tone   & yes & $0.258 \pm 0.003$ & $0.273 \pm 0.016$ & 1156 & energy \\
BPSK $+$ AM-DSB $+$ tone   & no  & $0.166 \pm 0.001$ & $0.353 \pm 0.042$ & 7624 & energy \\
tone $+$ chirp (equal)     & yes & $0.155 \pm 0.002$ & $0.160 \pm 0.004$ & 3854 & energy (close) \\
AR(1) strong $+$ weak tone & no  & $0.281 \pm 0.006$ & $0.287 \pm 0.011$ & 11598 & tie \\
QPSK $+$ BPSK $+$ tone     & no  & $0.130 \pm 0.002$ & $0.335 \pm 0.020$ & 4248 & energy \\
\bottomrule
\end{tabular}
\end{table*}

\begin{figure*}[t]
\centering
\includegraphics[width=\textwidth]{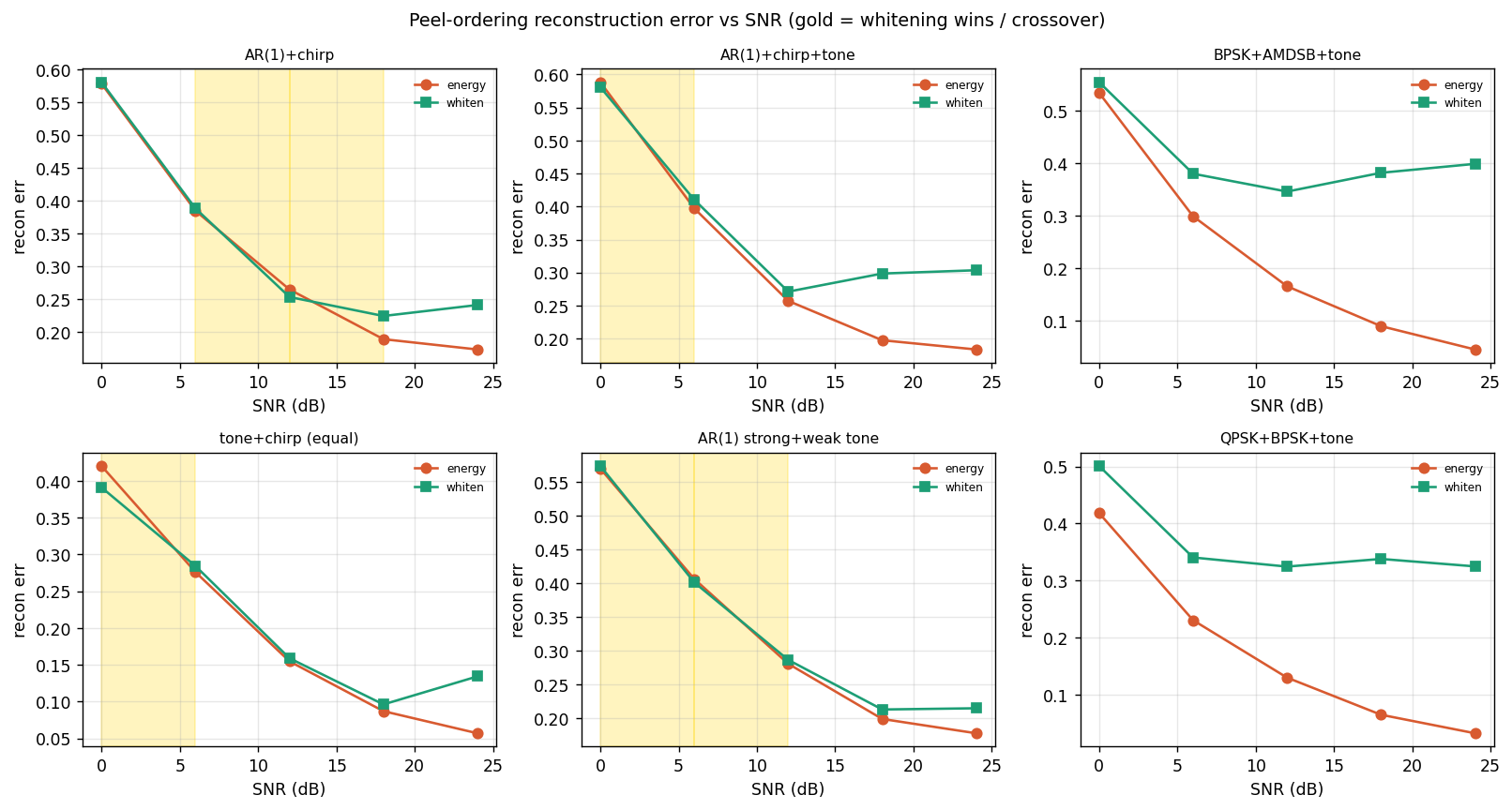}
\caption{Reconstruction error versus signal-to-noise ratio for the two peel orderings. Shaded bands mark the range in which residue whitening wins or crosses over. The continuous overlapping scenes (top-left, top-center, bottom-left, bottom-center) show a low-to-moderate-noise window in which whitening is competitive or superior and a high-signal regime in which energy ordering wins; the digital finite-group scenes (top-right, bottom-right) favor energy ordering at every noise level, with the gap widening as the noise falls.}
\label{fig:peelsweep}
\end{figure*}

The sweep shows that the better ordering depends on the operating point rather than being fixed. At low signal-to-noise ratio the two orderings converge, because the structured components are buried in noise and removing the most power and flattening the residue most become nearly the same operation. As the noise falls, the two orderings separate, and the direction of the separation depends on the structure. For the continuous overlapping scenes, those built from a chirp, a colored AR(1) process, and a tone, there is a window roughly between 6 and 15 decibels in which residue whitening is competitive or superior: the AR(1) plus chirp scene is whitening-favored near 12 decibels, and the equal-power tone and chirp scene is whitening-favored at 0 decibels. Above about 15 decibels energy ordering takes the lead in these scenes as well. For the digital finite-group scenes, the binary and quadrature signals, energy ordering is better at every noise level and its margin grows as the noise falls. The high-signal failure of residue whitening has a clear mechanism: with a rich library the whitening objective is captured by large groups, the Weyl-Heisenberg and wreath groups, whose coherent subspaces remove broad portions of the residue and flatten it most even when they match no true component. Restricting the candidate set to the truly matched groups removes most of this penalty at high signal-to-noise ratio (Figure~\ref{fig:peelsweep} is for the full library; a curated library narrows the high-signal gap on the AR(1) plus chirp scene), confirming that candidate-library richness is part of the decision as well.

These results revise the earlier conclusion that energy ordering is uniformly preferable, which held only at high signal-to-noise ratio with a rich library. The operative picture is a decision rule: energy ordering is the robust choice when the structure stands well above the noise floor and is cleanly separable (high signal-to-noise ratio, or digital signals with distinct finite-group structure that the best-matched-group test discriminates easily), while residue whitening is competitive or superior when the noise floor is high and the components overlap in their eigen-structure (a chirp lying across a colored process), the regime that is hardest for group discrimination. The two orderings are inexpensive to run together, so near the crossover, where the noise floor is moderate and the components are poorly separated, the practical recommendation is to compute both orderings and retain the reconstruction with the smaller residual structure. Energy ordering remains the default; residue whitening earns its place as the alternative to try when the operating point is noisy and the scene is hard to separate.

\subsection{The DAD--CAD bridge: continuous relaxation via a Lie-algebra GEVP}
\label{sec:dadcad}

Blind group matching across a large library is combinatorial in the library size. A continuous relaxation reduces the problem to a polynomial-time eigenvalue computation on a finite-dimensional Lie algebra, which we call the \emph{DAD--CAD bridge}, abbreviating \emph{Discrete Algebraic Diversity} (the discrete group-library formulation of Sections~\ref{sec:psi}--\ref{sec:dcv}) and \emph{Continuous Algebraic Diversity} (the continuous formulation we now develop on directions in the Lie algebra $\mathfrak{u}(M)$). The relaxation and its optimality and complexity properties are developed in full in a companion paper~\cite{thornton2026gevp}; we summarize here the elements needed for the framework. The bridge is the passage from the discrete library to the continuous parameter space: discrete algebraic diversity (selection from a discrete group library) relaxes to continuous algebraic diversity in the sense of optimization over directions in $\mathfrak{u}(M)$, the Lie algebra of skew-Hermitian matrices on $\mathbb{C}^M$. The qualifier \emph{finite-dimensional} is important: the relaxation here is a continuous parameter space attached to a fixed-dimensional ambient space, distinct from extensions of AD to Lie groups acting on infinite-dimensional function spaces such as $L^2(\mathbb{R})$~\cite{thornton2026continuous}, which lie outside the scope of this paper.

Replace the discrete commutativity residual $\delta(G, \mathbf{R}) = \max_{g \in G} \|\pi_g \mathbf{R} - \mathbf{R} \pi_g\|_F / \|\mathbf{R}\|_F$ (see~\cite{thornton2026ad} for definition and properties) by its continuous analog
\begin{equation}\label{eq:delta_cont}
\delta(A, \mathbf{R}) \;=\; \frac{\|[A, \mathbf{R}]\|_F}{\|A\|_F \, \|\mathbf{R}\|_F}, \quad A \in \mathfrak{u}(M).
\end{equation}
Minimizing $\delta(A, \mathbf{R})$ over $A$ in a finite-dimensional subspace $\mathcal{B} = \mathrm{span}\{B_1, \ldots, B_d\} \subseteq \mathfrak{u}(M)$ is equivalent to a generalized eigenvalue problem (GEVP; an eigenvalue problem of the form $\mathbf{K} \mathbf{v} = \lambda \mathbf{M} \mathbf{v}$ with two matrices $\mathbf{K}$ and $\mathbf{M}$, reducing to the ordinary eigenvalue problem when $\mathbf{M} = \mathbf{I}$) involving the double-commutator superoperator $\mathrm{ad}^2_{\mathbf{R}}(A) = [\mathbf{R}, [\mathbf{R}, A]]$.

\begin{theorem}[Double-commutator GEVP]\label{thm:dcgevp}
Let $\mathbf{R}$ be Hermitian positive definite, and let $\mathcal{B} = \mathrm{span}\{B_1, \ldots, B_d\} \subseteq \mathfrak{u}(M)$ be a finite-dimensional subspace of skew-Hermitian matrices. The element $A^* \in \mathcal{B}$ minimizing $\|[\mathbf{R}, A^*]\|_F / \|A^*\|_F$ is the generalized eigenvector of smallest eigenvalue in
\begin{equation}
\mathbf{M} \mathbf{c} \;=\; \lambda\, \mathbf{G} \mathbf{c},
\end{equation}
with $M_{ij} = \Tr\bigl(B_i^H [\mathbf{R}, [\mathbf{R}, B_j]]\bigr)$ and $G_{ij} = \Tr(B_i^H B_j)$. Solution cost is $O(d^2 M^2 + d^3)$.
\end{theorem}

The theorem converts the combinatorial group-library search into a standard linear-algebraic operation. The double-commutator form arises because the matrices involved are Hermitian, so the quartic cost in $A$ reduces, through the Frobenius norm's cyclic trace property and the skew-Hermiticity of $A$, to the quadratic form above. This is the specific technical advantage that makes the continuous relaxation tractable.

A closed-form expression underwrites the entire bridge:

\begin{lemma}[Commutator decomposition in $\mathbf{R}$'s eigenbasis]\label{lem:commdecomp}
For any $\mathbf{A} \in \mathbb{C}^{M \times M}$ and any Hermitian $\mathbf{R}$ with eigendecomposition $\mathbf{R} = \mathbf{U} \diag(\lambda_1, \ldots, \lambda_M) \mathbf{U}^H$, with $\mathbf{A}$ expressed in $\mathbf{R}$'s eigenbasis as $\tilde{\mathbf{A}} = \mathbf{U}^H \mathbf{A} \mathbf{U}$,
\begin{equation}\label{eq:commdecomp}
[\mathbf{A}, \mathbf{R}]_{jk} = (\lambda_k - \lambda_j) \tilde{A}_{jk}, \qquad \|[\mathbf{A}, \mathbf{R}]\|_F^2 = \sum_{j \neq k} (\lambda_j - \lambda_k)^2 |\tilde{A}_{jk}|^2.
\end{equation}
\end{lemma}

The lemma states that the commutator depends only on the off-diagonal elements of $\mathbf{A}$ in $\mathbf{R}$'s eigenbasis, weighted by eigenvalue gaps. The diagonal elements of $\tilde{\mathbf{A}}$ are invisible to the commutator. Two specializations follow immediately. For a permutation $\sigma$ of the eigen-indices of $\mathbf{R}$, that is, for $\mathbf{P}_\sigma = \mathbf{U}\,\boldsymbol{\Pi}_\sigma \mathbf{U}^H$ with $\boldsymbol{\Pi}_\sigma$ the permutation matrix of $\sigma$, the eigenvalue-gap weights collapse onto a permutation of indices and
\begin{equation}\label{eq:permcomm}
\|[\mathbf{P}_\sigma, \mathbf{R}]\|_F^2 \;=\; \sum_{k=1}^M \bigl(\lambda_k - \lambda_{\sigma(k)}\bigr)^2,
\end{equation}
the \emph{eigenvalue-difference formula} for the permutation commutativity residual. A permutation of the sample indices is in general not of this form, and its residual is given by the lemma itself, $\sum_{j \neq k}(\lambda_j - \lambda_k)^2 |\tilde{P}_{jk}|^2$ with $\tilde{\mathbf{P}} = \mathbf{U}^H \mathbf{P}_\sigma \mathbf{U}$, which vanishes exactly when $\mathbf{P}_\sigma$ maps each eigenspace of $\mathbf{R}$ to itself. For graph-diffusion covariance $\mathbf{R} = f(\mathbf{L})$ with $f$ injective on the spectrum, this residual is exactly zero if and only if the vertex permutation $\sigma$ is a graph automorphism, so the double-commutator GEVP is an exact algebraic oracle for graph symmetry detection. For a general skew-Hermitian generator $A \in \mathfrak{u}(M)$, equation~\eqref{eq:commdecomp} expresses $\delta(A, \mathbf{R})$ as a weighted off-diagonal energy on $\tilde{A}$, which is the form used in the GEVP construction of Theorem~\ref{thm:dcgevp}.

\subsection{The transform manifold revisited}
\label{sec:manifold_revisit}

The continuous relaxation is the natural setting in which the transform manifold of Section~\ref{sec:manifold} lives. The skew-Hermitian Lie algebra $\mathfrak{u}(M)$ is the tangent space to the unitary group $U(M)$ at the identity; distinguished directions in $\mathfrak{u}(M)$ correspond to specific transforms via exponentiation. The Fourier basis corresponds to the cyclic-shift generator, the DCT basis corresponds to the cyclic-shift generator of the even-reflected extension together with its reflection~\cite{thornton2026unification}, and the KL transform corresponds, through the Cayley graph of $S_M$, to a generic direction tied to the signal's covariance structure. The DAD--CAD bridge is the passage from a discrete group library to a continuous search on this manifold, and the double-commutator GEVP is the tractable finite-dimensional restriction of that continuous search.

\subsection{Non-Abelian extension: the Sequential GEVP and the multi-generator problem}
\label{sec:seqgevp}

A single minimizing direction $A^*$ identifies a single generator. Non-Abelian groups, generated by multiple non-commuting elements, require a sequential procedure. The naive joint optimization over $r$-tuples of generators is hard because the group presentation (which relations the generators must satisfy) is itself part of the unknown: one cannot impose group relations as constraints when the target group is not yet identified.

The \emph{Sequential GEVP with group-theoretic deflation} addresses the multi-generator problem by constructing the group one generator at a time. At each step it solves the GEVP of Theorem~\ref{thm:dcgevp} on a deflated basis, rounds the resulting $A^*$ to the nearest permutation matrix $P^*$ via linear assignment (the Hungarian algorithm), tests whether $\delta(P^*, \mathbf{R})$ is below a threshold $\tau$, and if so extends the accumulated subgroup $G \leftarrow \langle G, P^* \rangle$ and deflates the candidate basis against the span of permutation matrices in $G$ before re-solving. The procedure terminates when the rejection test fires or the deflated basis is empty. Four properties of the procedure can be established: every accepted permutation produces a strictly larger subgroup (forward progress); the order of the discovered subgroup at least doubles per accepted iteration (strict subgroup growth); the procedure terminates in at most
\begin{equation}
K \;\leq\; \lceil \log_2 |G_K| \rceil \;\leq\; \lceil \log_2 M! \rceil \;=\; O(M \log M)
\label{eq:iterbound}
\end{equation}
accepted iterations; and at threshold $\tau = 0$ the discovered subgroup satisfies
\begin{equation}
G_K \;\subseteq\; \Aut(\mathbf{R}),
\label{eq:gen-conv}
\end{equation}
where $\Aut(\mathbf{R}) := \{\sigma \in S_M : \mathbf{P}_\sigma \mathbf{R} = \mathbf{R} \mathbf{P}_\sigma\}$. The cost per iteration is dominated by the GEVP solve and the deflation update at $O(d^2 M^3 + d^3 + |G_k|^2 M^2)$, giving worst-case total cost polynomial in $d$, $M$, and $|G_K|$.

The inclusion~\eqref{eq:gen-conv} is one-sided: it guarantees that every accepted permutation is a genuine automorphism, but it does \emph{not} guarantee that $G_K$ recovers all of $\Aut(\mathbf{R})$. The gap arises because the rounding step operates on the deflation residual $A^*$ rather than on a candidate permutation matrix directly. When $\Aut(\mathbf{R}) = S_M$, every rounded permutation lies in $\Aut$ and the rejection test never fires; the procedure recovers all of $\Aut(\mathbf{R})$ trivially. When $\Aut(\mathbf{R})$ is a proper subgroup of $S_M$, the residual after deflating against a previously discovered subgroup $G_k$ may round to a permutation outside $\Aut$, terminating the procedure prematurely. A representative example is $\mathbf{R} = (\mathbf{I} + \mathbf{L}_{C_6})^{-1}$ with $\Aut(\mathbf{R}) = D_6$ (order $12$): with the lean generating basis $\{\mathbf{P}_\tau - \mathbf{I}, \mathbf{P}_\rho - \mathbf{I}\}$ where $\tau$ is the cyclic shift and $\rho$ a reflection, the plain procedure, which examines only the single smallest-eigenvalue direction at each step, accepts $\tau$ and then stalls, terminating at $G_K = \langle\tau\rangle \subsetneq D_6$. The stall is an artifact of examining one direction rather than the candidate span: the reflections that complete $D_6$ are reachable by rounding other directions in $\mathrm{span}\{\mathbf{P}_\tau - \mathbf{I}, \mathbf{P}_\rho - \mathbf{I}\}$, and a search-augmented step that, when the smallest direction rounds to a non-automorphism, searches the span for any direction whose rounding is a genuine automorphism extending the current subgroup, recovers the full order-twelve $D_6$. This is the same move that resolves the degenerate case below: let the group test, rather than the eigenvalue ordering alone, select the generator. The recovery is partial in the precise sense that it succeeds whenever the candidate span contains directions rounding to a generating set of $\Aut(\mathbf{R})$; a span too lean to reach a generating set, for example $\mathrm{span}\{\mathbf{P}_\tau - \mathbf{I}\}$ alone, recovers only $\langle\tau\rangle$, and no search over that span can do better. The remaining open part is therefore the a~priori selection of a sufficiently rich basis, not the recovery procedure given one.

The status of the multi-generator problem under the Sequential GEVP is summarized in Table~\ref{tab:gap2}.

\begin{table}[t]
\centering
\caption{Status of the multi-generator problem (Gap~2) for non-Abelian blind group matching, under the Sequential GEVP with group-theoretic deflation. The procedure always returns a subgroup of $\Aut(\mathbf{R})$ (soundness); whether the returned subgroup equals $\Aut(\mathbf{R})$ (completeness) depends on the case.}
\label{tab:gap2}
\renewcommand{\arraystretch}{1.15}
\setlength{\tabcolsep}{4pt}
\footnotesize
\begin{tabular}{@{}p{0.40\columnwidth}ll@{}}
\toprule
\textbf{Case} & \textbf{Status} & \textbf{Mechanism} \\
\midrule
Abelian & Closed & Direct GEVP \\
$\Aut(\mathbf{R}) = S_M$ & Closed & Seq.\ GEVP, exact \\
$\Aut(\mathbf{R}) \subsetneq S_M$, sep.\ $\delta$ & Partial & Span search \\
General, degenerate $\delta$ & Partial & Null-space tiebreaker \\
\bottomrule
\end{tabular}
\end{table}

The Sequential GEVP is therefore \emph{sound} (every accepted permutation is in $\Aut(\mathbf{R})$) but, in its plain single-direction form, not always \emph{complete} (the recovered subgroup may be proper in $\Aut(\mathbf{R})$). Two completeness gaps were noted, and both are partially resolved by the same principle: replace the single smallest-eigenvalue direction with a search over the relevant subspace, letting the group test select the generator. The first is the proper-subgroup case with well-separated $\delta$, treated above: span search recovers the full group whenever the candidate basis spans a generating set, leaving only the a~priori basis-richness question open. The second is the degenerate-spectrum case, where the smallest GEVP eigenvalue has high multiplicity and no single rounded direction is distinguished. When the degeneracy is \emph{removable}, the true generators lie within the degenerate null space, and a tiebreaker that searches the null space for directions whose rounded permutation has small symmetry defect recovers them. On a representative non-Abelian instance ($\Aut(\mathbf{R}) = D_4$ on $M = 8$, with a double-commutator null space of dimension six), the null-space defect tiebreaker recovers the full order-eight group where rounding the single smallest eigenvector alone does not. What remains open in the degenerate case is the genuinely \emph{intrinsic} sub-case, in which distinct generating sets are indistinguishable at second order, so the question is one of convention (which subgroup of $\Aut(\mathbf{R})$ to designate the matched group) rather than of recovery, together with the efficient enumeration of a high-dimensional null space. We summarize the practical state as a working hypothesis subject to these open conditions:

\begin{conjecture}[Practical universality of the Sequential GEVP, with basis-richness caveat]\label{conj:practical}
For a broad class of Hermitian positive-definite $\mathbf{R}$ arising from physical signal models (periodic, chirp, autoregressive, graph-filtered, array-manifold, and the like), there exists a basis $\mathcal{B}$ whose span reaches a generating set of $\Aut(\mathbf{R})$, and for any such basis the span-search form of the Sequential GEVP recovers $\Aut(\mathbf{R})$ in full. Identifying or characterizing such a basis from $\mathbf{R}$ alone is the remaining basis-richness open problem; given a sufficiently rich basis, the recovery procedure is no longer the obstruction.
\end{conjecture}

The conjecture is narrower than the universal claim made in earlier versions of this paper, which asserted that the plain single-direction procedure recovers $\Aut(\mathbf{R})$ in full. The $C_6$ example shows that the plain procedure does not: it stalls at $\langle\tau\rangle$. The span-search form does recover the full $D_6$ from a basis whose span reaches the reflections, so the recovery procedure is settled; what the conjecture leaves open is only the a~priori characterization of which bases are rich enough. This qualified statement is consistent with the empirical evidence available.

\subsection{Current methodology in summary}
\label{sec:methodology}

At the highest level, the methodology has two stages: \emph{library discovery}, in which candidate matched groups are proposed directly from the data rather than supplied by hand, followed by a \emph{two-tier selection} that chooses among the candidates and rejects spurious matches. Discovery is the role of the double-commutator GEVP of Theorem~\ref{thm:dcgevp} and its sequential extension: the continuous relaxation returns one or more generators whose closure proposes groups, so that even a practitioner with no prior model obtains a short list of structured candidates instead of searching the full subgroup lattice. Selection is then two-tiered. The first tier is a fast prefilter on effective dimension: a candidate $G$ whose group-averaged estimator $\hat{\mathbf{R}}_G$ has a participation ratio inconsistent with the data, $D_2(\hat{\mathbf{R}}_G)$ far from $D_2(\hat{\mathbf{R}})$, is screened out before any expensive test. The second tier scores the survivors for consistency, by the cross-validation residual $D_{CV}$ of~\eqref{eq:dcv} against the no-symmetry baseline, with the $\hat{\kappa}$ trajectory as confirmation. Discovery proposes; the two tiers dispose. The steps below instantiate this two-stage architecture in the order they are applied.
\begin{enumerate}[label=(\arabic*)]
\item \textbf{Structure check via $\alpha$.} Compute the coloring index $\alpha(\mathbf{R}) = \|\mathbf{R} - \bar{q} \mathbf{I}\|_F / \|\mathbf{R}\|_F$, with $\bar{q} = \Tr(\mathbf{R})/M$ (see~\cite{thornton2026ad} for properties and bounds). If $\alpha \approx 0$, the signal is white and no group offers an advantage over the trivial group; stop.
\item \textbf{Continuous search via DCG EVP.} Apply the double-commutator GEVP of Theorem~\ref{thm:dcgevp} on a candidate subspace $\mathcal{B} \subseteq \mathfrak{u}(M)$ to produce one or more candidate generators.
\item \textbf{Group assembly via sequential deflation.} When multiple generators are needed, iterate the sequential GEVP of Section~\ref{sec:seqgevp} to construct $G_K \subseteq \Aut(\mathbf{R})$ one generator at a time, using the span-search form (searching the candidate span, not only the smallest direction, for a defect-zero extension) so that the procedure recovers the full $\Aut(\mathbf{R})$ whenever the candidate basis spans a generating set. The returned subgroup is always sound (every accepted permutation is in $\Aut(\mathbf{R})$); it equals $\Aut(\mathbf{R})$ when the basis is rich enough and is otherwise a proper subgroup.
\item \textbf{Consistency validation via $D_{CV}$.} Use the cross-validation criterion of~(\ref{eq:dcv}) on at least two snapshots to confirm that the assembled group is consistent with the data. For Abelian groups, an exhaustive partition-enumeration search over all Abelian groups of the relevant order is feasible and takes milliseconds.
\item \textbf{Verification via $\kappa$ trajectory.} Plot the structural capacity trajectory $\hat{\kappa}(L)$ as $L$ increases; a flat trajectory confirms the match.
\end{enumerate}

The methodology covers the library-free Abelian case (step 4 via partition enumeration), the structured-library non-Abelian case under the soundness guarantee of Section~\ref{sec:seqgevp} (steps 2 and 3 with sequential deflation, returning a subgroup of $\Aut(\mathbf{R})$), and the off-grid continuous parameter case (via conjugation sweeps tracked by $\hat{\psi}$, which is valid within an orbit class even though it is biased across orbit classes).

\section{Blind Group Matching Enables AD for Blind Signal Processing}
\label{sec:blindproblems}

Blind group matching is the AD counterpart of the parameter-identification step in many blind signal processing problems. Once a mechanism exists to identify the signal's algebraic symmetry from data, the reach of the AD toolkit extends from measurement (where the group acts on the observation vector) to the broader class of blind and adaptive problems (where the group acts on the output of a cost functional). This section develops that extension in detail for blind equalization, then briefly indicates a set of further blind problems covered by the same treatment, with detailed development of each left to forthcoming companion papers.

Blind signal processing problems, from blind equalization through blind source separation, carrier recovery, and timing recovery, share a common mathematical shape. In each, some unknown quantity (a channel impulse response, a mixing matrix, a carrier phase, a timing offset) must be estimated from the received signal alone, without access to pilot symbols, training sequences, or other auxiliary side information. The standard approach is to formulate a \emph{cost functional} $J(\mathbf{w})$ defined on the parameter vector $\mathbf{w}$ of an adaptive processor (equalizer, demixer, phase-locked loop, timing recovery loop), whose value is intended to be minimized when the processor has converged to a correct solution. The cost is chosen to exploit some known structural property of the transmitted signal class that survives the unknown propagation path: constant envelope of a digital-communication constellation, non-Gaussianity of the source symbols, cyclostationarity of the modulation format, and so on. Stochastic gradient descent on $J(\mathbf{w})$ is then used to drive $\mathbf{w}$ toward a minimum of the cost, and the hope is that the converged $\mathbf{w}$ inverts the unknown distortion.

The blind group matching problem studied in Section~\ref{sec:blindmatching} is no exception to this pattern: the cross-validation criterion $D_{CV}(G)$ plays the role of the cost functional, with the candidate group $G$ playing the role of the parameter, and minimization over a library (or, in the continuous relaxation, over directions in $\mathfrak{u}(M)$) is the adaptive step. The symmetry properties of $D_{CV}$ relative to the signal's symmetries determined which library entries were recoverable from data and which were not. The same question can be asked of a blind adaptive processor in general: what symmetry does its cost functional possess, and how does that symmetry relate to the symmetry of the signal class it operates on? This is the content of the cost-symmetry matching principle that follows.

\subsection{The cost-symmetry matching principle}
\label{sec:costsym}

Let $J(\mathbf{w})$ be a real-valued cost functional depending on an adaptive filter tap vector $\mathbf{w}$ through a scalar or vector output $y = \mathbf{w}^\top \mathbf{x}$. Define the \emph{cost invariance group}
\begin{equation}
G_{\mathrm{cost}} \;=\; \bigl\{\, g \;:\; J(y) = J(g \cdot y) \text{ for all } y \,\bigr\},
\end{equation}
the group of output transformations under which $J$ is invariant. Let $\mathcal{A}$ be the target alphabet and $G_{\mathrm{sig}}$ its natural automorphism group.

\begin{proposition}[Cost-Symmetry Matching]\label{prop:costsym}
Under stochastic gradient descent on $J$, at any converged minimum the filter output is statistically equivalent to $\mathcal{A}$ up to a transformation by some $g^* \in G_{\mathrm{cost}}$. Three cases arise:
\begin{enumerate}[label=(\roman*)]
\item If $G_{\mathrm{cost}} = G_{\mathrm{sig}}$, then $g^* \in G_{\mathrm{sig}}$ and the residual ambiguity is exactly the natural alphabet ambiguity.
\item If $G_{\mathrm{cost}} \supsetneq G_{\mathrm{sig}}$, then $g^*$ is distributed on the coset space $G_{\mathrm{cost}}/G_{\mathrm{sig}}$; the residual ambiguity exceeds the natural ambiguity by this coset.
\item If $G_{\mathrm{cost}} \subsetneq G_{\mathrm{sig}}$, the cost breaks symmetries the alphabet preserves, giving a biased estimator at steady state.
\end{enumerate}
\end{proposition}

Proposition~\ref{prop:costsym} is the cost-function counterpart of the commutativity condition $[\pi_g, \mathbf{R}] = 0$ of the measurement-domain AD framework. The proof follows from direct equivariance. If $\mathbf{w}^\star$ is a global minimum of $J$ coinciding with the zero-forcing equalizer (so that the output matches the transmitted alphabet up to delay), then for any $g \in G_{\mathrm{cost}}$ one has $J(g \cdot \mathbf{w}^\star) = J(\mathbf{w}^\star)$ by definition of $G_{\mathrm{cost}}$, so $g \cdot \mathbf{w}^\star$ is also a global minimum. The set of global minima is therefore the $G_{\mathrm{cost}}$-orbit of $\mathbf{w}^\star$. In case (i), that orbit consists of copies of the transmitted alphabet related by elements of $G_{\mathrm{sig}}$ alone, and no information is lost. In case (ii), the orbit also contains points that do not correspond to any $G_{\mathrm{sig}}$-relabelling of the alphabet but are rotated copies of it distinguishable only by external reference; the distinct orbit points modulo the $G_{\mathrm{sig}}$-orbit are parameterized by the coset space $G_{\mathrm{cost}}/G_{\mathrm{sig}}$, which is a continuum when $G_{\mathrm{cost}}$ is a continuous group. In case (iii), the minima of $J$ do not form a $G_{\mathrm{sig}}$-orbit; each element of $G_{\mathrm{sig}} \setminus G_{\mathrm{cost}}$ maps the alphabet to itself but does not preserve $J$, producing a bias that depends on which symbols the cost emphasizes.

\subsection{Blind equalization: CMA, MMA, and AD-$\mathbb{Z}_M$}
\label{sec:blindeq}

Blind equalization is a canonical blind signal processing problem in digital communications. A digital symbol stream $\{s[n]\}$ drawn from a known finite constellation (BPSK, QPSK, $M$-PSK, 16-QAM, and the like) is transmitted through a linear channel with unknown impulse response $h[n]$, producing a received signal $r[n] = (h * s)[n] + w[n]$, where $w[n]$ is additive noise. The receiver's task is to recover an estimate of the symbol stream by designing an equalizing filter $\mathbf{w}$ whose output $y[n] = (\mathbf{w} * r)[n]$ approximates $s[n-d]$ for some delay $d$. When a training or pilot sequence is available, the equalizer can be estimated directly by least-squares methods; when no such auxiliary information is provided, the problem becomes blind: the equalizer must be designed from the received signal alone, relying only on structural properties of the constellation that survive the channel.

The problem has a long history. Sato~\cite{sato1975} introduced the first blind equalization algorithm in 1975, exploiting the sign structure of the transmitted signal to derive an adaptive update without training. Godard~\cite{godard1980} generalized this approach in 1980 to the constant modulus class, producing the Constant Modulus Algorithm (CMA) that remains the most widely deployed blind equalization method in practice. Treichler and Agee~\cite{treichler1983} independently reformulated the same cost for the signal processing community in 1983. Since then, blind equalization has developed into a substantial literature addressing decision-directed recovery, fractionally-spaced equalization, the Bussgang class of algorithms~\cite{bellini1986}, higher-order-statistics methods~\cite{shalvi1990}, and the Multi-Modulus Algorithm (MMA) of Yang, Werner, and Dumont~\cite{yang2002} that addresses the failure of CMA on square QAM constellations. The problem is important across the full range of wireless and wired communication systems: modern cellular (3GPP, 5G NR), digital broadcast (DVB, ATSC), cable and fiber modems, underwater acoustic links, and any other system in which channel memory is long enough to require equalization but overhead constraints or broadcast-protocol limitations preclude dedicated training.

We restate the three main blind equalization algorithms (CMA, MMA, and an AD-matched cost to be introduced below) in the language of the cost-symmetry matching principle of Proposition~\ref{prop:costsym}. In each case, the key question is the relationship between the cost symmetry group $G_{\mathrm{cost}}$ and the constellation symmetry group $G_{\mathrm{sig}}$. The resulting analysis reproduces the classical CMA failure modes (residual phase ambiguity on $M$-PSK, excess misadjustment on QAM) as direct consequences of $G_{\mathrm{cost}} \supsetneq G_{\mathrm{sig}}$, predicts the numerical value of the residual-phase standard deviation in closed form, and organizes the ad-hoc fix-up literature under a single algebraic criterion.

The Constant Modulus Algorithm (CMA) of Godard~\cite{godard1980} and of Treichler and Agee~\cite{treichler1983} minimizes
\begin{equation}\label{eq:cma}
J_{\mathrm{CMA}}(\mathbf{w}) \;=\; \mathbb{E}\bigl[(|y|^2 - R^2)^2\bigr],
\end{equation}
with $R^2 = \mathbb{E}[|s|^4]/\mathbb{E}[|s|^2]$ and $s$ a symbol from the modulation constellation. The cost depends on $y$ only through $|y|^2$, so $G_{\mathrm{cost}}^{\mathrm{CMA}} = U(1)$, the continuous circle group. For a constellation with rotational symmetry $\mathbb{Z}_M$ ($M$-PSK) or $D_4$ (square QAM), case (ii) of Proposition~\ref{prop:costsym} applies: CMA's residual ambiguity is distributed on the continuous coset space $U(1)/\mathbb{Z}_M$ or $U(1)/D_4$.

\begin{corollary}[CMA residual phase]\label{cor:cma}
For a constellation with rotational symmetry $\mathbb{Z}_M$ and CMA at steady state, the residual phase $\theta^*$ of the equalizer output with respect to the nearest $\mathbb{Z}_M$ grid point is uniformly distributed on $[-\pi/M, \pi/M]$ over an ensemble of random channel realizations. The standard deviation is $\pi/(M\sqrt{3})$, evaluating to $45^\circ/\sqrt{3} \approx 25.98^\circ$ for $M = 4$.
\end{corollary}

The Multi-Modulus Algorithm (MMA) of Yang, Werner, and Dumont~\cite{yang2002} has $G_{\mathrm{cost}}^{\mathrm{MMA}} = D_4$ and is matched to square QAM in the sense of Proposition~\ref{prop:costsym}(i). A cost matched exactly to cyclic-group constellations is
\begin{equation}\label{eq:adzm}
J_{\mathrm{AD}\text{-}\mathbb{Z}_M}(\mathbf{w}) \;=\; \mathbb{E}[|y^M - C_M|^2], \quad C_M = \mathbb{E}[s^M],
\end{equation}
with $G_{\mathrm{cost}}^{\mathrm{AD}\text{-}\mathbb{Z}_M} = \mathbb{Z}_M$. For $M$-PSK this matches $G_{\mathrm{sig}} = \mathbb{Z}_M$ exactly; for square QAM, $\mathbb{Z}_4 \subsetneq D_4$, and the Supergroup Dominance Theorem (Theorem~\ref{thm:supergroup}) predicts that $D_4$-matched costs (MMA) recover both phase and magnitude while $\mathbb{Z}_4$-matched costs recover only the phase.

\paragraph{Empirical validation.} On 3GPP TR~38.901~\cite{tr38901} TDL-A through TDL-D multipath channels at 25~dB SNR, over 200 random realizations per profile, the measured residual phase spread for CMA falls between $24.4^\circ$ and $26.4^\circ$ across all eight (profile, constellation) combinations, within $1.6^\circ$ of the Corollary~\ref{cor:cma} prediction. MMA and AD-$\mathbb{Z}_M$ reduce the spread to sub-degree levels where they are group-matched. Figure~\ref{fig:tdl_grid} displays the result.

\paragraph{CMA modulus mismatch on 16-QAM.} A separate failure mode of CMA on QAM is predicted by Proposition~\ref{prop:costsym} through the modulus structure rather than the phase structure. The constellation-point moduli of 16-QAM take three distinct values after unit-average-power normalization, but the CMA cost drives $|y|^2$ to a single value $R^2 \approx 1.32$ that is not attained by any transmitted symbol. The equalizer output therefore lies on a compromise circle that approximates no single symbol exactly, inflating the nearest-symbol decision error. In simulation on 16-QAM, this modulus-collapse effect produces approximately twice the symbol mean-square error that MMA and AD-$\mathbb{Z}_M$ achieve on the same channel realizations, consistent with the cost-symmetry analysis of Proposition~\ref{prop:costsym}(ii). A regularized variant $J_{\mathrm{CMA}} + \varepsilon J_{\mathrm{AD}\text{-}\mathbb{Z}_M}$ with a small phase-locking term in $\mathbb{Z}_4$ recovers the standard 16-QAM performance while preserving the CMA convergence properties, and provides a concrete example of how the matching principle suggests specific remedies where a mismatch is diagnosed.

\begin{figure*}[t]
\centering
\includegraphics[width=\textwidth]{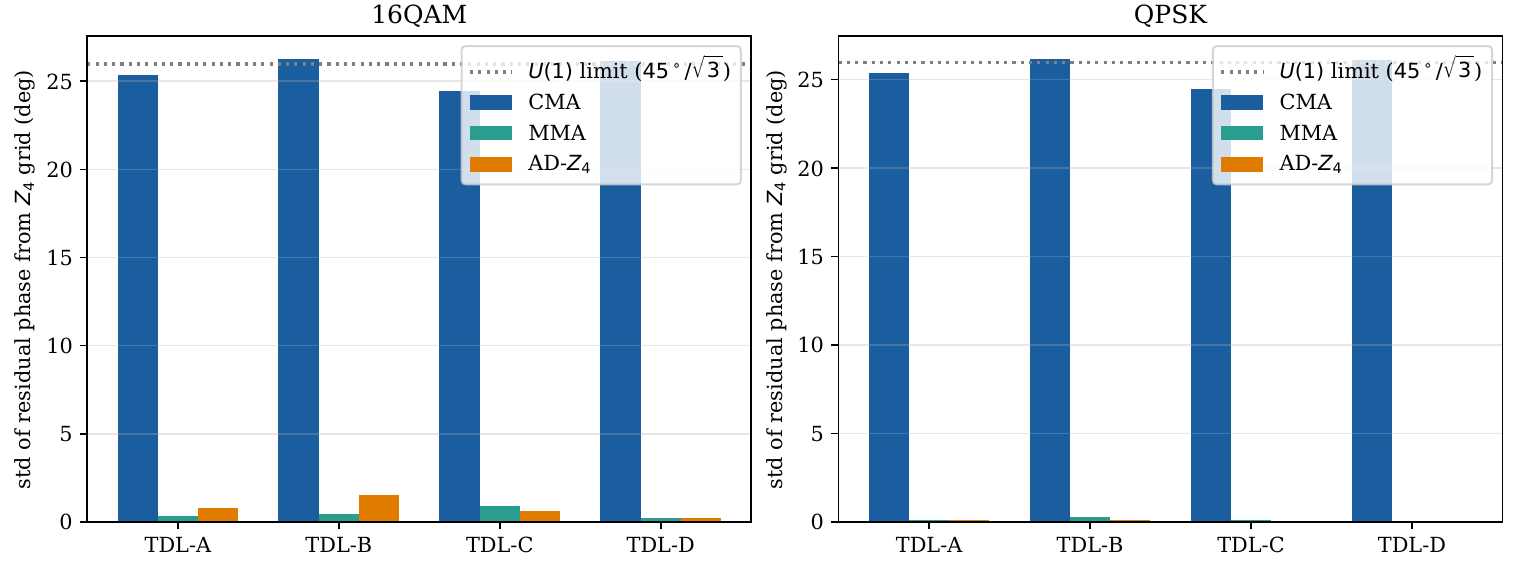}
\caption{Residual phase standard deviation from the nearest $\mathbb{Z}_4$ grid point, by 3GPP TR~38.901 TDL profile (A, B, C NLOS; D LOS) and constellation (16-QAM, left; QPSK, right), 200 random channel realizations per profile at 25~dB SNR. The CMA bars saturate the closed-form $U(1)$-uniform-within-cell prediction $45^\circ/\sqrt{3} \approx 25.98^\circ$ of Corollary~\ref{cor:cma} (dashed gray line) within $1.6^\circ$ on every profile and constellation. MMA and AD-$\mathbb{Z}_M$ reduce the residual to sub-degree levels.}
\label{fig:tdl_grid}
\end{figure*}

\subsection{Further blind problems in the framework's scope}
\label{sec:conjectures}

The cost-symmetry matching principle of Section~\ref{sec:costsym} and the blind group matching machinery of Section~\ref{sec:blindmatching} together extend to a range of classical blind signal processing problems beyond blind equalization. We indicate the mapping in each case below. Detailed development, simulations, and comparison with existing methods for each problem are left to forthcoming companion papers; the purpose here is to locate these problems within the framework rather than to treat any of them exhaustively.

\emph{Blind source separation (BSS).} The separating-matrix ambiguity group is $G_{\mathrm{sig}}^{\mathrm{BSS}} = S_n \ltimes (\mathbb{C}^*)^n$, the semidirect product of a permutation of the $n$ sources with a complex scaling per channel. Independent component analysis algorithms (FastICA, JADE, Infomax, kurtosis-based methods) use cost functionals with different invariance groups, and Proposition~\ref{prop:costsym} organizes them by case. The BSS cost-function landscape maps onto $(G_{\mathrm{cost}}, G_{\mathrm{sig}}^{\mathrm{BSS}})$ pairs, which provides a principled basis for selecting the algorithm matched to a given problem.

\emph{Blind carrier-frequency recovery.} A residual carrier frequency acts as a one-parameter subgroup of $U(1)$ on the baseband output. The matched-cost construction is the AD counterpart of conventional carrier recovery, operating on a short observation window with explicit $U(1)$-orbit awareness.

\emph{Blind timing recovery.} A timing offset is a translation on the sampled output. The matched group is a discrete subgroup of the translation group, and the matched-cost construction produces an explicit interpolation procedure.

\emph{Blind channel identification.} For a channel with known input-signal symmetry structure, blind identification via second-order cyclostationary statistics (Gardner~\cite{gardner1991}) exhibits $\mathbb{Z}_N$ symmetry at the symbol rate. AD with the $\mathbb{Z}_N$ matched group provides an alternative to the cyclostationary treatment, with a concrete tie to the group-averaged estimator via rank promotion of the cyclostationary statistics into tensor form.

\emph{Blind dictionary discovery for compressed sensing.} When dictionary atoms are to be learned from data and no fixed basis is available, AD offers an approach in which the atoms are generated as the orbit of a seed vector under a group identified by blind group matching. The construction provides a natural link to the compressed-sensing discussion of Section~\ref{sec:cs}, where dictionary construction is one of several fundamental difficulties alongside measurement operator design, reconstruction algorithms, and robustness to model mismatch.

\section{Spectral Coherence as Commutant Off-Diagonal Structure}
\label{sec:coherence}

\subsection{Overview}

The framework estimates structure from a single ensemble of observations by averaging an outer-product estimator over a matched finite group, and by reading the geometry of the resulting commutant. The preceding sections develop the \emph{diagonal} content of this construction (Section~\ref{sec:diagoffdiag}): the group-averaged (Reynolds) projection of the sample covariance, read in the group-adapted basis, is a spectral estimator that generalizes the discrete Fourier, cosine, and Karhunen--Lo\`{e}ve cases. This section develops the \emph{off-diagonal} content. We show that the part of the matched covariance that the commutant projection discards, the commutativity residual, is exactly the cross-spectral correlation between components, and that this single quantity recovers the established hierarchy of stationary, cyclostationary, and spectrally correlated descriptions as special cases. The matched-group selection step plays the role of the frequency-warping operation in that literature, and the commutant convergence rate plays the role of the reliability tradeoff that governs how well the cross-spectral structure can be measured from finite data.

The practical consequence is a capability that the diagonal alone cannot provide: detecting and quantifying coherence between components that overlap so strongly in the matched domain that no power-spectrum method can separate them. We characterize this capability empirically and identify its limiting mechanism as the conditioning of the component geometry.

\subsection{Setup and assumptions}
\label{sec:coh-setup}

Let $\{x_k\}_{k=1}^{K}$ be an ensemble of length-$N$ observations, each a mixture of group-structured components. The sample covariance is $\hat{\mathbf{R}} = \tfrac{1}{K}\sum_{k} x_k x_k^{H}$. A matched finite group $G$ acts on $\mathbb{C}^{N}$ by unitaries; for a linear frequency modulated (chirp) component of rate $\beta$, the matched group is the cyclic shift conjugated into the dechirped (fractional) domain, $U_\beta = D_\beta^{-1} Z D_\beta$, where $Z$ is the cyclic shift and $D_\beta$ is the dechirp $\diag(e^{-\imath \pi \beta (n/N)^2})$. We write $\mathbf{R}_\beta = D_\beta \hat{\mathbf{R}} D_\beta^{H}$ for the dechirped covariance and $\mathbf{M} = \mathbf{F}\, \mathbf{R}_\beta\, \mathbf{F}^{H}$ for its representation in the Fourier basis, with $\mathbf{F}$ the unitary DFT.

We make two modeling assumptions explicit, because they separate this development from the parameter-estimation problem that usually accompanies it.

\textbf{(A1) The group parameters are known or pre-estimated.} The chirp rates that define the matched group are assumed available, either from the blind group matching of Section~\ref{sec:blindmatching} or from a separate classical estimator. Rate estimation from sufficiently many observed cycles is a well-posed problem with mature low-SNR solutions based on dechirp-and-search, fractional autocorrelation \cite{akay1998fracautocorr}, and cross-correlation; it is decoupled here from the coherence question, which concerns the relationship \emph{between} already-located components rather than their location.

\textbf{(A2) Phase-only coherence model.} Two components carry independent magnitudes; the classes differ only in whether their relative phase is locked across the ensemble (coherent) or random (incoherent). This isolates coherence as a property of the complex cross-moment and removes any magnitude cue, so that a phase-blind statistic has, by construction, no information to use.

\subsection{The diagonal and the off-diagonal}
\label{sec:coh-diag}

A matrix is circulant if and only if $\mathbf{F}$ diagonalizes it, so the Frobenius projection of $\mathbf{R}_\beta$ onto the circulant (cyclic-commutant) subspace is the operation that retains the Fourier-domain diagonal,
\begin{equation}
  P_{G}(\mathbf{R}_\beta) \;=\; \mathbf{F}^{H}\,\diag(\mathbf{M})\,\mathbf{F} .
\end{equation}
The diagonal $\{\mathbf{M}_{ff}\}$ is the dechirped power spectrum. A short calculation using the circular Wiener--Khinchin relation shows that this diagonal, read as a spectrum, equals the matched periodogram exactly; the full cyclic group therefore contributes no variance reduction by itself, and all few-shot gains come from a genuine reduction of effective dimension. The complementary quantity, the commutativity residual,
\begin{equation}
  s(\beta)^2 \;=\; \frac{\big\lVert \mathbf{R}_\beta - P_{G}(\mathbf{R}_\beta)\big\rVert_F^2}
                       {\lVert \mathbf{R}_\beta\rVert_F^2}
            \;=\; \frac{\big\lVert \mathrm{offdiag}(\mathbf{M})\big\rVert_F^2}
                       {\lVert \mathbf{M}\rVert_F^2},
\label{eq:resid}
\end{equation}
is the fraction of energy off the Fourier diagonal; it is exactly the off-diagonal leg of the structural-index decomposition of Section~\ref{sec:diagoffdiag}, the squared commutant residual. By unitary invariance it is the off-diagonal mass of $\mathbf{M}$, which is the discrete Lo\`{e}ve bifrequency spectrum of the dechirped process. Equation~\eqref{eq:resid} is the central identification of this section: \emph{the commutant residual is the cross-spectral (bifrequency) correlation, and the diagonal that the projection keeps is the ordinary power spectrum.}

For two components with per-snapshot amplitudes $a_{1k}, a_{2k}$ and matched steering vectors $c_1, c_2$, projecting each snapshot onto $[\,c_1\;c_2\,]$ yields a two by two source covariance whose off-diagonal entry is the complex cross-moment $\gamma_{12} = \mathbb{E}[a_1 \overline{a_2}]$. The magnitude-squared coherence, normalized as $|\gamma_{12}|^2/(P_1 P_2)$, is a sufficient statistic for the coherence of the pair and uses the full complex cross-moment; it is the spectral coherence index $\chi$ of Section~\ref{sec:diagoffdiag} written for a resolved component pair. The power spectrum, by contrast, depends on $\gamma_{12}$ only through the interference term
\begin{equation}
  \mathbf{M}_{ff} \;\supset\; 2\,\mathrm{Re}\!\big(\gamma_{12}\, W[f]\big),
\label{eq:interf}
\end{equation}
with $W[f] = (\mathbf{F} D_\beta c_1)[f]\,\overline{(\mathbf{F} D_\beta c_2)[f]}$, a fixed, phase-rotated real projection of $\gamma_{12}$ integrated against the pattern $W$. The diagonal therefore keeps a shadow of the cross-moment; the off-diagonal keeps all of it. This is why a power-spectrum coherence test is phase-dependent while the commutant off-diagonal test is not.

\subsection{The unifying view}
\label{sec:coh-unify}

The location of the off-diagonal mass of $\mathbf{M}$ classifies the second-order description of the process, and the matched group selects the basis in which that mass lands on the main diagonal.

For a wide-sense stationary process the Lo\`{e}ve bifrequency spectrum has support only on the main diagonal: distinct spectral components are uncorrelated, $\mathbf{M}$ is diagonal, and the power spectrum is the complete second-order description. There is no off-diagonal and hence no coherence to read.

For an almost-cyclostationary process \cite{gardner1986spectral, gardner1994cyclostationarity} the off-diagonal mass lies on lines parallel to the main diagonal, offset by a countable set of cycle frequencies. The normalized off-diagonal at a cycle frequency is exactly the spectral coherence of cyclostationary theory, and its use for detection, time-difference estimation, and signal-selective direction finding is the content of the unifying coherence treatment of \cite{gardner1992unifying} and the survey \cite{gardner2006cyclostationarity}.

For a spectrally correlated process \cite{napolitano2011sampling, napolitano2003uncertainty, napolitano2019cyclostationary} the off-diagonal mass lies on a countable set of support \emph{curves} rather than lines, and chirp and angle-modulated signals are the canonical examples. The density of the Lo\`{e}ve bifrequency spectrum on those curves is the bifrequency spectral correlation density. The matched dechirp $D_\beta$ warps the relevant support curve onto the main diagonal, which is precisely the frequency-warping that maps a spectrally correlated process toward stationarity \cite{napolitano2012spectral}; in that warped basis the cross-spectral correlation becomes the off-diagonal of a near-circulant covariance, which is the commutant residual of Eq.~\eqref{eq:resid}.

Two AD constructs acquire an interpretation in these terms. Blind group matching, which selects the matched group, is the selection of the warping that brings the support curve to the diagonal; it is the data-driven analogue of choosing the cycle frequency or the warping law. The commutant convergence result for the fractional case shows that the matched, dechirped covariance approaches a circulant (stationary) form as the record lengthens, with a controlled rate; this rate is the framework-internal counterpart of the documented reliability tradeoff in spectrally correlated estimation, where single-record measurement of the cross-spectral correlation degrades as the support-curve slope departs from unity \cite{napolitano2003uncertainty}. In one statement: the AD commutant places the stationary, cyclostationary, and spectrally correlated descriptions on a common footing as the diagonal-to-off-diagonal structure of a single group-averaged covariance, with the matched group fixing the basis and the convergence rate fixing the measurability.

\subsection{Empirical characterization}
\label{sec:coh-empirical}

We illustrate the off-diagonal capability and its limit on crossing chirp pairs under assumptions (A1) and (A2), with the rates supplied to both the off-diagonal statistic and a power-spectrum statistic. The detection task is to classify a pair as coherent or incoherent from $K$ snapshots; performance is reported as the area under the ROC curve (AUC).

Figure~\ref{fig:coh-delta} sweeps the relative phase between the two components. The off-diagonal statistic is flat at perfect separation across all relative phases, consistent with its use of the full complex cross-moment. A power-spectrum interference detector, which can access only the projection in Eq.~\eqref{eq:interf}, swings with the relative phase and is below chance for in-phase pairs, exactly the phase-incompleteness predicted by the projection. A magnitude-only control is uninformative throughout, confirming that the model carries no magnitude cue.

\begin{figure}[t]
  \centering
  \includegraphics[width=0.92\linewidth]{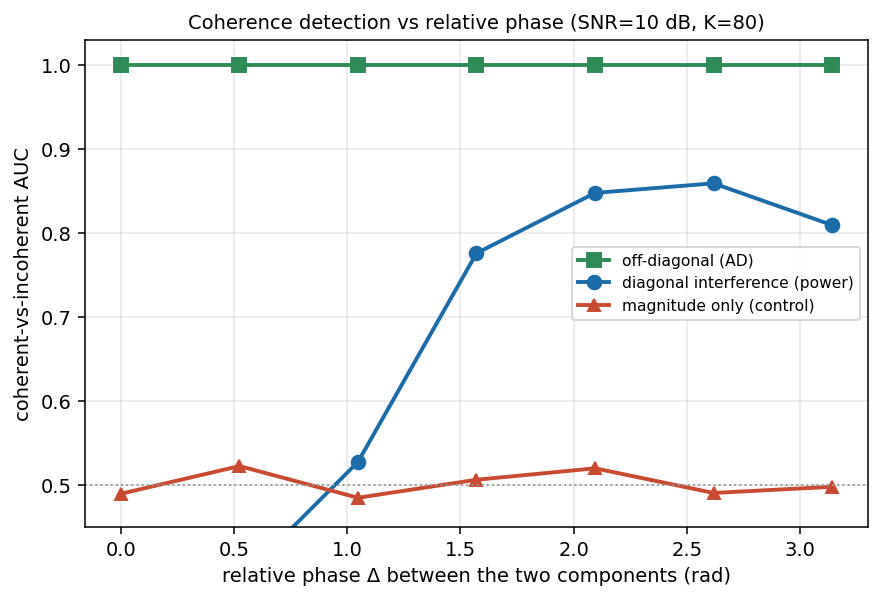}
  \caption{Coherent versus incoherent detection AUC against the relative phase between the two components. The commutant off-diagonal statistic is phase-complete (flat at unity); the power-spectrum interference statistic is phase-incomplete; the magnitude-only control is uninformative. Rates known; $K$ snapshots at fixed SNR.}
  \label{fig:coh-delta}
\end{figure}

Figure~\ref{fig:coh-overlap} sweeps the component overlap $\mu = |\langle c_1, c_2\rangle|/(\lVert c_1\rVert\,\lVert c_2\rVert)$ toward unity, far past the overlap at which the components cease to be separable on the diagonal. The off-diagonal statistic holds perfect separation up to extreme overlap, with the reachable overlap set by the SNR: in these runs to roughly $\mu \approx 0.996$ at $10$~dB, $\mu \approx 0.99$ at $3$~dB, and $\mu \approx 0.97$ at $-3$~dB. The breakdown mechanism, shown in Figure~\ref{fig:coh-mechanism}, is conditioning. The Gram matrix of the two matched components has condition number $(1+\mu)/(1-\mu)$, which diverges as $\mu \to 1$ and amplifies noise into spurious off-diagonal correlation; the estimated coherence of the incoherent class then inflates toward that of the coherent class and the two distributions merge. Increasing SNR pushes the wall toward $\mu = 1$; increasing the snapshot count helps only modestly, because the limit is set by the conditioning times noise product rather than by averaging. At $\mu = 1$ the limit becomes structural: identical components have no well-defined between-component coherence to recover.

\begin{figure}[t]
  \centering
  \includegraphics[width=0.92\linewidth]{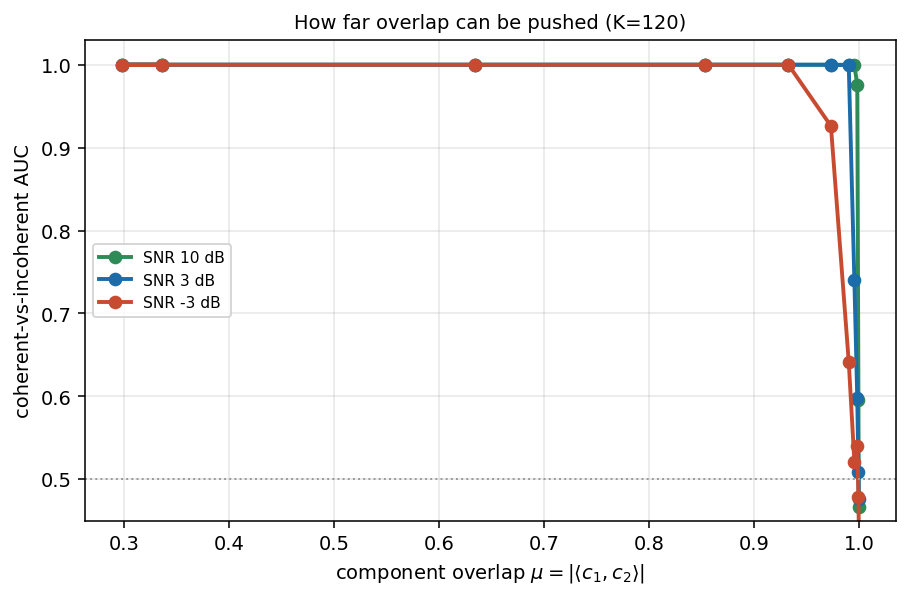}
  \caption{Reachable overlap. Coherent versus incoherent AUC against component overlap $\mu$, at three SNRs and fixed snapshot count. The off-diagonal statistic holds perfect separation deep into the unresolvable regime, with the cliff location set by SNR.}
  \label{fig:coh-overlap}
\end{figure}

\begin{figure}[t]
  \centering
  \includegraphics[width=0.92\linewidth]{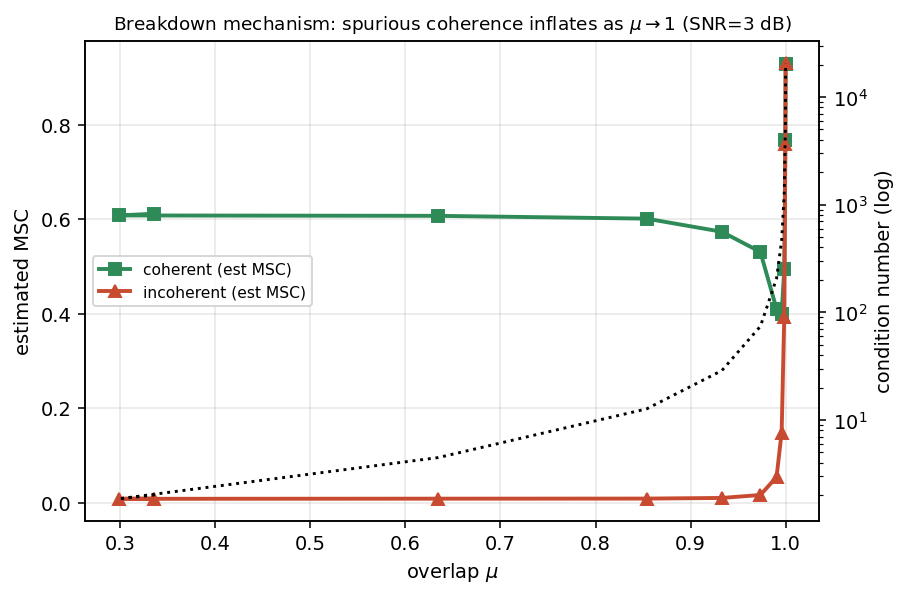}
  \caption{Breakdown mechanism. As $\mu \to 1$ the Gram condition number $(1+\mu)/(1-\mu)$ diverges (dotted, right axis) and spurious coherence inflates the incoherent class (estimated magnitude-squared coherence, left axis) until the classes merge.}
  \label{fig:coh-mechanism}
\end{figure}

The conditioning origin of the wall indicates the natural remedy. The raw pseudo-inverse separation used here makes no attempt to control the noise amplification; a whitened or covariance-reconstruction estimator is expected to push the wall closer to $\mu = 1$ at a fixed SNR. This is the same device that spatial smoothing and covariance reconstruction provide for coherent sources in array processing \cite{shan1985, qi2005spatialdiff}, and it connects the present construction to the coherent-source and coherent-multipath detection literature \cite{amani2018correlator}, where the discrimination of coherent from incoherent components is a long-standing applied problem.

\subsection{Discussion}
\label{sec:coh-disc}

The contribution of this section is unification rather than a new measurement primitive. The fact that cross-spectral coherence lives off the diagonal, and that the power spectrum is blind to it, is the founding premise of cyclostationary analysis \cite{gardner1992unifying}; its extension to chirp and other nonstationary signals is the spectrally correlated process theory \cite{napolitano2011sampling, napolitano2019cyclostationary}; and the discrimination of coherent components is mature in array processing and navigation \cite{shan1985, amani2018correlator}. What the AD framework adds is a single group-theoretic statement that carries all of these as special cases. The commutant residual is the bifrequency spectral correlation; the matched group is the warping that stationarizes the support curve; the power-spectrum and cross-spectrum split is the diagonal and off-diagonal of one group-averaged covariance; and the commutant convergence rate is the measurability tradeoff. Read this way, the framework does for spectral coherence what it does elsewhere for spectral estimation, namely it subsumes a layered sequence of case-specific descriptions into one construction and a choice of group. The empirical characterization then quantifies how far the off-diagonal can be pushed into the unresolvable regime, and locates the limit in the conditioning of the component geometry, a quantity the framework already exposes through the commutant.

\section{Information Structure Theory: Four Theorems and a Conjugate Capacity Bound}
\label{sec:fourthms}

The technical apparatus of the preceding sections supports a more general claim: the structural capacity $\kappa$ is the operational measure of single-observation estimation efficiency in the same sense that Shannon's entropy is the operational measure of source content. We make this concrete by stating four theorems that parallel Shannon's foundational results and a fifth result, the conjugate capacity bound for non-commuting generators on $\mathbb{C}^M$. The four theorems together constitute what we call \emph{information structure theory}, complementary to Shannon's information content theory.

\subsection{Structural capacity, two forms}
\label{sec:kappa_two}

The structural capacity admits two expressions, which coincide on flat spectra. The operational form is
\begin{equation}\label{eq:kappa_op}
\kappa(f, \mathbf{x}) \;=\; \max_{G} \, d_{\mathrm{eff}}(G, f),
\end{equation}
where the maximum is over finite groups $G$ acting unitarily on $\mathbb{C}^M$, with the PASE constraint $|G| \leq M$. The spectral form, for the outer-product statistic $f(\mathbf{x}) = \mathbf{x}\mathbf{x}^H$ with covariance $\mathbf{R}$, is
\begin{equation}\label{eq:kappa_spec}
\kappa(\mathbf{R}) \;=\; \frac{(\Tr \mathbf{R})^2}{\|\mathbf{R}\|_F^2} \;=\; \frac{1}{\sum_k p_k^2},
\end{equation}
with $p_k = \lambda_k(\mathbf{R}) / \Tr(\mathbf{R})$ the normalized eigenvalue distribution. The quantity $\sum_k p_k^2$ equals $\exp(-H_2)$, where $H_2$ is the Rényi-2 entropy of the spectrum, so $\kappa$ is a functional of $H_2$. The reciprocal $1/\sum_k p_k^2$ that appears in the spectral form is the \emph{participation ratio} of the eigenvalue distribution, also known in the ecological-diversity literature as the inverse Simpson index or Hill's diversity order 2~\cite{hill1973}, and the spectral form of the structural capacity is that participation ratio itself, so that the trivial-group baseline sits at $\kappa = 1$ in agreement with the operational form. The two forms measure the same resource from two sides, and coincide precisely on flat spectra. The operational form is the algebraic ceiling: the matched group attains the maximum, and $d_{\mathrm{eff}}$ counts the free spectral coordinates the orbit average populates. The spectral form reads off the value of that resource \emph{realized} at the particular eigenvalue distribution: by Theorem~\ref{thm:outer}, the participation ratio $(\Tr\mathbf{R})^2/\|\mathbf{R}\|_F^2$ is the per-entry variance-reduction factor actually delivered at covariance $\mathbf{R}$, so the spectral form is bounded by the operational form, with equality of the participation ratio and $d_{\mathrm{eff}}$ exactly when the matched-basis spectrum is flat, the case in which every free coordinate carries equal energy. A concentrated spectrum leaves the coordinate count unchanged but loads the energy unevenly, and the realized capacity is correspondingly smaller. The operational form generalizes to arbitrary statistics; the spectral form gives an immediate computational handle for the covariance setting and is the quantity to quote for spectrally concentrated signals. Geometrically (Section~\ref{sec:diagoffdiag}), this participation ratio is the $\ell_2$ summary of the matched power spectrum, the companion to the $\ell_\infty$ peak measured by the spectral concentration $\hat{\psi}$ of Section~\ref{sec:psi}: the two metrics read the same diagonal of the matched-basis covariance through different norms.

The spectral form is a ratio of Frobenius risks for a single circular complex Gaussian observation, namely $\mathbb{E}\|\mathbf{x}\mathbf{x}^H - \mathbf{R}\|_F^2 = (\Tr\mathbf{R})^2$ for the unconstrained outer product against $\mathbb{E}\|\hat{\mathbf{R}}_G - \mathbf{R}\|_F^2 = \sum_k \lambda_k^2$ for the matched group average. Both moments are taken in the circular complex convention; the corresponding real-Gaussian moments are $(\Tr\mathbf{R})^2 + \sum_k\lambda_k^2$ and $2\sum_k\lambda_k^2$, and mixing the two conventions shifts the ratio by one. Whether this risk ratio coincides with the ratio of the corresponding Cram\'{e}r--Rao bounds, rather than of the risks of these two particular estimators, is not established here and is left open.

The structural entropy $H_{\mathrm{struct}}$ of equation~\eqref{eq:hstruct} measures how the signal energy distributes across observation dimensions. White noise maximizes $H_{\mathrm{struct}}$ (uniform energy across all modes); a rank-one signal minimizes it (all energy on a single mode). A matched group reveals the spectral decomposition that exhibits the true $H_{\mathrm{struct}}$ from a single observation; an unmatched group (or the trivial group with a single observation) does not.

\subsection{Theorem 1: Structural Capacity as the Fundamental Measure}

The first of the four theorems establishes that the structural capacity $\kappa$ is characterized uniquely by a small set of operational properties, in the same way that Shannon's entropy is characterized uniquely by continuity, monotonicity, and the chain rule.

\textbf{Shannon (1948).} The entropy $H = -\sum_k p_k \log p_k$ is the unique measure (up to scaling) of information content satisfying continuity, monotonicity in the source distribution, and the chain rule.

\begin{theorem}[Structural Capacity as the Fundamental Measure]\label{thm:fundamental}
The structural capacity $\kappa(f, \mathbf{x})$ is the unique measure of single-observation estimation efficiency satisfying:
\begin{enumerate}[nosep, label=(S\arabic*)]
\item Non-negativity: $\kappa \geq 1$, with equality if and only if no nontrivial group improves estimation.
\item Monotonicity in structure: if $G_1 \subseteq G_2$ with $d_{\mathrm{eff}}(G_1, f) \leq d_{\mathrm{eff}}(G_2, f)$, then $\kappa \geq d_{\mathrm{eff}}(G_2, f)$.
\item PASE saturation: $\kappa \leq M$.
\item Trivial-group baseline: $\kappa = 1$ for scalar observations and for any $S_M$-invariant scalar statistic.
\item Multiplicativity on the $(G, L)$ continuum: total estimation quality from $L$ independent observations under group $G$ is $d_{\mathrm{eff}}(G, f) \cdot L$, capped at $\kappa \cdot L$.
\end{enumerate}
\end{theorem}

Property (S5) is the formal expression of the $(G, L)$ continuum that has been used throughout the paper. Properties (S2)-(S4) capture the operational content: structure can be exploited up to but not beyond the PASE bound, and the trivial group is the baseline for unstructured scalar estimation.

\subsection{Theorem 2: Structural Source Coding}

The second theorem plays the role of Shannon's source coding theorem: it identifies the minimum achievable estimation variance from a single observation in terms of the structural capacity, and identifies the matched group as the unique algebraic code that attains it.

\textbf{Shannon (1948).} A source with entropy $H$ can be compressed to $H$ bits per symbol; any code using fewer than $H$ bits per symbol incurs distortion.

\begin{theorem}[Structural Source Coding]\label{thm:source}
For a statistic $f$ and observation $\mathbf{x}$ with structural capacity $\kappa(f, \mathbf{x})$, the minimum achievable single-observation estimation variance is
\begin{equation}
\Var(\hat{\theta}) \;\geq\; \frac{C(f)}{\kappa(f, \mathbf{x})},
\end{equation}
where $C(f)$ is the single-evaluation variance of $f$. The matched group $G^* = \arg\max_G d_{\mathrm{eff}}(G, f)$ achieves this bound; any other group achieves $\Var = C(f) / d_{\mathrm{eff}}(G, f) > C(f) / \kappa$.
\end{theorem}

In Shannon's framework, suboptimal codes waste bits by transmitting redundancy. In the structural framework, mismatched groups waste structure: they leave exploitable algebraic regularity untouched and compensate by requiring more observations. The trivial group $\{e\}$ is the maximally wasteful structural code: it achieves $d_{\mathrm{eff}} = 1$ regardless of how much structure is available, and compensates with $N$ independent observations to reach variance $C(f)/N$. Classical temporal averaging is the structural equivalent of repeating every bit $N$ times for redundancy.

\begin{conjecture}[Converse: Structural Capacity Equals the CRB]\label{conj:converse_kappa}
For the matched group $G^*$, the variance bound $C(f)/\kappa$ equals the Cram\'{e}r-Rao lower bound for the $G^*$-invariant parameters from a single observation. The matched group achieves the information-theoretic limit among all estimators of these parameters, not merely the best group-averaged estimator.
\end{conjecture}

The Converse Theorem of Section~\ref{sec:gaat} establishes Conjecture~\ref{conj:converse_kappa} for the outer-product statistic in the Gaussian case, both Abelian (with $\Var(\hat{\lambda}_k) = \lambda_k^2 = \mathrm{CRB}(\lambda_k)$) and non-Abelian (with $\Var(\hat{\lambda}_i) = \lambda_i^2/d_i = \mathrm{CRB}(\lambda_i)$ via Schur's lemma). The non-Gaussian case is covered in two practically important ways. First, for any distribution with finite fourth moments, the group-averaged estimator with $L$ observations is asymptotically efficient as $L \to \infty$ by the central limit theorem applied to the $L \cdot |G|$ orbit elements; this extends the $(G, L)$ continuum to arbitrary non-Gaussian distributions with only an asymptotic qualifier. Second, for complex elliptical distributions of the form $\mathbf{x} = \sqrt{\tau}\,\mathbf{R}^{1/2}\mathbf{z}$ with $\mathbf{z}$ uniform on the unit sphere and $\tau$ a positive scalar texture (a family that includes Gaussian, compound Gaussian, and complex $t$-distributions), the spectral coefficients remain independent in the character basis, the Gaussian proof applies verbatim, and the group-averaged estimator again attains the CRB with a texture-dependent prefactor $\mathbb{E}[\tau^2]$ in place of unity. Between the CLT-based asymptotic result for any finite-fourth-moment distribution and the exact single-observation result for complex elliptical distributions, the Converse covers a wide class of signal models of practical relevance; a fully unified proof across arbitrary non-Gaussian distributions would be of theoretical interest but is unlikely to change the operational picture.

\subsection{Theorem 3: Structural Channel Capacity}

The third theorem is the structural counterpart of Shannon's channel capacity: it identifies the data structure itself as the object that plays the role of a communication channel, with the group and its effective order corresponding to the code and rate, and with structural capacity $\kappa$ acting as the upper bound on single-observation estimation quality.

\textbf{Shannon (1948).} A channel with capacity $C = \max_{p(x)} I(X; Y)$ supports reliable transmission at any rate $R < C$ and not at any rate $R > C$.

\begin{theorem}[Structural Channel Capacity]\label{thm:channel}
For an observation of tensor rank $r$ with dimension $M$ along each index, the structural channel capacity is
\begin{equation}
\mathcal{C}_{\mathrm{struct}}(r, M, f) \;=\; \max_{G} d_{\mathrm{eff}}(G, f) \;=\; \kappa(f, \mathbf{x}).
\end{equation}
The data structure (tensor rank, dimension, symmetry class) plays the role of the channel; the group $G$ plays the role of the code; $d_{\mathrm{eff}}(G, f)$ plays the role of the rate. The $(G, L)$ continuum extends the budget: total estimation quality from $L$ independent observations under group $G$ is $d_{\mathrm{eff}}(G, f) \cdot L$, capped at $\kappa \cdot L$.
\end{theorem}

\begin{table}[t]
\centering
\caption{Structural capacity of common data structures, with the matched group achieving $\kappa$.}
\label{tab:structures}
\renewcommand{\arraystretch}{1.2}
\begin{tabular}{@{}lcccl@{}}
\toprule
\textbf{Structure} & $f$ & $\kappa$ & \textbf{Matched group} \\
\midrule
Scalar                & any        & 1   & $\{e\}$              \\
Vector $\mathbb{C}^M$ & outer prod & $M$ & $\mathbb{Z}_M$       \\
Vector $\mathbb{C}^M$ & scalar mean &  1 & $\{e\}$              \\
Matrix $\mathbb{C}^{M \times N}$ & outer prod & $MN$ & $\mathbb{Z}_M \times \mathbb{Z}_N$ \\
Graph (n vertices)    & adj.\ outer & $|\Aut(\mathcal{G})|$ & $\Aut(\mathcal{G})$ \\
Quantum state ($d$)   & density mat & $d^2$ & Heisenberg-Weyl \\
\bottomrule
\end{tabular}
\end{table}

Two signals with identical Shannon entropy but different data structures (one stored as a vector, the other as a graph) have different structural capacities and require different numbers of observations for the same estimation quality. Shannon's capacity is a property of the channel, not the source; structural capacity is a property of the data structure, not the data content.

\subsection{Theorem 4: The Structural Processing Inequality}

The fourth theorem is the structural counterpart of the data processing inequality: it asserts that no operation performed on an observation can increase its structural capacity, mirroring Shannon's statement that processing cannot create information.

\textbf{Shannon (1948).} For any Markov chain $X \to Y \to Z$, the data processing inequality $I(X; Z) \leq I(X; Y)$ holds: processing cannot create information.

\begin{theorem}[Structural Processing Inequality]\label{thm:processing}
Let $h: \mathbb{C}^M \to \mathbb{C}^{M'}$ be any measurable function applied to the observation. For any statistic $f$,
\begin{equation}
\kappa(f \circ h^{-1},\, h(\mathbf{x})) \;\leq\; \kappa(f, \mathbf{x}).
\end{equation}
Processing cannot increase structural capacity.
\end{theorem}

The proof has three cases. If $h$ reduces dimensionality ($M' < M$), the set of available group actions shrinks ($G$ acts on $M'$ elements rather than $M$) and $\kappa$ can only decrease. If $h$ symmetrizes (sorting, taking the norm), the group orbit collapses and $d_{\mathrm{eff}}$ falls. If $h$ is an injective linear map preserving dimension, $\kappa$ is preserved but not increased. Examples of structure-destroying operations include scalar projection $h(\mathbf{x}) = \|\mathbf{x}\|^2$ (which collapses $\kappa$ from $M$ to $1$), matrix flattening (which discards row-column joint structure), sorting (which makes the observation $S_M$-invariant), and subsampling (which caps $\kappa$ at the retained dimension).

The structural processing inequality concerns operations on a single observation. The eigentensor hierarchy of Section~\ref{sec:eigentensor}, which combines multiple observations, is not a violation: it creates new index structure by combining independent observations rather than processing a single one.

\subsection{A Conjugate Capacity Bound for Non-Commuting Generators}
\label{sec:conjugate}

A complementary statement to the four theorems above arises when one considers two candidate groups acting on the same observation space whose infinitesimal generators do not commute. In this setting the structural capacities of the two groups are constrained jointly: allocating capacity to one places a limit on what is simultaneously available for the other. We present this as a finite-dimensional structural-capacity statement on $\mathbb{C}^M$; the extension to infinite-dimensional Lie-group actions on function spaces such as $L^2(\mathbb{R})$ is outside the scope of this paper.

\begin{definition}[Generator pair on $\mathbb{C}^M$]
Let $A, B$ be Hermitian matrices on $\mathbb{C}^M$. The associated unitary one-parameter groups are $U_A(\alpha) = e^{i\alpha A}$ and $U_B(\beta) = e^{i\beta B}$.
\end{definition}

\begin{definition}[Single-vector structural capacity for a generator]
For a unit vector $\mathbf{x} \in \mathbb{C}^M$ and a Hermitian generator $A$, the \emph{generator-specific structural capacity} of $\mathbf{x}$ for $A$-estimation is
\begin{equation}
\kappa_A(\mathbf{x}) \;=\; \frac{1}{(\Delta_{\mathbf{x}} A)^2}, \qquad (\Delta_{\mathbf{x}} A)^2 = \mathbf{x}^H A^2 \mathbf{x} - (\mathbf{x}^H A \mathbf{x})^2.
\end{equation}
A small uncertainty $\Delta_{\mathbf{x}} A$ corresponds to high structural capacity: $\mathbf{x}$ is well-localized in the $A$-eigenbasis, which is what enables efficient single-snapshot estimation under the group generated by $A$.
\end{definition}

\begin{theorem}[Conjugate Capacity Bound]\label{thm:conjugate}
Let $A, B$ be Hermitian matrices on $\mathbb{C}^M$ with $[A, B] = ic\,\mathbf{I} + C$ where $c \in \mathbb{R}$ and $C$ is a Hermitian residual. For any unit vector $\mathbf{x}$,
\begin{equation}\label{eq:conj_bound}
\kappa_A(\mathbf{x}) \cdot \kappa_B(\mathbf{x}) \;\leq\; \frac{4}{|c + \mathbf{x}^H C \mathbf{x}|^2}.
\end{equation}
When $C = 0$ (so the commutator is a scalar multiple of the identity), the bound simplifies to $\kappa_A \cdot \kappa_B \leq 4/c^2$, independent of $\mathbf{x}$.
\end{theorem}

\begin{proof}
The Robertson inequality applied to Hermitian $A, B$ with commutator $[A, B] = ic\mathbf{I} + C$ gives $\Delta_{\mathbf{x}} A \cdot \Delta_{\mathbf{x}} B \geq |\langle [A, B] \rangle_{\mathbf{x}}|/2 = |c + \mathbf{x}^H C \mathbf{x}|/2$. Squaring and reciprocating yields the stated bound. When $C = 0$, the second term vanishes and the right-hand side is the state-independent constant $4/c^2$.
\end{proof}

The bound is the structural-capacity restatement of the long-standing Robertson inequality~\cite{robertson1929}, recast here in the language of $\kappa$ rather than $\Delta$. Its purpose in the present framework is to identify the finite-dimensional setting in which complementary group choices on the same observation are subject to a joint capacity budget. The classical time-frequency uncertainty relation $\Delta\omega \cdot \Delta t \geq 1/2$ on $L^2(\mathbb{R})$ is a well-known infinite-dimensional instance of the same Robertson construction with a canonical commutator; we record this connection but do not pursue the infinite-dimensional Lie-group treatment here. The corresponding double-commutator structure $[A, [A, \mathbf{R}]]$ enters the blind-matching machinery of Section~\ref{sec:dadcad}: when $[A, B] = ic\mathbf{I}$ exactly, the AD double-commutator GEVP for $A$ and the GEVP for $B$ have decoupled Lie-algebra residuals; otherwise they couple through the residual $C$.

\subsection{Table of Shannon-AD correspondences}

\begin{table}[t]
\centering
\caption{Shannon's content theory (1948) and the AD structure theory developed here, with both functionals applying to the same eigenvalue distribution through different Rényi orders.}
\label{tab:correspondences}
\renewcommand{\arraystretch}{1.2}
\begin{tabular}{@{}p{0.46\columnwidth} p{0.46\columnwidth}@{}}
\toprule
\textbf{Shannon (Content, R\'{e}nyi-1)} & \textbf{AD (Structure, R\'{e}nyi-2)} \\
\midrule
Entropy $H = -\sum p_k \log p_k$ & Structural capacity $\kappa = \max_G d_{\mathrm{eff}}$ \\
Bits per symbol & Estimation units per observation \\
Source coding: rate $\geq H$ & Var bound: $\Var \geq C/\kappa$ \\
Channel capacity $C = \max I(X;Y)$ & Structural capacity $\mathcal{C}_{\mathrm{struct}}$ \\
Optimal code & Matched group $G^*$ \\
Suboptimal code & Mismatched group ($d_{\mathrm{eff}} < \kappa$) \\
Data processing inequality & Structural processing inequality \\
Mutual information $I(X;Y)$ & Spectral concentration $\psi$ \\
Deterministic source ($H = 0$) & Scalar observation ($\kappa = 1$) \\
Maximum-entropy source & Full-rank tensor (max structure) \\
Complementary channels & Conjugate generators (Robertson) \\
\textemdash{} & Conjugate capacity bound on $\mathbb{C}^M$ \\
\textemdash{} & Law of large numbers ($\kappa = 1$) \\
\bottomrule
\end{tabular}
\end{table}

The last two rows of Table~\ref{tab:correspondences} have no Shannon analog: both the conjugate capacity bound and the law of large numbers are structural phenomena invisible from the content perspective. The LLN is the structural-capacity bound for rank-0 data with $\kappa = 1$, recovered from Theorem~\ref{thm:source} with $C(f) = \sigma^2$ and $\kappa = 1$ giving $\Var \geq \sigma^2 / N$ from $N$ observations. This is a one-line consequence of the four theorems above.

\section{Data versus Structure: the Effective-Dimension--Snapshot Residue Surface}
\label{sec:gl-surface}

The framework offers a practitioner two distinct ways to sharpen a second-order estimate. One can collect more data, increasing the number of snapshots $L$, or one can impose more algebraic structure, averaging the single-observation estimate over a larger matched group $G$. The PASE principle quantifies the second route through the group gain, and the $(G, L)$ continuum of Section~\ref{sec:fourthms} places the two routes on a common footing. This section makes that footing explicit by treating the algebraic residue energy as a surface over the two coordinates and reading its two slopes as marginal values. The surface turns out to be a bias--variance object whose variance term carries the group-gain law, and its shape tells a practitioner whether the right next move is to find structure or to collect data. A non-Abelian experiment then shows that the correct horizontal axis is not the group order but the effective dimension, refining the $d_{\mathrm{eff}}$ used in Section~\ref{sec:fourthms}.

\subsection{The residue surface is a bias--variance decomposition}

Let $\mathbf{R}$ be the $M \times M$ population covariance, let $\hat{\mathbf{R}}$ be the sample covariance formed from $L$ snapshots, and let $\rho$ be a unitary representation of a finite group $G$. The group-averaged estimator is
\begin{equation}
\hat{\mathbf{R}}_G \;=\; \frac{1}{|G|}\sum_{g \in G} \rho(g)\, \hat{\mathbf{R}}\, \rho(g)^{\dagger},
\end{equation}
which is the orthogonal projection of $\hat{\mathbf{R}}$ onto the commutant $\mathcal{C}(\rho) = \{X : \rho(g) X = X \rho(g)\ \forall g\}$. Define the normalized residue energy $\mathcal{E}(|G|, L) = \mathbb{E}\,\|\hat{\mathbf{R}}_G - \mathbf{R}\|_F^2 / \|\mathbf{R}\|_F^2$. Because group averaging is a projection, the error splits cleanly into a structural term and a sampling term,
\begin{equation}
\label{eq:gl-bv}
\mathbb{E}\,\|\hat{\mathbf{R}}_G - \mathbf{R}\|_F^2
=
\underbrace{\|\Pi_{\mathcal C}\mathbf{R} - \mathbf{R}\|_F^2}_{\text{bias (depends on }G\text{)}}
+
\underbrace{\mathbb{E}\,\|\Pi_{\mathcal C}(\hat{\mathbf{R}} - \mathbf{R})\|_F^2}_{\text{variance}\,\propto\,1/(d_{\mathrm{eff}} L)},
\end{equation}
where $\Pi_{\mathcal C}$ is the projection onto the commutant. The first term is the energy of $\mathbf{R}$ that does not lie in the commutant of $G$; it is zero when $\mathbf{R}$ respects the symmetry $G$ and otherwise is a floor that no amount of data can lower. The second term is the sampling fluctuation surviving the projection. For a matched signal, the surface is pure variance and falls in both coordinates; for a mismatched signal, the bias term lifts part of the surface above the floor where data cannot reach.

\subsection{The matched surface: a substitution law between data and structure}

When $\mathbf{R}$ lies in the commutant of every group in a subgroup chain, the bias term in \eqref{eq:gl-bv} vanishes and the residue is governed entirely by the variance term. Taking $\mathbf{R}$ in the commutant and isotropic sampling, the surviving fluctuation occupies the commutant subspace, so
\begin{equation}
\label{eq:gl-var}
\mathcal{E}(|G|, L) \;\approx\; \frac{c}{d_{\mathrm{eff}}\, L},
\qquad
d_{\mathrm{eff}} \;=\; \frac{M^2}{\dim \mathcal{C}(\rho)} ,
\end{equation}
with $d_{\mathrm{eff}}$ the variance-reduction factor of the group-averaged estimator, that is, the group gain expressed as a dimension. The variance therefore depends on $|G|$ and $L$ only through the product $d_{\mathrm{eff}}\, L$: in the matched regime, structure and data are substitutes, and the level sets of the residue surface are hyperbolas in the $(d_{\mathrm{eff}}, L)$ plane. The left panel of Figure~\ref{fig:gl-surface} shows this for a circulant signal with full cyclic symmetry. The contours are the predicted hyperbolas, and the product $\mathcal{E} \cdot d_{\mathrm{eff}} L$ is approximately constant across the grid.

The substitution is exact in the white-signal limit. For a scalar matched signal, $\mathbf{R} \propto \mathbf{I}$, the sampling fluctuation is unitarily invariant and \eqref{eq:gl-var} holds with equality, $d_{\mathrm{eff}} = M^2/\dim \mathcal{C}(\rho)$; the residue then depends only on the product $d_{\mathrm{eff}} L$ and the iso-contours are exact hyperbolas. A numerical sweep confirms this: the coefficient of variation of $\mathcal{E} \cdot d_{\mathrm{eff}} L$ across the entire grid is $0.005$ for a white matched signal, and along a fixed product $d_{\mathrm{eff}} L = 16$ the residue is flat to within one percent as snapshots are traded for group order. For a colored matched signal the two routes are no longer perfectly interchangeable. The sampling fluctuation inherits the anisotropy of $\mathbf{R}$, so the commutant directions retained by a larger group carry, on average, a different share of the fluctuation than independent snapshots remove; along a fixed product the data route slightly outperforms the structure route. For the circulant signal of Figure~\ref{fig:gl-surface} this deviation is modest but real, persisting at high trial count, with the product coefficient of variation rising to $0.35$. The exchange rate between the two routes is therefore set by $d_{\mathrm{eff}}$, exactly for white signals and with an $\mathbf{R}$-dependent correction otherwise.

\subsection{The mismatched surface: the valley and the sign of the structure slope}

When $\mathbf{R}$ respects only a proper subgroup, increasing the group order past the matched order forces invariances the signal does not have. The bias term in \eqref{eq:gl-bv} then grows with $|G|$, and the slope of the surface in the structure coordinate changes sign. The center panel of Figure~\ref{fig:gl-surface} shows a signal whose true symmetry is $\mathbb{Z}_8$. Residue falls in $|G|$ only up to the matched order, marked by the dashed line, and rises above it; the minimum-residue locus is a valley along the matched order rather than the corner of largest group. The right panel isolates the mechanism: at and below the matched order the residue rides the $\mathcal{E} \propto 1/L$ line of pure variance, while above the matched order it flattens to a nonzero bias floor and, at large $L$, lies above the matched curve. Past the matched order, more data buys nothing and the correct response is a different group, not a larger one.

\begin{figure*}[t]
\centering
\includegraphics[width=\textwidth]{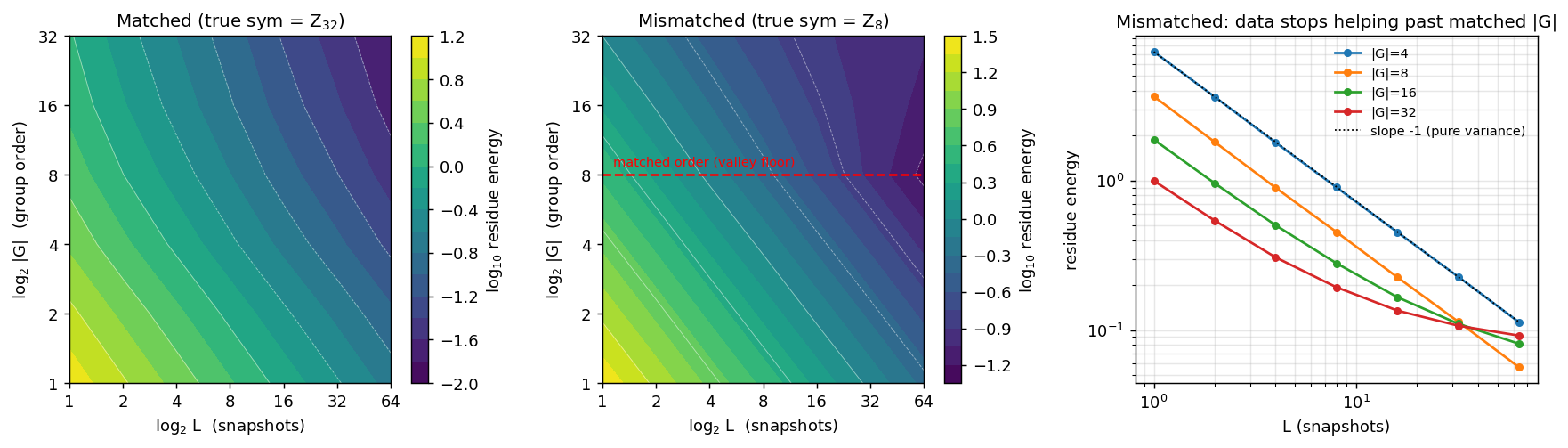}
\caption{Residue energy over group order and snapshot count, $M = 32$. \textbf{Left:} matched signal (full cyclic symmetry); contours are hyperbolas, so structure and data substitute. \textbf{Center:} mismatched signal (true symmetry $\mathbb{Z}_8$); the residue valley lies along the matched order (dashed line) and the structure slope turns positive above it. \textbf{Right:} for the mismatched signal, group orders at or below the matched order follow the pure-variance slope $-1$, while larger group orders flatten to bias floors that more data cannot lower.}
\label{fig:gl-surface}
\end{figure*}

\subsection{A ground-truth-free diagnostic}

Equations \eqref{eq:gl-bv} and \eqref{eq:gl-var} suggest a procedure that needs no knowledge of the true covariance. At a fixed candidate group, measure the residue at several snapshot counts and fit $\mathcal{E}(L) \approx B + c/L$. The intercept $B$ estimates the structural bias floor, the wall that more data cannot breach, and the slope term $c/L$ estimates the data-limited part. Their ratio at the operating point, $(c/L)/B$, is a single actionable number: when it is large the estimate is data-limited and the right move is to collect more snapshots; when it is small the estimate is structure-limited and the right move is to change the group. Applied to the mismatched signal of Figure~\ref{fig:gl-surface}, the fit returns $B \approx 0$ for every group order at or below the matched order and a positive, growing $B$ above it, and the ratio correctly flags the over-large groups as structure-limited while reading everything at or below the matched order as data-limited. Sweeping the group order at the largest affordable $L$ and locating the minimum recovers the matched order itself, and the upturn beyond it is the observable signal that the symmetry has been over-imposed.

The decomposition has a model-selection reading. The intercept $B$ is the approximation (bias) floor and $c/L$ the estimation (variance) term, so the fit is a structured-estimation analogue of the two-part code of minimum description length \cite{rissanen1978} and of the learning curve assessed by cross-validation \cite{stone1974}. What differs is the role of model complexity: a richer matched group lowers the variance term while adding no bias at all until the true symmetry order is crossed, so within the matched regime structure behaves as a variance-reduction resource interchangeable with data rather than as the usual bias-variance dial.

\subsection{The correct axis is effective dimension, not group order}

Equation \eqref{eq:gl-var} identifies the variance-reduction factor as $d_{\mathrm{eff}} = M^2/\dim \mathcal{C}(\rho)$ rather than the group order $|G|$. For an Abelian group in its regular representation the two coincide, so a cyclic study cannot separate them. To separate them, consider non-Abelian groups and representations in which the group order, the sum of squared irreducible dimensions, and the commutant dimension are deliberately decorrelated. For a representation with isotypic decomposition into irreducibles of dimension $d_i$ and multiplicity $m_i$, one has $M = \sum_i m_i d_i$ and $\dim \mathcal{C}(\rho) = \sum_i m_i^2$, so $d_{\mathrm{eff}} = M^2 / \sum_i m_i^2$. This reduces to $|G|$ for the Abelian regular representation and to $\sum_i d_i^2$ for the regular representation and for any single-isotypic representation, recovering the framework's effective-dimension formula of Section~\ref{sec:fourthms} in those canonical cases, but it is the more general invariant.

Table~\ref{tab:gl-nonab} and Figure~\ref{fig:gl-nonab} report the measured group gain, the ratio $\|\hat{\mathbf{R}} - \mathbf{I}\|_F^2 / \|\hat{\mathbf{R}}_G - \mathbf{I}\|_F^2$ for a matched signal $\mathbf{R} = \mathbf{I}$, against the three candidate predictors. The measured gain tracks $M^2/\dim \mathcal{C}(\rho)$ across every case. Two cases are decisive. The two-dimensional irreducible representation of the quaternion group $Q_8$ has gain near $4$, the square of the irreducible dimension, not its group order $8$, so the axis is not $|G|$. The unbalanced representation $\mathrm{triv} \oplus (\text{2-dim})$ of $S_3$ has gain near $4.5 = M^2/\dim \mathcal{C}(\rho)$, not $\sum_i d_i^2 = 5$, so the operative invariant is the commutant dimension and the sum-of-squares formula is its special case. The single-isotypic representations of $S_3$ at multiplicities $1, 2, 3$ all give gain near $4$ independent of multiplicity, confirming that the gain depends on the representation type and not on how many copies are stacked.

\begin{table*}[t]
\centering
\caption{Measured group gain for a matched white signal ($\mathbf{R} = \mathbf{I}$) against three candidate predictors, $L = 64$, $6000$ Monte Carlo trials. The commutant dimension is computed exactly from the characters, $\dim \mathcal{C} = \frac{1}{|G|}\sum_g |\chi(g)|^2$, and $M^2/\dim \mathcal{C}$ is the exact predicted gain. The measured gain matches it to within one percent in every case; it equals $\sum_i d_i^2$ only for the regular and single-isotypic representations, and it never equals $|G|$ except where the two coincide.}
\label{tab:gl-nonab}
\small
\renewcommand{\arraystretch}{1.15}
\begin{tabular}{lrrrrr}
\toprule
representation & $M$ & $|G|$ & $\sum_i d_i^2$ & $M^2/\dim \mathcal{C}$ & measured \\
\midrule
$S_3$ 2-dim, $m=1$            & 2 & 6 & 4 & 4.0 & 4.04 \\
$S_3$ 2-dim, $m=2$            & 4 & 6 & 4 & 4.0 & 4.02 \\
$S_3$ 2-dim, $m=3$            & 6 & 6 & 4 & 4.0 & 3.99 \\
$S_3$ regular                & 6 & 6 & 6 & 6.0 & 6.00 \\
$S_3$ $\mathrm{triv}\oplus$2-dim & 3 & 6 & 5 & 4.5 & 4.50 \\
$Q_8$ 2-dim, $m=1$           & 2 & 8 & 4 & 4.0 & 4.00 \\
$Q_8$ 2-dim, $m=2$           & 4 & 8 & 4 & 4.0 & 4.02 \\
$\mathbb{Z}_8$ regular        & 8 & 8 & 8 & 8.0 & 8.06 \\
\bottomrule
\end{tabular}
\end{table*}

\begin{figure}[t]
\centering
\includegraphics[width=0.95\linewidth]{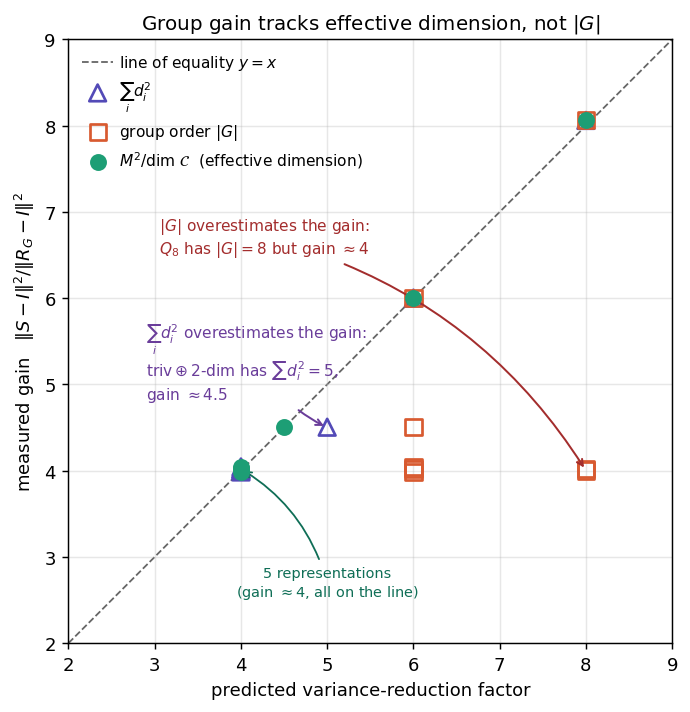}
\caption{Measured group gain against three candidate predictors. Only $M^2/\dim \mathcal{C}(\rho)$, the effective dimension, lies on the line of equality for all representations. The group order and the sum of squared irreducible dimensions both fall off the line in the decisive cases ($Q_8$ for $|G|$, the unbalanced $S_3$ representation for $\sum_i d_i^2$).}
\label{fig:gl-nonab}
\end{figure}

\subsection{Practitioner recipe}

The surface yields a compact rule. Replace the group-order axis by the effective dimension $d_{\mathrm{eff}} = M^2/\dim \mathcal{C}(\rho)$. In the matched regime, structure and data substitute along $d_{\mathrm{eff}}\, L$, exactly for a white signal and with an $\mathbf{R}$-dependent correction otherwise. To act without ground truth, sweep $L$ at a fixed candidate group, fit $B + c/L$, and read the data-versus-structure ratio; sweep $d_{\mathrm{eff}}$ at the largest affordable $L$ to locate the matched order as the valley floor and to detect over-imposed symmetry as the upturn beyond it. The single message is that more data and a richer group are interchangeable only up to the structural floor set by the signal's true symmetry, and the floor, not the group order, is what a diagnostic must find.

\section{Erlangen Reading: Three Information-Theoretic Regimes}
\label{sec:erlangen}

The four theorems of Section~\ref{sec:fourthms} admit a higher-level interpretation in the style of Klein's Erlangen program~\cite{klein1872}. The reading is an interpretation, not itself a theorem.

\subsection{Klein's program and the analogue here}

In Klein's Erlangen program, a geometry is defined by the group that preserves its structure: Euclidean geometry by the rigid-motion group, projective geometry by the projective group, affine geometry by the affine group. The properties studied within each geometry are exactly those invariant under the defining group. The analogue suggested by the present material is that the information extractable from an observation is determined by the symmetry group brought to bear on it: the covariance imposes a bilinear structure on the observation space, and the group-averaged estimator imposes a representation-theoretic structure on top of it. Different groups reveal different aspects of the same observation, and the symmetric group via the Cayley graph construction reveals the most, recovering the full KL spectrum.

\subsection{Shannon information content (R\'{e}nyi-1)}

Shannon's theory~\cite{shannon1948} measures information content via the entropy of order 1,
\begin{equation}\label{eq:shannon_ent}
H_1(p) \;=\; -\sum_k p_k \log p_k,
\end{equation}
a functional of the probability distribution of a source. It is observer-independent: it depends neither on the instrument used nor on any extraction mechanism. Shannon's strength is precisely this universality; his boundary is the same property. The entropy measures how many bits a source contains, not how many can be extracted from one observation. The classical answer to the extraction question is the law of large numbers: to extract the $H_1$ bits from a noisy channel, average many samples. This is the trivial-group case ($\kappa = 1$) of Theorem~\ref{thm:source} (structural source coding), recovered explicitly in Section~\ref{sec:fourthms}.

\subsection{AD information structure (R\'{e}nyi-2)}

The structural capacity $\kappa$ of equation~\eqref{eq:kappa_spec} is a functional of the Rényi-2 entropy of the eigenvalue spectrum of $\mathbf{R}$, in direct analogy to Shannon's content as a functional of Rényi-1. It counts the effective number of spectral modes over which the energy is spread, which by Theorem~\ref{thm:outer} is the per-entry variance reduction that a matched group delivers from a single observation. For a rank-one signal, $\kappa = 1$ (the floor: a concentrated spectrum admits the trivial-group baseline only and no additional exploitable degeneracy). For white noise on $\mathbb{C}^M$, $\kappa = M$ (the ceiling: all $M$ spectral modes are equally weighted, giving the maximum participation ratio and the maximum exploitable degeneracy). For a rank-one signal in noise, $\kappa$ decreases from $M$ toward $1$ as SNR grows and the covariance becomes increasingly rank-one-dominated. As established in Section~\ref{sec:fourthms}, $\kappa$ is the operational fundamental measure of single-observation estimation efficiency.

A sinusoid in silence sits at the floor of both measures: its spectrum has a single mode, so its Shannon entropy is zero and $\kappa = 1$, and the matched transform isolates that mode with no averaging gain available or needed. White noise sits at the ceiling of both: every mode is equally present, the Shannon entropy is maximal, and $\kappa = M$, so the matched group average delivers its full $M$-fold variance reduction although no mode is distinguished from any other. How strongly the spectrum concentrates on its dominant mode is read separately by the spectral concentration $\hat{\psi}$ of Section~\ref{sec:psi}. Shannon theory is the limit for data communication; structural capacity is the limit for single-observation estimation. They are not competitors; they measure different aspects of the same reality.

\subsection{Von Neumann entropy and the quantum measurement setting}
\label{sec:vn}

A third regime is the quantum measurement setting, governed by the \emph{von Neumann entropy}~\cite{vonneumann1932}
\begin{equation}\label{eq:vn}
S(\rho) \;=\; -\Tr(\rho \log \rho) \;=\; -\sum_k \lambda_k(\rho) \log \lambda_k(\rho),
\end{equation}
where $\rho$ is a density matrix and $\{\lambda_k(\rho)\}$ are its eigenvalues. The von Neumann entropy is Shannon's entropy applied to the spectrum of the density matrix: it measures information content in the quantum setting in the Shannon sense.

The quantum measurement problem is that a density matrix cannot be measured twice: the no-cloning theorem and the measurement postulate together forbid repeated access to a single copy of an unknown state. Classical quantum state tomography requires the preparation of many identical copies, which is the quantum analog of the trivial-group approach of averaging many independent classical observations: the Shannon-style $H_1$ content of the state is accessible through $L$ independent copies via the law of large numbers applied to measurement outcomes.

Algebraic diversity in the quantum setting~\cite{thornton2026quantum} replaces the many-copy protocol with a single-copy protocol that applies a group-structured set of measurement bases to one copy, then averages the rank-1 outcomes over the group. The result is a full-rank density-matrix estimator from one copy, at the cost of requiring the measurement bases themselves to be group-structured. The information-theoretic interpretation is that quantum AD extracts the structural capacity of $\rho$, $\kappa(\rho)$, from one copy, while the conventional multi-copy tomography extracts the von Neumann content $S(\rho)$ at asymptotic rate $1/\sqrt{L}$. The two quantities are distinct. The von Neumann entropy measures the informational content of the state for purposes of channel capacity and data transmission; the structural capacity measures the single-copy extractable structure of the state. Both are functionals of the same eigenvalue distribution through different R\'{e}nyi orders.

\subsection{Three regimes, one principle}

The three information measures, Shannon's $H_1$, the AD structural capacity (a functional of $H_2$), and the von Neumann $S$, correspond to three complementary perspectives on information. Shannon and von Neumann measure \emph{content} (classical and quantum respectively); AD measures \emph{structure}. The structural capacity is the R\'{e}nyi-2 analogue of the Shannon and von Neumann R\'{e}nyi-1 entropies, and it is the natural information measure for single-observation algebraic estimation. Whether this R\'{e}nyi-order separation ultimately proves productive beyond the concrete results reported here and in~\cite{thornton2026ad, thornton2026quantum} is for subsequent work to establish.

\section{Related Work}
\label{sec:related}

Four streams of prior work have substantial conceptual overlap with the algebraic diversity framework, and a proper assessment of what AD contributes requires situating it against each of them. We treat the four streams in historical order: classical invariant estimation and group-invariant hypothesis testing as developed from the mid-twentieth century onward, minimax and convex-optimization approaches to robust estimation developed primarily over the past two decades, algebraic signal processing developed by P\"uschel, Moura, and their co-authors starting in the mid-2000s together with its graph signal processing descendants, and compressed sensing developed by Cand\`{e}s, Romberg, Tao, and Donoho in the mid-2000s. In each case we identify the overlap with AD, state what AD adds that the prior framework did not, and note where AD and the prior framework can be composed rather than competed. No claim of subsumption is made in either direction; the goal is to clarify relationships rather than to rank frameworks.

\subsection{Invariant estimation and hypothesis testing}
\label{sec:invariant}

The use of group-invariance as a principled restriction on statistical procedures has a long history, running through the invariant hypothesis testing framework of Lehmann and subsequent authors~\cite{lehmann_romano2005} and entering signal processing through the detection-and-estimation textbook tradition represented by Poor~\cite{poor1994}. In this framework, a decision problem or estimation problem is said to be \emph{invariant} under a group $\mathcal{G}$ of transformations if the problem itself is unchanged under the action of $\mathcal{G}$ on the observation space, in the sense that the hypothesis structure, the parameter space, and the loss function all commute with the group action. Under such a symmetry, the principle of invariance asserts that a statistical procedure should also be invariant under the same group, and this restriction to invariant procedures gives rise to a sharp theory of optimality within the invariant class: the uniformly most powerful invariant (UMPI) test, the equivariant Bayes estimator, and the constant-false-alarm-rate (CFAR) detector are all consequences of the invariance principle applied in the appropriate setting. For signal processing applications, the framework has been particularly productive in the design of detectors whose false alarm rate does not depend on unknown nuisance parameters, and in the derivation of maximal-invariant statistics whose distribution is free of the transformation group under which the problem is invariant.

The relationship to the algebraic diversity framework is substantive, and the distinction between the two settings is worth stating precisely. In the classical invariant-estimation framework, the group under which the problem is invariant is part of the problem's specification: it is known in advance and enters as a structural constraint on the space of admissible procedures. The AD framework, in contrast, treats the signal's algebraic symmetry as data: the matched group is identified from a single observation (or a small number of observations) by the blind group-matching methodology of Section~\ref{sec:blindmatching}, and the resulting group is then used to construct a group-averaged estimator that is invariant under the identified symmetry. In this sense AD is the data-driven counterpart of the classical invariant estimation framework. Where classical invariance analysis asks ``given the symmetry group of this problem, what is the best procedure within the class of procedures respecting that symmetry?'', AD asks ``given this observation, what symmetry group does the signal carry, and what variance reduction does the corresponding group-averaged estimator provide?''. The two viewpoints compose cleanly: AD supplies the identification of the group from data, and classical invariant-estimation theory supplies the optimality framework once the group is known.

A second distinction is in the role the group plays. In classical invariant hypothesis testing, the group acts on the problem setup (the null and alternative distributions) and the practitioner's interest is in procedures whose operating characteristics are free of nuisance parameters absorbed by the group. In AD, the group acts on the signal itself through a unitary representation, and the practitioner's interest is in the variance reduction delivered by averaging a statistic over the orbit of that action. Nothing in the invariant-estimation literature is inconsistent with AD, and the two viewpoints can operate on the same problem in complementary modes: classical invariance shapes the procedure's nuisance-parameter robustness, while AD supplies the single-observation variance reduction through group orbit averaging.

\subsection{Minimax and convex-optimization approaches to robust estimation}
\label{sec:minimax}

A distinct modern tradition in signal processing and estimation, developed extensively by Eldar and co-authors over the past two decades, treats estimator design as a convex optimization problem under uncertainty~\cite{eldar_merhav2004, eldar_bental2005, palomar_eldar2010}. In the minimax formulation, the estimator is designed to minimize the worst-case mean-squared error across a specified uncertainty set in the signal covariance, the observation matrix, or the noise statistics; when the uncertainty set is represented by linear matrix inequalities, the resulting problem is a semidefinite program that can be solved efficiently by interior-point methods. A related line of work considers minimax \emph{regret} formulations~\cite{eldar_merhav2004}, in which the estimator minimizes the worst-case gap between its own mean-squared error and the mean-squared error of the best estimator that would have known the uncertain quantity, again under a specified uncertainty set. The framework applies broadly across signal processing and communications, and has produced substantive advances in robust beamforming, radar waveform design, and signal recovery under model uncertainty. A comprehensive treatment of the modern convex-optimization approach to signal processing is given in the Palomar-Eldar edited volume~\cite{palomar_eldar2010}.

The algebraic diversity framework shares with this tradition the use of a tractable optimization to drive estimator design, but the two frameworks approach the design problem from orthogonal directions. The minimax-and-convex tradition proceeds from a specified uncertainty set on the signal or the channel and designs the estimator to be robust across that set, producing estimators whose form often resembles a regularized least-squares or a generalized-inverse solution with a specific weighting. The algebraic diversity framework proceeds from an identified algebraic symmetry of the signal and designs the estimator to exploit that symmetry through group orbit averaging, producing estimators of the form given in equation~\eqref{eq:gae}. The specific optimization that arises in AD, namely the continuous relaxation via the double-commutator generalized eigenvalue problem on the Lie algebra $\mathfrak{u}(M)$ of Section~\ref{sec:dadcad}, is a polynomial-time eigenvalue problem on a fixed-dimensional vector space rather than a semidefinite program on a cone of positive matrices, and it returns a distinguished direction in a Lie algebra rather than a robust-optimal estimator across an uncertainty set. In principle the two frameworks can be composed: a minimax-regret robust estimator could be computed downstream of an AD-identified matched group, thereby combining the symmetry-exploitation benefit of AD with the uncertainty-robustness benefit of the minimax formulation. Whether such a composition delivers practical advantage in any particular application is a matter for empirical investigation and is not addressed here.

\subsection{Algebraic signal processing}
\label{sec:asp}

Algebraic signal processing (ASP), developed by P\"uschel and Moura in a foundational series of papers~\cite{puschel2008, puschel2008space} and extended by Sandryhaila and co-authors in subsequent work~\cite{sandryhaila2012nn, sandryhaila2013graphs, sandryhaila2014frequency}, assigns to each signal model a polynomial algebra on a finite group, from which the appropriate transform is derived by representation theory. The correspondence at the center of the framework is as follows. The cyclic group gives shift-invariant discrete-time signals with Fourier transforms; symmetric boundary conditions give the DCT, which in the group-theoretic reading of~\cite{thornton2026unification} is the dihedral group $D_{2M}$ acting on the even-reflected extension; the symmetric group gives graph signals on complete graphs; compound groups give compound models. The original two P\"uschel-Moura papers~\cite{puschel2008, puschel2008space} cover 1-D time (discrete-time signals with directional shift) and 1-D space (signals on a line with undirected shift and symmetric boundary conditions, a setting that covers the 16 distinct discrete trigonometric transforms including all DCT and DST variants). Subsequent work extends the framework in two directions of particular relevance here. Sandryhaila, Kova\v{c}evi\'{c}, and P\"uschel~\cite{sandryhaila2012nn} develop 1-D nearest-neighbor models in which the shift is symmetric between adjacent indices, yielding a framework in which orthogonal polynomial bases on the real line arise naturally as the corresponding Fourier transforms; Sandryhaila and Moura~\cite{sandryhaila2013graphs, sandryhaila2014frequency} then extend the shift-based framework to signals on arbitrary graphs, establishing the graph signal processing field by taking the weighted adjacency matrix of the graph as the elementary shift operator. The resulting graph signal processing framework has become a substantial subfield in its own right, with its own frequency notion, its own filters, and its own sampling theory.

The relationship between ASP and AD is complementary, and it is worth stating the division of labor precisely. ASP starts from a \emph{known} signal model, specified as a polynomial algebra together with an associated module and shift, and constructs the matched transform by the representation theory of the corresponding algebra. AD supplies two ingredients that ASP as formulated does not directly address: a single-observation estimator whose variance is reduced by $d_{\mathrm{eff}}(G, f)$ relative to the trivial-group case, and a model-selection criterion that identifies the group (and thence the ASP transform) from data when the model is not known in advance. We expand on each of these contributions.

The first contribution is the estimator. Given a single observation $\mathbf{x}$ and a hypothesized matched group $G$, the group-averaged estimator of equation~\eqref{eq:gae} produces a full-rank spectral estimate whose variance decreases by a factor of $d_{\mathrm{eff}}(G, f)$ relative to the sample-covariance estimator of the same observation. ASP, which is fundamentally a transform-theoretic framework organized around the algebraic structure of shifts and filters, does not directly supply a variance-reducing estimator at the single-observation level; the transform is computed, and the estimation problem is then passed to whatever statistical procedure the practitioner selects.

The second contribution is the model-selection criterion. When the signal model is not known in advance, the blind group-matching methodology of Section~\ref{sec:blindmatching} selects the group (and thence the ASP transform matched to that group) from data by evaluating the cross-validation residual $D_{CV}$ across a library of candidate groups, or, via the continuous relaxation of Section~\ref{sec:dadcad}, by a polynomial-time eigenvalue computation on the Lie algebra $\mathfrak{u}(M)$. ASP as formulated assumes the model is given, and passes the model-selection question to the designer; AD supplies the data-driven selection step that closes this loop.

Taken together, ASP and the AD viewpoint form a natural two-step pipeline. Blind group matching identifies the model from observations, and ASP assigns the transform matched to the identified model. The pipeline inherits ASP's full representation-theoretic machinery for the transform step, together with the AD diagnostic suite (coloring index $\alpha$, commutativity residual $\delta$, spectral coherence index $\chi$, structural capacity $\kappa$, cross-validation $D_{CV}$, effective group order $d_{\mathrm{eff}}$) for the identification step. The graph-automorphism theorem of Section~\ref{sec:dadcad} is particularly relevant to the ASP graph-signal-processing descendant~\cite{sandryhaila2013graphs}: the double-commutator residual is zero if and only if the permutation is a graph automorphism, which gives an exact algebraic oracle for symmetry detection on graph signals without requiring spectral-signature comparison across candidate automorphisms.

\subsection{Complete reconstruction: the structured part and the algebraic residue}
\label{sec:residue}

The development to this point has treated the matched group as a means of variance reduction in second-order estimation: the group-averaged estimator of equation~\eqref{eq:gae} concentrates the signal's energy into the irreducible subspaces of the matched group, and the resulting transform is the natural basis for analysis. A question left open by that development is what becomes of the part of the signal that the group does \emph{not} concentrate. Conventional signal processing answers the analogous question for the trivial group with the model $\mathbf{x} = \mathbf{s} + \mathbf{n}$, structure plus noise, but algebraic diversity has not so far said how a signal is reconstructed in full from its group-theoretic parts. This subsection closes that gap, and in doing so places algebraic diversity in exact parallel with the energy-and-structure picture of compressed sensing developed in the subsection that follows.

Fix a matched group $G$ with group-averaged covariance $\mathbf{R}_G$ from equation~\eqref{eq:gae}, and let $\mathbf{U}$ be the group-adapted Karhunen--Lo\`{e}ve basis that block-diagonalizes $\mathbf{R}_G$ into its irreducible components. Partition the coordinates of $\mathbf{U}$ into the \emph{coherent} subspace, the eigenvectors of $\mathbf{R}_G$ whose eigenvalues exceed the isotropic floor, and its orthogonal complement. Let $\mathbf{P}$ be the orthogonal projector onto the coherent subspace. Then any observation $\mathbf{x}$ decomposes as
\begin{equation}\label{eq:sym_plus_res}
\mathbf{x} \;=\; \mathbf{x}_{\mathrm{str}} + \mathbf{x}_{\mathrm{res}}, \qquad
\mathbf{x}_{\mathrm{str}} = \mathbf{P}\mathbf{x}, \quad
\mathbf{x}_{\mathrm{res}} = (\mathbf{I} - \mathbf{P})\mathbf{x},
\end{equation}
the \emph{structured part} carried by the coherent subspace of the matched group, and the \emph{algebraic residue} carried by its complement. The decomposition is exact and orthogonal by construction, $\mathbf{x}_{\mathrm{str}}^H \mathbf{x}_{\mathrm{res}} = 0$ and $\|\mathbf{x}\|^2 = \|\mathbf{x}_{\mathrm{str}}\|^2 + \|\mathbf{x}_{\mathrm{res}}\|^2$, so it loses no information and introduces no cross term. This is the algebraic-diversity counterpart of the conventional $\mathbf{s} + \mathbf{n}$ split, with the matched group, rather than an assumed noise model, determining which part is which.

\begin{theorem}[Complementary decomposition]\label{thm:complementary}
Let $G$ be a matched group with group-averaged covariance $\mathbf{R}_G$ and coherent projector $\mathbf{P}$. The decomposition~\eqref{eq:sym_plus_res} is the unique orthogonal decomposition of $\mathbf{x}$ into a component in the coherent subspace of $\mathbf{R}_G$ and a component in its complement. The structured part is the best $k$-term approximation of $\mathbf{x}$ in the group-adapted basis in the mean-square sense, where $k$ is the dimension of the coherent subspace, and the algebraic residue is the corresponding approximation error.
\end{theorem}

The structured part is, by the Eckart--Young theorem applied in the group-adapted basis, the most compressible representation of the signal that the matched group affords: its energy is concentrated into the $k$ coherent modes, and the participation ratio
\begin{equation}\label{eq:participation}
D_2 \;=\; \frac{(\Tr \mathbf{R}_G)^2}{\|\mathbf{R}_G\|_F^2}
\end{equation}
measures the effective number of those modes, with the structural capacity $\kappa = D_2$ of Section~\ref{sec:fourthms} as its information-structure form. The algebraic residue, by contrast, is the component on which the matched group acts isotropically: its group-averaged covariance is a scalar multiple of the identity, so no nontrivial group concentrates it, and it carries the maximum-entropy, flat-spectrum part of the signal. The residue is therefore not noise in the conventional sense but the structureless remainder relative to the chosen family of groups; what is residue under one matched group can become structured under a larger group that concentrates a component the smaller group leaves isotropic.

\begin{theorem}[Maximum-entropy characterization of the residue]\label{thm:maxent}
Among all decompositions of $\mathbf{x}$ into a $k$-dimensional structured component and a complementary residue with the energy budget fixed by~\eqref{eq:sym_plus_res}, the decomposition induced by the matched group's coherent projector maximizes the R\'{e}nyi-2 entropy of the residue's modal energy distribution. The residue's group-averaged covariance is isotropic, and its coefficients in the group-adapted basis follow the flat, maximum-entropy law.
\end{theorem}

These two theorems complete the reconstruction story: an arbitrary signal is recovered exactly as a structured part, which is the compressible head of the group-adapted spectrum and the portion that the variance-reduction machinery of the preceding sections exploits, plus an algebraic residue, which is the incompressible, isotropic floor. The structured part is what algebraic diversity can analyze, estimate, and compress; the residue is what it cannot, and the decomposition makes the boundary between them precise and, importantly, relative to the matched group rather than absolute. Proofs and a detailed empirical validation, including the exact orthogonality to machine precision, the energy-compaction knee at $D_2$, and the maximum-entropy coefficient law, are developed in a companion treatment; here the decomposition serves to connect the variance-reduction view of the preceding sections to the energy-structure view of compressed sensing that follows.

\subsection{Compressed sensing}
\label{sec:cs}

Compressed sensing (CS), developed by Cand\`{e}s, Romberg, and Tao~\cite{candes2006} and by Donoho~\cite{donoho2006}, recovers signals sparse in a dictionary $\boldsymbol{\Phi}$ from undersampled linear measurements under suitable incoherence conditions. CS is one of the central advances of the last two decades in signal processing and statistics, and has reorganized both fields. Applying CS in practice involves several interconnected challenges, among them the choice or learning of a suitable sparsifying dictionary, the design of the measurement operator, the choice of reconstruction algorithm, and robustness to model mismatch when real signals deviate from ideal sparsity. Each of these challenges has received substantial attention in the CS literature, and each remains fundamental in its own right.

Dictionary construction is the challenge to which AD contributes most directly, and we describe the contribution here in the form of an observation together with a proposed recipe; the observation is rigorous, but the recipe is at present conjectural and awaits experimental validation. The observation is the following. When a class of signals admits a common algebraic symmetry group $G$, the orbit of a seed vector $\mathbf{v}$ under the group action,
\begin{equation}\label{eq:orbit_dict}
\boldsymbol{\Phi}_{G, \mathbf{v}} \;=\; \{\pi_g \mathbf{v} : g \in G\},
\end{equation}
constitutes a dictionary whose atoms inherit the algebraic structure of $G$ by construction. The dictionary is specified by a small amount of information, namely the group $G$ together with a seed vector $\mathbf{v}$, rather than by the much larger training set that a general learned dictionary would require. The blind group-matching methodology of Section~\ref{sec:blindmatching} suggests a recipe for using this observation in practice: identify $G$ via $D_{CV}$ (or its continuous relaxation) from a small number of observations, estimate a seed vector $\mathbf{v}$ from the dominant eigenvector of the group-averaged estimator, and construct the orbit dictionary from the identified group-seed pair. Whether the AD viewpoint contributes usefully to the other CS challenges enumerated above (measurement design, reconstruction, and robustness to sparsity mismatch) is not addressed here and is a separate question for future work.

The proposed recipe is at present \emph{conjectural} in its empirical claim, though the structural foundation under it is now established rather than assumed. The complementary decomposition of Section~\ref{sec:residue} makes precise what the orbit dictionary can and cannot represent: the structured part $\mathbf{x}_{\mathrm{str}}$ is the compressible head of the group-adapted spectrum, sparsely representable in the orbit dictionary up to the effective dimension measured by the participation ratio $D_2$, while the algebraic residue $\mathbf{x}_{\mathrm{res}}$ is the incompressible, maximum-entropy floor that no orbit dictionary of the matched group concentrates. This is the algebraic counterpart of the sparse-plus-error decomposition that underlies compressed sensing: it identifies, before any reconstruction is attempted, the irreducible error floor against which any AD-CS recovery must be measured, and it does so as a theorem rather than an empirical observation. What remains conjectural is therefore narrower than before: not whether a signal decomposes into a structured part and a residue, which Section~\ref{sec:residue} establishes, but whether the orbit dictionary built from the blindly identified group yields a practical reconstruction advantage over an unstructured learned dictionary on sparse-recovery problems. That comparison has not been performed at the time of writing, and no theoretical result yet establishes the superiority of orbit dictionaries over learned dictionaries in reconstruction quality; the structural argument, that a dictionary determined by a finite group action is specified by $|G|$ parameters rather than the much larger number a general learned dictionary requires, is suggestive of reduced sample complexity in the regime where the orbit hypothesis holds, but the practical payoff remains a matter for empirical work. Figures~\ref{fig:adcs_dense}, \ref{fig:adcs_sparse}, and \ref{fig:adcs_mixed} illustrate the conjectural combined AD-CS recovery on three test signals of different structural character: dense-but-symmetric, sparse-but-asymmetric, and sparse-and-symmetric.

\begin{figure}[t]
\centering
\includegraphics[width=\columnwidth]{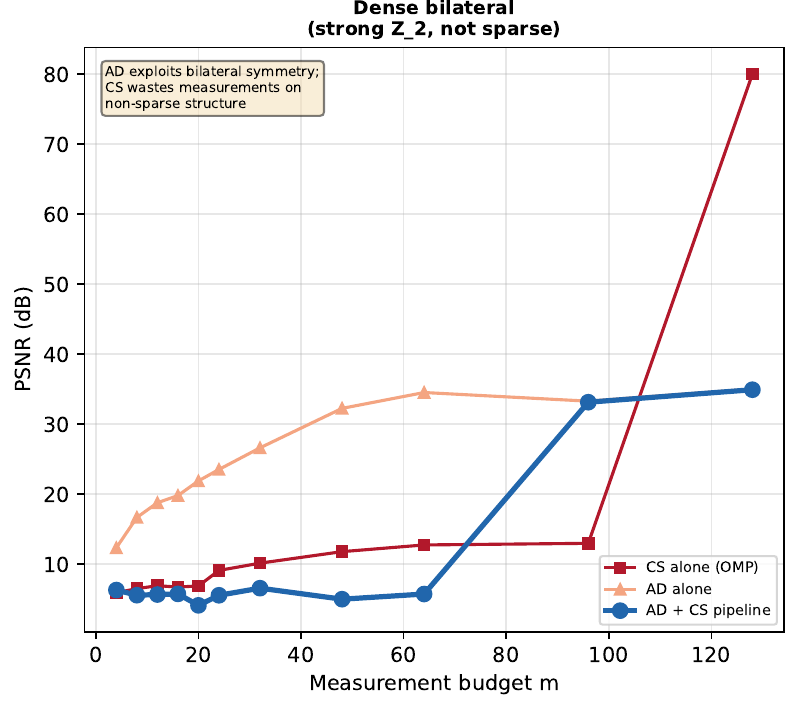}
\caption{Conjectural AD-CS recovery, dense bilateral signal at 25~dB SNR (M=128, strong $\mathbb{Z}_2$ symmetry, not sparse). AD exploits bilateral symmetry; CS wastes measurements on non-sparse structure; the AD+CS pipeline tracks the better of the two across the measurement budget.}
\label{fig:adcs_dense}
\end{figure}

\begin{figure}[t]
\centering
\includegraphics[width=\columnwidth]{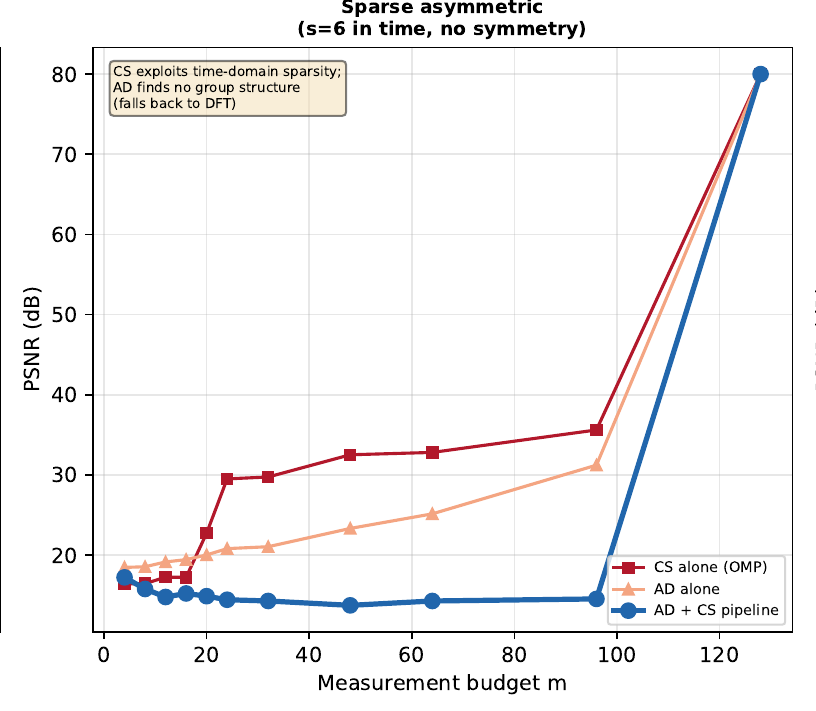}
\caption{Conjectural AD-CS recovery, sparse asymmetric signal at 25~dB SNR (M=128, $s = 6$ in time, no group symmetry). CS exploits time-domain sparsity; AD finds no group structure and falls back to the DFT; the AD+CS pipeline matches CS performance.}
\label{fig:adcs_sparse}
\end{figure}

\begin{figure}[t]
\centering
\includegraphics[width=\columnwidth]{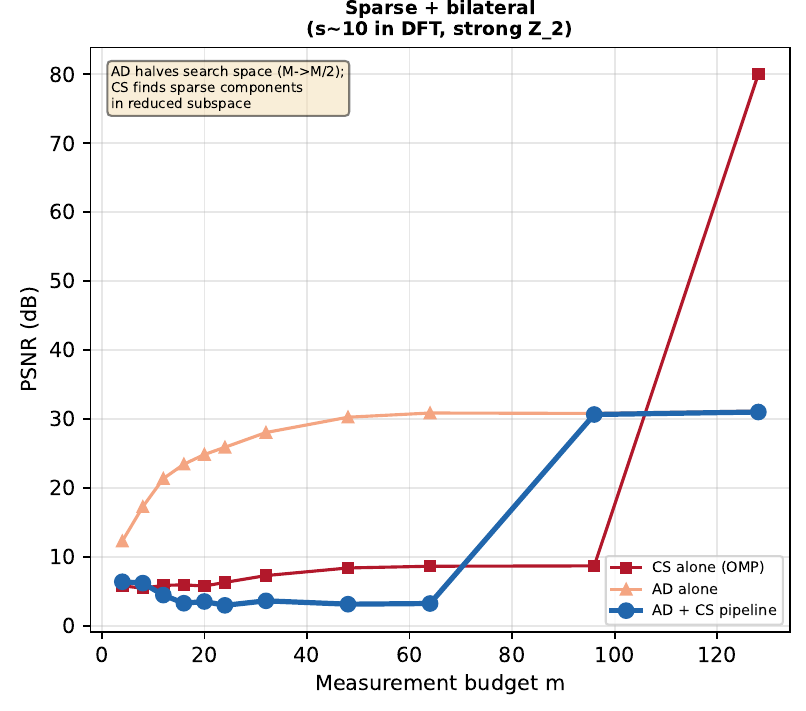}
\caption{Conjectural AD-CS recovery, sparse-and-bilateral signal at 25~dB SNR (M=128, $s \approx 10$ in DFT domain, strong $\mathbb{Z}_2$ symmetry). AD halves the search space ($M \to M/2$); CS finds the sparse components in the reduced subspace; the AD+CS pipeline obtains the best performance of the three settings.}
\label{fig:adcs_mixed}
\end{figure}

\section{Conclusion}
\label{sec:conclusion}

We have taken the algebraic diversity viewpoint developed for single-observation second-order measurement in~\cite{thornton2026ad} and broadened its reach in five directions.

\emph{Rank promotion} (Section~\ref{sec:rankpromote}) stratifies scalar data into tensorial form admitting nontrivial group action, allowing the $(G, L)$ continuum to apply outside the vector-observation setting; the blood-pressure example illustrates the mechanism, and the Trivial Group Embedding Theorem together with the General Algebraic Averaging Theorem (proved for the outer product, scalar symmetric statistics, and via Rao--Blackwell, conditional for higher-degree statistics pending the Clebsch--Gordan tensor lemma) establishes that the law of large numbers is the trivial-group instance of a broader algebraic averaging principle. The structural coding rate conjecture $n^* \approx \lceil 2^{H_{\mathrm{struct}}} \rceil$ provides the source-coding analog within this rank-promoted setting.

The \emph{eigentensor hierarchy} (Section~\ref{sec:eigentensor}) handles nested symmetry via wreath products and admits an $O(M \log M)$ tensor-FFT computation in the Abelian case; the abelianization trick extends this to non-Abelian groups with a quantified loss.

\emph{Blind group matching} (Section~\ref{sec:blindmatching}) evolves from an early spectral-concentration criterion (shown to have orbit-size bias) through a cross-validation criterion $D_{CV}$ to a polynomial-time continuous relaxation via a double-commutator generalized eigenvalue problem on the finite-dimensional Lie algebra $\mathfrak{u}(M)$. A closed-form commutator decomposition gives the eigenvalue-difference formula $\|[\mathbf{P}_\sigma, \mathbf{R}]\|_F^2 = \sum_k(\lambda_k - \lambda_{\sigma(k)})^2$ for permutations of the eigen-indices and an exact zero-residual test for graph automorphisms, and the Sequential GEVP with group-theoretic deflation produces a subgroup $G_K \subseteq \Aut(\mathbf{R})$ in at most $\lceil \log_2 |G_K| \rceil$ accepted iterations; full recovery $G_K = \Aut(\mathbf{R})$ is established when $\Aut(\mathbf{R}) = S_M$, the proper-subgroup case is recovered by a span-search form whenever the candidate basis spans a generating set (leaving only the a~priori basis-richness question open), and the degenerate-spectrum case is partially resolved, with the removable degeneracy handled by a null-space tiebreaker that selects generators by symmetry defect and the intrinsic case (distinct generating sets indistinguishable at second order) remaining open. The overall methodology is organized as library discovery, in which the double-commutator GEVP proposes candidate groups directly from data, followed by a two-tier selection that prefilters candidates by effective dimension and then scores the survivors for consistency. Signals carrying several symmetries at once, which a single group cannot match, are addressed by iterative stripping, an algebraic-diversity analog of the CLEAN algorithm that peels one structured component at a time and isolates the algebraic residue; this partially solves the mixed-structure case, separating components with distinct matched groups above a soft signal-to-noise-ratio floor while leaving the same-group and subspace-overlap cases open. The empirical reliability of the methodology is supported by a variance-scaling dichotomy that distinguishes matched (convergent power-law variance) from mismatched (constant, SNR-insensitive variance) groups across the AD diagnostic suite.

\emph{Blind signal processing} (Section~\ref{sec:blindproblems}) uses blind group matching as the bridge that extends the AD toolkit from measurement to blind and adaptive problems: the cost-symmetry matching principle predicts the residual ambiguity of a blind algorithm with a well-defined cost symmetry group $G_{\mathrm{cost}}$ as the coset space $G_{\mathrm{cost}} / G_{\mathrm{sig}}$, and the CMA residual phase standard deviation $45^\circ/\sqrt{3}$ is matched within $1.6^\circ$ on 3GPP TDL multipath. A further set of blind problems covered by the framework (BSS, carrier recovery, timing recovery, blind channel identification, blind dictionary discovery) is briefly indicated in Section~\ref{sec:conjectures}, with detailed treatments left to forthcoming companion papers.

\emph{Information structure theory} (Section~\ref{sec:fourthms}) formalizes the structural capacity $\kappa$ as the operational measure of single-observation estimation efficiency through four theorems paralleling Shannon's foundational results, complemented by a finite-dimensional conjugate capacity bound $\kappa_A \cdot \kappa_B \leq 4/c^2$ for non-commuting Hermitian generators on $\mathbb{C}^M$. The law of large numbers is recovered as the $\kappa = 1$ degenerate case.

An expanded Erlangen-program reading (Section~\ref{sec:erlangen}) distinguishes three information-theoretic regimes: Shannon's information content (R\'{e}nyi-1), the AD structural capacity (R\'{e}nyi-2), and the von Neumann entropy of quantum measurement (R\'{e}nyi-1 on the density-matrix spectrum). Content and structure are complementary functionals of the same eigenvalue distribution. The relationship to four streams of prior work (Section~\ref{sec:related}) is treated without claim of subsumption: AD is the data-driven counterpart of classical invariant estimation (Section~\ref{sec:invariant}) in which the matched group is identified from observations rather than supplied by the problem specification; AD is orthogonal to and composable with the minimax and convex-optimization tradition (Section~\ref{sec:minimax}) in which estimators are designed to be robust across uncertainty sets; AD supplies both a single-observation estimator and a data-driven model-selection step that algebraic signal processing (Section~\ref{sec:asp}) leaves to the designer; and AD contributes a complementary decomposition of a signal into a group-structured part and an algebraic residue, which grounds an orbit-based dictionary construction (Section~\ref{sec:cs}) whose structural foundation is established while its empirical advantage over learned dictionaries remains conjectural on the compressed-sensing side.

Open problems noted in the development include: a formal proof of the structural coding rate conjecture; a formal proof of the GAAT for general statistics beyond the outer product; a proof of the converse for the structural source coding theorem (that $\kappa$ equals the CRB universally rather than only in the established Gaussian outer-product case); efficient non-Abelian fast transforms generalizing the Abelian tensor-FFT; the uniqueness and order-independence of the iterative-stripping decomposition of mixed-structure signals when component subspaces overlap, and its extension to components sharing a matched group (the distinct-group case now being partially resolved, Section~\ref{sec:mixed}); characterization of which bases are rich enough for the span-search Sequential GEVP to recover $\Aut(\mathbf{R})$ in full when $\Aut(\mathbf{R})$ is a proper subgroup of $S_M$ (Conjecture~\ref{conj:practical}), the recovery procedure itself being established given a sufficiently rich basis; and the experimental validation of the AD orbit-dictionary conjecture for compressed sensing. Several of these admit concrete next-step experiments, and progress on them will determine how much further the AD viewpoint extends.


\end{document}